\documentclass{article}

\usepackage{amsfonts}
\usepackage{amsmath}
\usepackage{amssymb}
\usepackage{amsthm}
\usepackage{bbm}
\usepackage{bm}
\usepackage{boxedminipage}
\usepackage{booktabs}
\usepackage{caption}
\usepackage{courier} 
\usepackage{contour} 

\usepackage{extramarks} 
\usepackage{fancyhdr} 
\usepackage{float}
\usepackage{graphicx} 
\usepackage{ifthen}
\usepackage{lastpage} 
\usepackage{listings} 
\usepackage{mathtools}
\usepackage{multirow}
\usepackage{nicefrac}
\usepackage{subcaption}
\usepackage{tabu}
\usepackage{todonotes}

\usepackage{url}
\usepackage{xcolor}
\usepackage{xparse}
\usepackage{xspace}

\usepackage{textcase}
\usepackage{soul}
\usepackage{indentfirst}

\usepackage{pgffor}

\usepackage{pgfplots}
\pgfplotsset{compat=1.17}

\usepackage{tikz}
\usetikzlibrary{positioning}
\usetikzlibrary{fit}
\usetikzlibrary{calc}
\usetikzlibrary{backgrounds}
\usetikzlibrary{shapes}
\usetikzlibrary{patterns}
\usetikzlibrary{matrix}

\usepackage{geometry}
\usepackage{enumitem}
\usepackage{algorithm}
\usepackage[noend]{algpseudocode}

\usepackage[pdftex,bookmarks=true,pdfstartview=FitH,colorlinks,linkcolor=blue,filecolor=blue,citecolor=blue,urlcolor=blue,pagebackref=true]{hyperref}
\makeatletter
\newtheorem*{rep@theorem}{\rep@title}
\newcommand{\newreptheorem}[2]{
\newenvironment{rep#1}[1]{
 \def\rep@title{\theoremref{##1} Restated}
 \begin{rep@theorem}}
 {\end{rep@theorem}}}
\makeatother

\makeatletter
\newtheorem*{rep@lemma}{\rep@title}
\newcommand{\newreplemma}[2]{
\newenvironment{rep#1}[1]{
 \def\rep@title{\lemmaref{##1} Restated}
 \begin{rep@lemma}}
 {\end{rep@lemma}}}
\makeatother

\newenvironment{proofof}[1]{\begin{trivlist} \item {\bf Proof
#1:~~}}
  {\qed\end{trivlist}}

\newtheorem{theorem}{Theorem}[section]
\newreptheorem{theorem}{Theorem}
\newtheorem{importedtheorem}{Imported Theorem}

\newreplemma{lemma}{Lemma}
\newtheorem{corollary}[theorem]{Corollary}
\newtheorem{lemma}[theorem]{Lemma}

\newtheorem{definition}[theorem]{Definition}

\newtheorem{setting}[theorem]{Setting}

\newcommand{\namedref}[2]{\hyperref[#2]{#1~\ref*{#2}}\xspace}
\newcommand{\lemmaref}[1]{\namedref{Lemma}{lem:#1}}
\newcommand{\theoremref}[1]{\namedref{Theorem}{thm:#1}}

\newcommand{\corolref}[1]{\namedref{Corollary}{corol:#1}}
\newcommand{\figureref}[1]{\namedref{Figure}{fig:#1}}
\newcommand{\equationref}[1]{\namedref{Equation}{eq:#1}}
\newcommand{\algorithmref}[1]{\namedref{Algorithm}{alg:#1}}

\newcommand{\settingref}[1]{\namedref{Setting}{setting:#1}}
\newcommand{\importedtheoremref}[1]{\namedref{Imported Theorem}{impthm:#1}}

\newcommand{\sectionref}[1]{\namedref{Section}{sec:#1}}
\newcommand{\appendixref}[1]{\namedref{Appendix}{app:#1}}

\newcommand{\abs}[1]{\ensuremath{\left\vert{#1}\right\vert}\xspace}

\newcommand{\eps}[0]{\ensuremath{\varepsilon}}
\let\epsilon\eps

\newcommand{\cA}{\ensuremath{{\mathcal A}}\xspace}
\newcommand{\cB}{\ensuremath{{\mathcal B}}\xspace}

\newcommand{\cD}{\ensuremath{{\mathcal D}}\xspace}
\newcommand{\cE}{\ensuremath{{\mathcal E}}\xspace}
\newcommand{\cF}{\ensuremath{{\mathcal F}}\xspace}

\newcommand{\cI}{\ensuremath{{\mathcal I}}\xspace}

\newcommand{\cN}{\ensuremath{{\mathcal N}}\xspace}

\newcommand{\cP}{\ensuremath{{\mathcal P}}\xspace}
\newcommand{\cQ}{\ensuremath{{\mathcal Q}}\xspace}

\newcommand{\cS}{\ensuremath{{\mathcal S}}\xspace}

\newcommand{\cV}{\ensuremath{{\mathcal V}}\xspace}
\newcommand{\cW}{\ensuremath{{\mathcal W}}\xspace}

\newcommand{\bs}{\ensuremath{{\mathbf s}}\xspace}

\newcommand{\bx}{\ensuremath{{\mathbf x}}\xspace}
\newcommand{\by}{\ensuremath{{\mathbf y}}\xspace}
\newcommand{\bz}{\ensuremath{{\mathbf z}}\xspace}

\newcommand{\bA}{\ensuremath{{\mathbf A}}\xspace}

\newcommand{\bW}{\ensuremath{{\mathbf W}}\xspace}

\newcommand{\bbE}{\ensuremath{{\mathbb E}}\xspace}

\newcommand{\bbR}{\ensuremath{{\mathbb R}}\xspace}

\newcommand{\defeq}[0]{\ensuremath{\;{\vcentcolon=}\;}\xspace}

\newcommand{\E}[0]{\mathop{\bbE}\xspace}

\newcommand{\iid}[0]{\text{i.i.d.}\xspace}

\newcommand{\mat}[1]{\boldsymbol{#1}}
\renewcommand{\vec}[1]{\boldsymbol{\mathrm{#1}}}
\newcommand{\vecalt}[1]{\boldsymbol{#1}}

\newcommand{\normof}[1]{\|#1\|}

\newcommand{\bmat}[1]{\begin{bmatrix} #1 \end{bmatrix}}

\newcommand{\frakS}{\mathfrak{S}}

\newcommand{\mA}{\ensuremath{\mat{A}}\xspace}
\newcommand{\mB}{\ensuremath{\mat{B}}\xspace}

\newcommand{\mR}{\ensuremath{\mat{R}}\xspace}
\newcommand{\mS}{\ensuremath{\mat{S}}\xspace}

\newcommand{\mU}{\ensuremath{\mat{U}}\xspace}

\newcommand{\mW}{\ensuremath{\mat{W}}\xspace}

\newcommand{\va}{\ensuremath{\vec{a}}\xspace}
\newcommand{\vb}{\ensuremath{\vec{b}}\xspace}

\newcommand{\vd}{\ensuremath{\vec{d}}\xspace}
\newcommand{\ve}{\ensuremath{\vec{e}}\xspace}
\newcommand{\vf}{\ensuremath{\vec{f}}\xspace}

\newcommand{\vp}{\ensuremath{\vec{p}}\xspace}

\newcommand{\vr}{\ensuremath{\vec{r}}\xspace}

\newcommand{\vu}{\ensuremath{\vec{u}}\xspace}
\newcommand{\vv}{\ensuremath{\vec{v}}\xspace}

\newcommand{\vx}{\ensuremath{\vec{x}}\xspace}
\newcommand{\vy}{\ensuremath{\vec{y}}\xspace}
\newcommand{\vz}{\ensuremath{\vec{z}}\xspace}

\newcommand{\vDelta}{\ensuremath{\vecalt{\Delta}}\xspace}

\sodef\allcapsspacing{\upshape}{0.15em}{0.65em}{0.6em}

\colorlet{todo_background_normal}{white}
\definecolor{todo_background_dark}{RGB}{39,40,34}

\definecolor{advice_text}{RGB}{78, 12, 123}
\colorlet{advice_background}{todo_background_normal}

\definecolor{incomplete_text}{RGB}{204, 64, 84}
\colorlet{incomplete_background}{todo_background_normal}

\newcounter{question}
\DeclareMathOperator*{\argmin}{argmin}

\makeatletter
\setkeys{Gin}{width=\ifdim\Gin@nat@width>\linewidth
    0.5\linewidth
\else
    \Gin@nat@width
\fi}
\makeatother

\newcommand\numberthis{\addtocounter{equation}{1}\tag{\theequation}}

\newcommand{\tsfrac}[2]{{\textstyle\frac{#1}{#2}}}

\DeclareMathOperator*{\argmax}{argmax}
\DeclareMathOperator*{\polylog}{polylog}
\DeclareMathOperator*{\poly}{poly}

\DeclareMathOperator*{\sign}{sign}

\newcommand{\PPr}[1]{\ensuremath{\mathbf{Pr}\left[#1\right]}}

\newcommand{\Ex}[1]{\ensuremath{\mathbb{E}\left[#1\right]}}
\newcommand{\Var}{\text{Var}}

\makeatletter
\renewcommand*{\@fnsymbol}[1]{\textcolor{darkpastelgreen}{\ensuremath{\ifcase#1\or *\or \dagger\or \ddagger\or
 \mathsection\or \triangledown\or \mathparagraph\or \|\or **\or \dagger\dagger
   \or \ddagger\ddagger \else\@ctrerr\fi}}}
\makeatother

\providecommand{\email}[1]{\href{mailto:#1}{\nolinkurl{#1}\xspace}}

\definecolor{ceruleanblue}{rgb}{0.16, 0.32, 0.75}
\definecolor{darkmidnightblue}{rgb}{0.0, 0.2, 0.4}
\definecolor{darkpastelgreen}{rgb}{0.01, 0.75, 0.24}
\definecolor{bleudefrance}{rgb}{0.19, 0.55, 0.91}

\hypersetup{
     colorlinks   = true,
     citecolor    = bleudefrance,
   linkcolor    = darkpastelgreen,
   urlcolor     = bleudefrance
}

\title{Near-Linear Sample Complexity for \texorpdfstring{\(L_p\)}{Lp} Polynomial Regression}

\author{\hspace{0.2in}
Raphael A. Meyer\thanks{New York University. 
E-mail: \email{ram900@nyu.edu}}
\and
Cameron Musco\thanks{University of Massachusetts Amherst.
E-mail: \email{cmusco@cs.umass.edu}}
\and
Christopher Musco\thanks{New York University. 
E-mail: \email{cmusco@nyu.edu}}\hspace{0.2in}
\and
David P. Woodruff\thanks{Carnegie Mellon University. 
E-mail: \email{dwoodruf@cs.cmu.edu}}
\and
Samson Zhou\thanks{UC Berkeley and Rice University. Work done in part while at Carnegie Mellon University. 
E-mail: \email{samsonzhou@gmail.com}}
}
\date{\today}
\allowdisplaybreaks

\definecolor{cb_orange}{HTML}{FF6A10}
\definecolor{cb_red}{HTML}{D12757}
\definecolor{cb_blue}{HTML}{2437E6}
\definecolor{cb_cyan}{HTML}{75CDFA}
\definecolor{cb_green}{HTML}{4DC961}

\newcommand{\R}{\mathbb{R}}

\newcommand{\norm}[1]{\normof{#1}}
\newcommand{\bv}[1]{\vec{#1}}

\begin{document}
\hypersetup{pageanchor=false}
\begin{titlepage}
\maketitle
\thispagestyle{empty}

\begin{abstract}
We study $L_p$ polynomial regression. Given query access to a function $f:[-1,1] \rightarrow \R$, the goal is to find a degree $d$ polynomial $\widehat{q}$ such that, for a given parameter $\varepsilon > 0$,
\begin{align*}
\|\widehat{q}-f\|_p\le (1+\epsilon) \cdot\min_{q:\deg(q)\le d}\|q-f\|_p.
\end{align*}
Here $\norm{\cdot}_p$ is the $L_p$ norm, $\norm{g}_p = (\int_{-1}^1 |g(t)|^p dt)^{1/p}$.
We show that querying $f$ at points randomly drawn from the Chebyshev measure on $[-1,1]$ is a near-optimal strategy for polynomial regression in all $L_p$ norms.
In particular, to find $\hat q$, it suffices to sample $O(d\, \frac{\polylog d}{\poly\,\varepsilon})$ points from $[-1,1]$ with probabilities proportional to this measure. While the optimal sample complexity for polynomial regression was well understood for $L_2$ and $L_\infty$, our result is the first that achieves sample complexity linear in $d$ and error $(1+\epsilon)$ for other values of $p$ without any assumptions.

Our result requires two main technical contributions. The first concerns $p\leq 2$, for which we  provide explicit bounds on the \textit{$L_p$ Lewis weight function}  of the infinite linear operator underlying polynomial regression.  Using tools from the orthogonal polynomial literature, we show that this function is bounded by the Chebyshev density. 
Our second key contribution is to take advantage of the structure of polynomials to reduce the $p>2$ case to the $p\leq 2$ case. By doing so, we obtain a better sample complexity than what is possible for general $p$-norm linear regression problems, for which $\Omega(d^{p/2})$ samples are required.

\end{abstract}
\end{titlepage}
\hypersetup{pageanchor=true}

\newpage

\section{Introduction}
\label{sec:intro}
We study the problem of learning a near optimal low-degree polynomial approximation to a function $f: [-1,1]\rightarrow \R$ based on as few queries $f(t_1), \ldots, f(t_n)$ to the function as possible. Studied since at least the 19th century with the work of Legendre and Gauss on least squares polynomial regression, this problem remains fundamental in statistics, computational mathematics, and machine learning.
Concretely, our goal is to find a degree $d$ polynomial $\hat{q}$ that satisfies the guarantee:
\[\|\widehat{q}(t)-f(t)\|_p\le(1+\eps)\cdot\min_{\substack{\text{degree } d \\\text{ polynomial }q}}\|q(t)-f(t)\|_p,\]
where $\eps$ is an input accuracy parameter and $\norm{\cdot}_p$ is the $L_p$-norm, i.e., $\norm{g}_p = \left(\int_{-1}^1 |g(t)|^p dt\right)^{1/p}$. 

The problem of near-optimal polynomial approximation, visualized in \figureref{sample:points} and \figureref{best:fit:p}, finds applications ranging from learning half-spaces \cite{KalaiKMS08}, to solving parametric PDEs \cite{HamptonDoostan:2015}, to surface reconstruction \cite{Pratt87}. The choice of norm depends on the application: for example, $p =1$ is used in robust approximation, $p=2$ is common in computational science settings \cite{cohen2017optimal}, and $p = \infty$ is popular in applications where $f$ is smooth and known to admit a good minimax polynomial approximation \cite{KaneKP17,Trefethen:2012}. 
Values of $p$ between $2$ and $\infty$ offer a compromise between robustness and uniform accuracy, and find applications, e.g., in the design of polynomial finite impulse response filters in signal processing \cite{BurrusBS94,dumitrescu2007positive}.

\begin{figure}[!bht]
	\centering
	\begin{minipage}{0.49\columnwidth}
    	\centering
		\includegraphics[width=0.7\columnwidth]{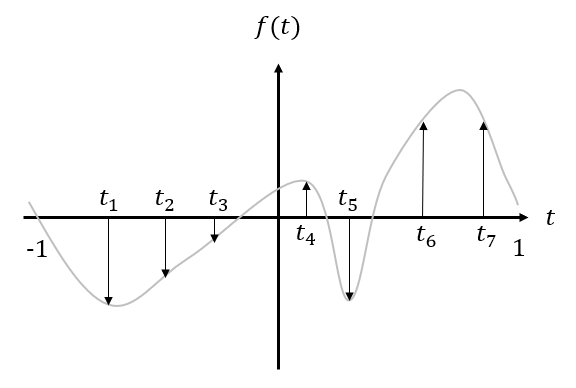}
		\caption{
    		We choose points $t_1,\ldots,t_n$ at which to query a function $f$. Based on $f(t_1), \ldots, f(t_n)$, we want  to find a polynomial approximating $f$ on $[-1,1]$. }
		\label{fig:sample:points}
	\end{minipage}\hfill
	\begin{minipage}{0.49\columnwidth}
    	\centering
		\includegraphics[width=0.7\columnwidth]{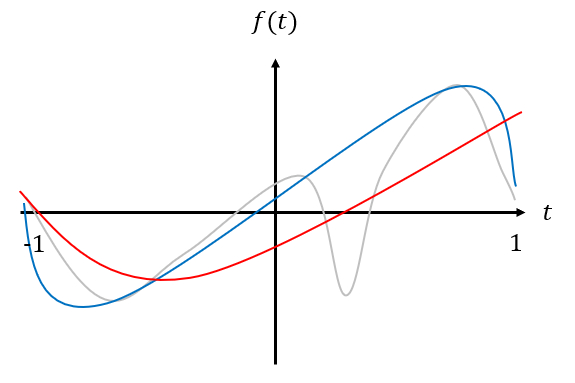}
		\caption{The blue curve is a near optimal approximating polynomial of degree $3$ for $p=2$, while the red curve is near optimal for $p=\infty$.}
		\label{fig:best:fit:p}
	\end{minipage}
\end{figure}

The above problem is an \emph{active learning} or \emph{experimental design} problem since we have the freedom to choose the query locations $t_1, \ldots, t_n$. 
Our goal is to answer two questions:
\begin{enumerate}
    \item As a function of the degree $d$, norm $p$, and tolerance $\epsilon$, how many queries $n$ are required to find $\hat{q}$?
    \item How should the query locations $t_1, \ldots, t_n$ be chosen from $[-1,1]$? 
\end{enumerate}
When $f$ is already a degree $d$ polynomial, via direct interpolation, $d+1$ queries are necessary and sufficient to exactly fit $f$. When $f$ is not a polynomial, we will require more than $d+1$ queries.

The above two questions have been studied extensively for $p = 2$ and $p = \infty$ \cite{Trefethen:2012,RauhutWard:2012,cohen2017optimal,HamptonDoostan:2015}. 
It is well-known that it is sub-optimal to select $t_1,\ldots,t_n$ either from an evenly spaced grid or uniformly at random: methods that try to recover $\hat{q}$ from uniform samples suffer from Runge's phenomenon \cite{Boyd:2009,CohenDavenportLeviatan:2013}. 
Improved results are obtained by selecting more queries near the \emph{edges} of the interval $[-1,1]$. 
When $p = {\infty}$, the typical approach is to select queries at the Chebyshev nodes  \cite{Trefethen:2012}.
Classical work in approximation theory shows that, with \(d+1\) samples, this approach gives an $O(\log d)$ approximation in the $L_\infty$ norm if either polynomial interpolation or a truncated Chebyshev series is used to construct the approximation $\widehat{q}$ \cite{powell1967maximum,Trefethen:2012}.

For $p = 2$, a recent line of work studies randomly querying according to the non-uniform Chebyshev density, which is the asymptotic density of the Chebyshev nodes:
\begin{definition}[Chebyshev density]\label{def:cheby}For $t\in [-1,1]$ the Chebyshev density at $t$ is $\frac{1}{\pi\sqrt{1-t^2}}$.
\end{definition}
The Chebyshev density is larger for values of $t$ near $1$ and $-1$, and is smallest in the center of the interval, as shown in \figureref{clipped_cheby}. 
Prior work proves that sampling query points independently according to this density and then solving a weighted least squares problem returns a solution to the $L_2$ polynomial regression problem with accuracy $(1+\eps)$ using $O\left(d\log d + \frac{d}{\eps}\right)$ queries \cite{RauhutWard:2012,CohenDavenportLeviatan:2013,cohen2017optimal}.
This bound is optimal up to a $\log d$ factor: Chen and Price achieve an $O\left(\frac{d}{\eps}\right)$ result using an alternative approach \cite{ChenPrice:2019a}, with a matching lower bound.
It has also been shown that Chebyshev density sampling solves the $L_\infty$ problem to a constant approximation factor with $O(d\log d)$ samples,
improving on the $O(\log d)$ approximation guarantee for $d+1$ samples that can be obtained via classic techniques \cite{KaneKP17}.

In contrast to $L_\infty$ and $L_2$, there have been far fewer results on near optimal polynomial regression for general $p$. The case of $L_1$ has been studied in the context of robust polynomial regression \cite{KaneKP17}, but results are only given under the strong assumption that $f$ is \(L_\infty\) close to an unknown polynomial. With effort, and at the cost of a computationally expensive sampling procedure, it is possible to extend existing results on active linear regression to obtain near optimal sample complexity bounds for $p\in [1,2]$ (see \sectionref{prior:work} for details). However, for larger values of $p$, all prior methods either require super-linear sample complexity ($\Omega(d^{2})$ or larger), or yield a constant factor instead of a $(1+\epsilon)$ factor approximation.

\subsection{Our Contributions}
We give the first algorithm for active polynomial approximation that simultaneously achieves sample complexity near-linear in $d$ and a $(1+\eps)$ approximation factor for all $L_p$ norms.
Moreover, our procedure is simple, computationally efficient, and \emph{universal}: we just sample points from the Chebyshev density, regardless of the value of $p$. That is, the same approach that works for the $L_2$ norm surprisingly extends to all $L_p$ norms.
Our main result is:
 \begin{theorem}
\label{thm:main}
For any degree $d$, $p \geq 1$, and accuracy parameter $\eps\in(0,1)$, there is an algorithm\footnote{By an artifact of our analysis, we sample \(n\sim B(n_0, \frac{m}{n_0})\) and run \algorithmref{chebyshev-const:Lp:eps}, where \(n_0 = \frac{d^6 p^{O(p)}}{\eps^{O(p^2)}}\) and \(m = d (\frac{p\log(d)}{\eps})^{O(p)}\). This has an overall sample complexity of \(d (\frac{p\log(d)}{\eps})^{O(p)}\) with very high probability.} that queries $f$ at $n = d \left(\frac{p\log(d)}{\eps}\right)^{O(p)}$ points $t_1, \ldots, t_n$, each selected independently at random according to the Chebyshev density on $[-1,1]$, and outputs a degree $d$ polynomial $\hat{q}(t)$ such that, with probability at least \(0.9\), 
\begin{align*}\|\widehat{q}(t)-f(t)\|_p^p\le(1+\eps)\cdot\min_{q:\deg(q)\le d}\|q(t)-f(t)\|_p^p.\end{align*}
\end{theorem}
In addition to the simple sampling procedure, the algorithm for recovering $\widehat{q}$ is also simple: to achieve a constant factor approximation, we show that it suffices to solve an $\ell_p$ polynomial regression problem to find the best degree $d$ polynomial approximating $f$ at our queried points, reweighted appropriately\footnote{We use $\ell_p$ to denote norms on finite dimensional spaces and $L_p$ to denote norms on infinite dimensional spaces.}. To obtain a $(1+\epsilon)$ factor approximation, we first compute a constant factor approximation $q(t)$, and then run the same regression algorithm on the residual $f(t)-q(t)$.
This type of two-stage approach has been used several times in prior work on active learning for linear regression problems \cite{DasguptaDHKM08,MuscoMWY21}. 

The full pseudocode is included in \algorithmref{chebyshev-const:Lp} and \algorithmref{chebyshev-const:Lp:eps} below.

\begin{algorithm}[!tbh]
\caption{Chebyshev sampling for \(L_p\) polynomial approximation, Constant Factor Approximation}
\label{alg:chebyshev-const:Lp}
\begin{algorithmic}[1]
\Require{Access to function \(f\), parameter \(p \geq 1\), degree \(d\), number of samples \(n\)}
\Ensure{Degree \(d\) polynomial \(q(t)\)}
\State{Sample \(t_1,\ldots,t_n\in[-1,1]\) i.i.d.~from the pdf \(\frac{1}{\pi\sqrt{1-t^2}}\)}
\State{Observe function samples \(b_i \defeq f(t_i)\) for all \(i\in[n]\)}
\State{Build \(\mA\in\bbR^{n\times(d+1)}\) and diagonal \(\mS\in\bbR^{n \times n}\) with \([\mA]_{i,j}=t_i^{j-1}\) and \([\mS]_{ii}=\big(\sqrt{1-t_i^2}\big)^{\nicefrac1p}\)}
\State{Compute \(\vx=\argmin_{\vx\in\bbR^{d+1}} \normof{\mS\mA\vx-\mS\vb}_p\)}
\State{Return \(q(t) = \sum_{i=1}^{d+1} x_i t^{i-1}\)}
\end{algorithmic}
\end{algorithm}

\begin{algorithm}[!tbh]
\caption{Chebyshev sampling for \(L_p\) polynomial approximation, Relative Error Approximation}
\label{alg:chebyshev-const:Lp:eps}
\begin{algorithmic}[1]
\Require{Access to function \(f\), parameter \(p \geq 1\), degree \(d\), number of samples \(n\)}
\Ensure{Degree \(d\) polynomial \(p(t)\)}
\State{Run \algorithmref{chebyshev-const:Lp} on \(f\) with \(\frac n2\) samples to get a polynomial \(q(t)\)}
\State{Run \algorithmref{chebyshev-const:Lp} on \(\hat f(t) \defeq f(t) - q(t)\) with \(\frac n2\) samples to get a polynomial \(\hat q(t)\)}
\State{Return \(p(t) \defeq q(t)+\hat q(t)\)}
\end{algorithmic}
\end{algorithm}

\theoremref{main} has a near-optimal dependence on $d$, since a linear dependence is required. We show that our dependence on $\epsilon$ is near optimal as well, proving the following lower bound:
\begin{theorem}
\label{thm:lb}
Let $p\ge 1$ be a fixed constant. 
Any algorithm that can output a \((1+\eps)\) approximation to \(L_p\) polynomial regression with probability \(\frac23\) must use \(n=\Omega(\frac{1}{\eps^{p-1}})\) queries.
\end{theorem}
It can be shown directly that no algorithm that queries $f$ at a finite number of locations can output better than a 2-factor approximation to the best polynomial approximation in the \(L_\infty\) norm with good probability (see \sectionref{L_infty} or \cite{KaneKP17} for details).
On the other hand, a $(1+\epsilon)$ factor approximation is achievable for $p=2$ with just a $1/\epsilon$ dependence in the sample complexity \cite{ChenPrice:2019a}.
Combined with \theoremref{main}, \theoremref{lb} helps complete the picture on the accuracy achievable for all other $L_p$ norms.

\subsection{Our Approach and Comparison to Existing Techniques}
\label{sec:prior:work}

Like prior work on optimal polynomial approximation in the $L_2$ norm \cite{ChenPrice:2019a,cohen2017optimal}, we prove \theoremref{main} by casting the general $L_p$ problem as an \emph{active linear regression problem} involving an infinitely tall design matrix (i.e., a linear operator).
In the finite active linear regression problem, we are given full access to a design matrix $\bA \in \R^{m\times d}$ and query access to a target vector $\vb\in \R^m$. The goal is to query a small number of entries from $\vb$, and based on their values, to approximately solve $\min_{\vx}\|\bA\vx-\vb\|_p$.

To solve the active regression problem for $p=2$, it is known that it suffices to sample $O(d(\log d)/\epsilon)$ entries of $\vb$ with probabilities proportional to the \emph{$\ell_2$ leverage scores} of the corresponding rows in $\bA$ \cite{Sarlo2006,ChenPrice:2019a}. This result generalizes to linear operators with an infinite number of rows in $\bA$ and entries in $\vb$ \cite{AvronKapralovMusco:2019}. The only difference is that for linear operators, we cannot explicitly compute the $\ell_2$ leverage scores (since there are an infinite number of them). To address this challenge,
prior results on $L_2$ polynomial approximation are based on showing that, for the infinite linear operator underlying polynomial regression, the leverage scores can be tightly upper bounded by the Chebyshev measure \cite{RauhutWard:2012,cohen2017optimal}. Sampling by this measure thus yields an upper bound of $O(d(\log d)/\epsilon)$ samples.

To extend these results to general $L_p$ norms, a natural starting point is to leverage generalizations of the $L_2$ leverage scores to other $L_p$ norms.
There are several possible generalizations in the finite matrix case, including the $\ell_p$ leverage scores~\cite{DasguptaDHKM08,CormodeDW18}, the $\ell_p$ sensitivities~\cite{ClarksonWW19,BravermanDMMUWZ20,meyer2021}, and the $\ell_p$ Lewis weights~\cite{CohenP15,ChenD21,ParulekarPP21,meyer2021}.
Unfortunately, na\"ive applications of these tools to the $L_p$ polynomial approximation problem all lead to sub-optimal guarantees. For example, 
it is possible to upper bound the $L_p$ sensitivities by a scaling of the Chebyshev measure. We could then apply recent work on active regression via sensitivity sampling \cite{MuscoMWY21}. However, that work leads to at best a quadratic dependence on $d$.

\begin{figure*}[b]\begin{center}
{\tabulinesep=1.2mm
\begin{tabu}{|c|c|c|}\hline
Approach & Sample Complexity & Approximation \\\hline\hline
$L_p$ sensitivity sampling \cite{MuscoMWY21} & $d^2\,\left(\frac{\log d}{\eps}\right)^{O(p)}$ & $(1+\epsilon)$\\\hline
$L_p$ sensitivity + Lewis weight sampling \cite{MuscoMWY21} & $d^{\max(1,p/2)}\,\left(\frac{\log d}{\eps}\right)^{O(p)}$ & $(1+\eps)$\\\hline
$L_1$ Lewis weight sampling \cite{meyer2021} & $dp^2 \, (\log dp)^{O(1)}$ & $O(1)$\\\hline\hline
\textbf{Chebyshev measure sampling for all $p \geq 1$ (our results)} & $d\,\left(\frac{\log d}{\eps^p}\right)^{O(p)}$ & $(1+\eps)$\\\hline
\end{tabu}
}
\end{center}
\caption{Summary of results for $L_p$ polynomial regression. Our result is the first to obtain both an optimal linear dependence on $d$ for all $p$ as well as a $(1+\epsilon)$ factor approximation. }
\label{fig:results}
\end{figure*}

Alternatively, we might hope to take advantage of recent work on active regression via sampling by $\ell_p$ Lewis weights -- a conceptually different generalization of the $\ell_2$ leverage scores than sensitivities \cite{ChenD21,ParulekarPP21,MuscoMWY21}.
However, there are a few major challenges.
First, we cannot explicitly compute the Lewis weights for the infinite dimensional polynomial operator, and it is much harder to obtain closed form bounds on these weights than it is for the $L_2$ leverage scores and $L_p$ sensitivities.
Second, for regression problems with $d$ features, like degree-$d$ polynomial regression, Lewis weight sampling requires $\tilde{O}(d^{\max(1,p/2)})$ rows \cite{CohenP15,MuscoMWY21}.
So, the approach na\"ively provides linear sample complexity results only for $p\in [1,2]$.\footnote{For $p\in [1,2]$, one option would be to first carefully discretize the regression operator before computing Lewis weights, e.g.,  using $L_p$ sensitivity sampling (the ``first stage'' in \sectionref{two-stage}). While less technically involved than the $p \geq 2$ case, analyzing this approach still requires proving a bound on the $L_p$ sensitivities of the polynomial operator. Moreover, this stage gives sub-optimal dimensionality reduction, so it would be necessary to compute the Lewis Weights of a \(\tilde O(\frac{d^5p^4}{\eps^{2+2p}}) \times d\) matrix, {using significant space and time}, and resulting in a sampling procedure that is not universally good for all $p$.} For polynomial regression specifically, it is possible to use a technique from \cite{meyer2021} to reduce from the general $p$ case to $p\in [1,2]$, which leads to a $dp$ dependence, as in our \theoremref{main}. However, this reduction yields at best a constant factor approximation. The limitations of existing techniques are summarized in \figureref{results}.

To prove \theoremref{main}, we circumvent the above limitations for $L_p$ Lewis weight sampling.
First, for $p\leq 2$, we provide an explicit bound on Lewis weight function of the infinite linear operator underlying polynomial regression, showing that the weight function is closely upper bounded by the Chebyshev measure.
This almost immediately yields our results for $p\in [1,2]$. 
As discussed in \sectionref{tech_overview}, doing so requires a significantly different approach than existing work on bounding leverage scores of the operator.
To the best of our knowledge, our bounds are the first on the Lewis weights of any natural infinite dimensional regression problem, so we hope they will be helpful in related settings where leverage scores have proven powerful.
Examples include active learning for sparse Fourier functions, for bandlimited functions, and for kernel methods in machine learning \cite{ChenKanePrice:2016,ChenPrice:2019, AvronKapralovMusco:2019,ErdelyiMuscoMusco:2020,meyer2020statistical}.

 Second, for $p> 2$, we need to obtain tighter bounds for Lewis weight sampling than available from black-box results that depend on $d^{p/2}$. To do so, we provide a new analysis tailored to the polynomial operator. We show that for \emph{any} $p$, it actually suffices to collect $d \polylog(d)$ samples according to the Lewis weights for some other $p'$ chosen in $[\frac23,2]$. Our analysis requires opening up a net analysis used in \cite{bourgain1989approximation} and \cite{MuscoMWY21} to analyze Lewis weight sampling for general linear operators. We leverage the fact that the $L_{p'}$ Lewis weights are close to the $L_p$ sensitivities -- both are approximated by the Chebyshev measure.

\section{Technical Overview}
\label{sec:tech_overview}
The algorithm that achieves \theoremref{main} is  the same for all $L_p$ norms (sample points via the Chebyshev measure and then solve two weighted $\ell_p$ regression problems -- see \algorithmref{chebyshev-const:Lp} and \algorithmref{chebyshev-const:Lp:eps}). Our analysis differs for \(p\in[1,2]\) and for \(p > 2\).
We first describe the \(p\in[1,2]\) analysis, which is more direct.

\begin{figure}
	\centering
	\begin{minipage}{0.48\columnwidth}
    	\centering
		\includegraphics[width=0.7\columnwidth]{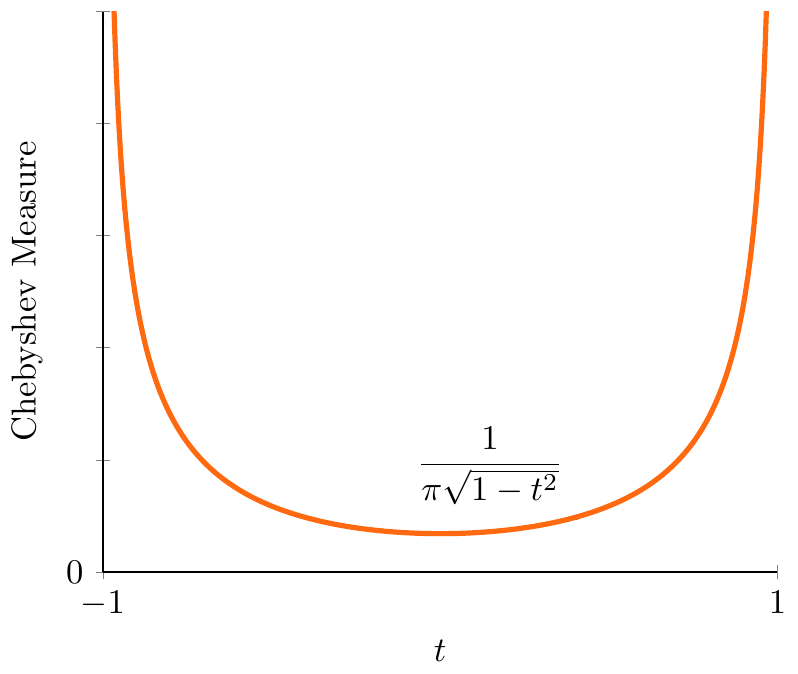}
		\caption{Plot of the Chebyshev Measure on \([-1,1]\). Sampling from the Chebyshev measure draws fewer points from the middle of \([-1,1]\), and more points from the ends of \([-1,1]\)}
		\label{fig:clipped_cheby}
	\end{minipage}\hfill
	\begin{minipage}{0.48\columnwidth}
    	\centering
		\includegraphics[width=0.7\columnwidth]{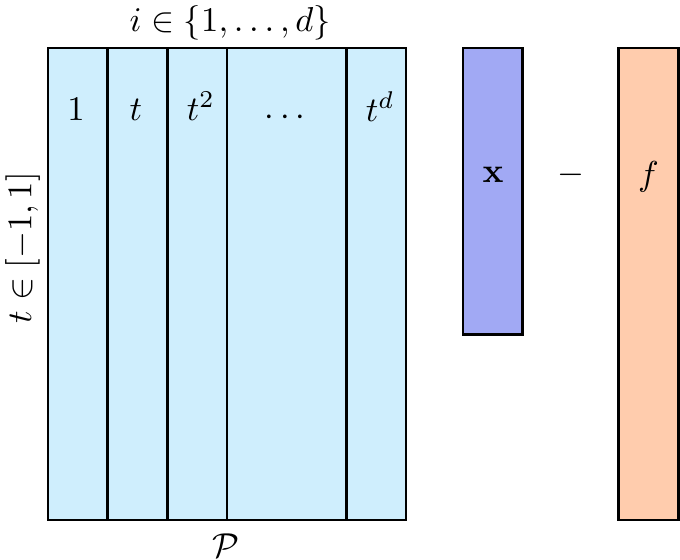}
		\caption{Visualization of the polynomial operator. \cP's column span is the set of degree \(d\) polynomials. We can approximately minimize $\|\cP\vx - \vf\|_p$ by leveraging row-sampling methods for finite matrices. 
}
		\label{fig:operator-as-matrix}
	\end{minipage}
\end{figure}

As discussed, we solve the active polynomial approximation problem by casting it as an \(L_p\) regression problem with an infinitely tall matrix. Concretely,
let \(\cP:\bbR^{d+1} \rightarrow L_2([-1,1])\) be the \emph{polynomial operator}, which maps a coefficient vector \(\vx\in\bbR^{d+1}\) to its corresponding degree \(d\) polynomial: \([\cP\vx](t) \defeq \sum_{k=0}^d \vx_k t^k\) for \(t\in[-1,1]\).
Our original regression problem is equivalent to finding a vector \(\hat\vx\) such that
\[
    \normof{\cP\hat\vx-f}_p^p \leq (1+\eps) \min_{\vx\in\bbR^{d+1}} \normof{\cP\vx-f}_p^p
\]
\figureref{operator-as-matrix} visualizes this operator as matrix with infinite rows.
The \(k^{th}\) column of \cP is the polynomial \(t \mapsto t^k\).
Each row of \cP, indexed by some \(t\in[-1,1]\), is the vector \(\bmat{1 & t & t^2 & \cdots & t^d}\).

As discussed, for $p=2$, an effective approach to solving linear regression problems using a small number of queries of the target function is via leverage score sampling. Specifically, entries of $f$ are sampled independently at random with probability proportional to the leverage score of the corresponding row in $\cP$. For a finite matrix \(\mA\in\bbR^{m \times d}\), the leverage score of the \(i^{th}\) row of \mA is
\[
    \tau[\mA](i) \defeq \max_{\vx\in\bbR^d, \normof\vx_2>0} \frac{([\mA\vx](i))^2}{\normof{\mA\vx}_2^2}.
\]
That is, \(\tau[\mA](i)\) is the maximum contribution that the \(i^{th}\) entry of a vector in \mA's range can make to its \(\ell_2\) norm. This definition naturally extends to linear operators \cite{AvronKapralovMusco:2019,ErdelyiMuscoMusco:2020}, and we can define
\[
    \tau[\cP](t) \defeq \max_{\vx\in\bbR^{d+1},\normof\vx_2>0} \frac{([\cP\vx](t))^2}{\normof{\cP\vx}_2^2}.
\]
For finite matrices the sum of leverage scores is always equal to the rank of $\mA$, and similarly we have that $\int_{-1}^1 \tau[\cP](t) dt = d+1$.
Recalling the particular definition of \cP, we can write \(\tau[\cP](t) = \max_{\deg(q)\leq d} \frac{(q(t))^2}{\normof{q}_2^2}\).
It turns out that this maximum is well-studied in the orthogonal polynomial literature, as it is equal to the reciprocal of the \emph{Christoffel function} \(\lambda_d(t) \defeq \min_{\deg(q)\leq d} \frac{\normof{q}_2^2}{(q(t))^2}\). While difficult to compute exactly, it can be shown that $\lambda_d(t) \geq \frac{c\sqrt{1-t^2}}{d}$ \cite{Nevai:1986}.
This directly implies that, with appropriate scaling, the Chebyshev density upper bounds the leverage function.
That is, we have \(\tau[\cP](t) \leq C \, v(t)\), where \(v(t) \defeq \frac{d+1}{\pi\sqrt{1-t^2}}\) is an appropriate scaling of the Chebyshev measure.
Moreover, since \(\int_{-1}^1 v(t) dt = d+1\), we know that this upper bound is tight up to constants.
Therefore, sampling from the Chebyshev density can be used to solve the \(L_2\) polynomial approximation problem with $O(d\log d + d/\epsilon)$ samples -- only a constant factor more than would be required if sampling by the true leverage scores, which integrate to $d+1$ \cite{ChenPrice:2019}.

\subsection{Bounding Lewis Weights of the Polynomial Operator}
\label{sec:overview-bounding-lewis-weights}

It has recently been shown that active regression results for finite matrices under general $\ell_p$ norms can be obtained by sampling by the Lewis weights, a generalization of the $\ell_2$ leverage scores \cite{ChenD21,MuscoMWY21}.
For a matrix \(\mA\in\bbR^{m \times d}\), the \(\ell_p\) Lewis weights for \mA are the unique numbers \(w_1,\cdots,w_m\) such that
\[
    \tau_i(\mW^{\frac12-\frac1p}\mA) = w_i
    \hspace{1cm}
    \text{for all ~} i\in[m],
\]
where \(\mW\in\bbR^{m \times m}\) is the diagonal matrix with \(\mW_{ii} = w_i\).
As for leverage scores, there are algorithms that compute the Lewis weights for finite matrices.
But since we want to apply the weights to sample from infinite operators, it is necessary to obtain closed form bounds.
It is much less clear how to do so: unlike the leverage scores, the Lewis weights are defined in a circular fashion, instead of as the solution of a natural optimization problem.

To handle this challenge, we turn to the definition of \(C\)-almost Lewis weights for matrices given in \cite{CohenP15}.
Specifically, we say that \(w_1,\cdots,w_m\) are \emph{\(C\)-almost} Lewis weights for \mA if
\begin{align}
    \frac1C w_i \leq \tau[\mW^{\frac12-\frac1p}\mA](i) \leq C w_i
    \hspace{1cm}
    \text{for all ~} i\in[m]
    \label{eq:almost-lewis-defn-matrix}
\end{align}
where \mW is again the matrix with \(w_i\) on its diagonal.
\cite{CohenP15} prove that, for $0 < p \leq 2$, after scaling by a factor of $C$, the \(C\)-almost Lewis weights are upper bounds for the true Lewis weights of a matrix.

This suggests a natural approach to bound the Lewis weights of a matrix: exhibit some weights \(w_1,\cdots,w_m\) and verify that the inequality above holds. In the case of the infinite operator $\cP$, our goal is to find a function \(w(t):[-1,1]\rightarrow\bbR\) such that
\begin{align}
    \frac1C w(t) \leq \tau[\cW^{\frac12-\frac1p}\cP](t) \leq C w(t)
    \hspace{1cm}
    \text{for all ~} t\in[-1,1]
    \label{eq:almost-lewis-defn-operator}
\end{align}
where \([\cW f](t) \defeq w(t) f(t)\) is the linear operator equivalent to a diagonal matrix.

As a first possible candidate for the Lewis weight function, we consider the Chebyshev density \(v(t)\) itself. To do so, we have to bound the leverage function \(\tau[\cV^{\frac12-\frac1p}\cP](t)\), where \([\cW f](t) \defeq w(t) f(t)\).
We establish a surprisingly direct bound based on the fact that for each $p$, the weighting \(\cV^{\frac12-\frac1p}\) aligns with the orthogonalization measure of certain Jacobi orthogonal polynomials. Specifically, we prove:
\begin{theorem}
Let \(J_i^{(\alpha,\beta)}(t)\) denote the degree \(i\) Jacobi Polynomial with parameters \(\alpha\) and \(\beta\).
Then, letting \(\alpha=\beta=\frac1p-\frac12\), we have
\[
    \tau[\cV^{\frac12-\frac1p}\cP](t) = \frac{1}{(1-t^2)^{\frac12-\frac1p}} \sum_{i=0}^d (J_i^{(\alpha,\beta)} (t))^2
\]
\end{theorem}

\pgfplotsset{
    compat=newest,
    /pgfplots/legend image code/.code={        \draw[mark repeat=2,mark phase=2,#1] 
            plot coordinates {
              (0cm,0cm) 
              (0.2cm,0cm)
              (0.4cm,0cm)        };
    },
}

\contourlength{0.2em}

\begin{figure}
	\centering
	\begin{minipage}{0.49\columnwidth}
    	\centering
    	\begin{tikzpicture}
\begin{axis}[
	title={},
	xlabel={\(t\)},
	ylabel={},
	grid=none,
	xmin=-1,
	xmax=1,
	ymin=-1,
	ymax=20,
	xtick = {-1,1,0.8, -0.8},
	xticklabels = {-1,1,\(({\scriptstyle 1-\frac1{d^2}})\)~~~~, \(~~~~~{\scriptstyle-}({\scriptstyle 1-\frac1{d^2}})\)},
	ytick = {-1, 4, 9, 14, 19},
	yticklabels = {{}},
	xtick pos=left,
	ytick pos=left,
	axis lines=left,
	axis line style=-,
	legend pos = north east,
	every axis plot/.append style={thick},
	width=0.85\columnwidth
]

\addplot [dashed, gray, mark=none,mark options={scale=2,solid}, style={thick}] coordinates {(0.8,-1) (0.8,20)};
\addplot [dashed, gray, mark=none,mark options={scale=2,solid}, style={thick}] coordinates {(-0.8,-1) (-0.8,20)};
\addplot [cb_orange] table [mark=none, x=times, y=v, col sep=comma]{levs_data.csv}; \label{legend:cheby}
\addplot [mark=square*, cb_orange, mark size=0.55pt] coordinates {(0.99,15.795072697882215)} node [black, left] {\contour{white}{\(v(t)\)}};
\addplot [cb_blue] table [mark=none, x=times, y=tau_v, col sep=comma]{levs_data.csv}; \label{legend:cheby_levs}
\addplot [mark=square*, cb_blue, mark size=0.55pt] coordinates {(0.75,3.9711731447822496)} node [black, above left] {\(\tau[\mathcal{V}^{\nicefrac12-\nicefrac1p}\mathcal{P}](t)\)};

\end{axis}

\end{tikzpicture}
		\caption{Plot of the scaled Chebyshev Measure (\ref{legend:cheby}) and corresponding reweighted leverage function $\tau[\cV^{\frac12-\frac1p}\cP](t)$ (\ref{legend:cheby_levs}) on \([-1,1]\) for $d=6$, $p=1$.
		For most values of \(t\) both curves are close, but for \(\abs{t}>1-\frac1{d^2}\) the curves diverge.
		This means that the Chebyshev density itself does not directly approximate the $L_p$ Lewis weights, motivating our study of a clipped version of the measure, denoted \(w(t)\).
		}
		\label{fig:cheby_levs_scores}
	\end{minipage}\hfill
	\begin{minipage}{0.49\columnwidth}
    	\centering
    	\begin{tikzpicture}
\begin{axis}[
	title={},
	xlabel={\(t\)},
	ylabel={},
	grid=none,
	xmin=0.5,
	xmax=1,
	ymin=-1,
	ymax=20,
	xtick = {-1,1,0.8, -0.8,0.5},
	xticklabels = {-1,1,\(({\scriptstyle 1-\frac1{d^2}})\)~~~~, \(~~~~~{\scriptstyle-}({\scriptstyle 1-\frac1{d^2}})\),0.5},
	ytick = {-1, 4, 9, 14, 19},
	yticklabels = {{}},
	xtick pos=left,
	ytick pos=left,
	axis lines=left,
	axis line style=-,
	legend pos = north east,
	every axis plot/.append style={thick},
	width=0.85\columnwidth
]

\addplot [dashed, gray, mark=none,mark options={scale=2,solid}, style={thick}] coordinates {(0.8,-1) (0.8,20)};
\addplot [dashed, gray, mark=none,mark options={scale=2,solid}, style={thick}] coordinates {(-0.8,-1) (-0.8,20)};
\addplot [cb_orange] table [mark=none, x=times, y=w, col sep=comma]{levs_data.csv}; \label{legend:clipped}
\addplot [mark=square*, cb_orange, mark size=0.75pt] coordinates {(0.75,3.3686751947823392)} node [black, below left] {\contour{white}{\(w(t)\)}}; 
\addplot [cb_blue] table [mark=none, x=times, y=tau_w, col sep=comma]{levs_data.csv}; \label{legend:clipped_levs}
\addplot [mark=square*, cb_blue, mark size=0.75pt] coordinates {(0.75,3.944238316263322)} node [black, above left] {\(\tau[\mathcal{W}^{\nicefrac12-\nicefrac1p}\mathcal{P}](t)\)};
\addplot [cb_red] table [mark=none, x=times, y=bump_poly, col sep=comma]{levs_data.csv}; \label{legend:spike_poly}
\addplot [mark=square*, cb_red, mark size=0.75pt] coordinates {(0.9,3.0247702891516717)} node [black, below right] {\(q(t)\)};

\addplot [cb_green] coordinates {(0.8, 16) (1, 16)};
\addplot [cb_green] coordinates {(-0.8, 16) (-1, 16)}; \label{legend:upper_bound}
\addplot [mark=square*, cb_green, mark size=0.75pt] coordinates {(0.8, 16)} node [black, above right] {\(O(d^2)\)};
\addplot [mark=square*, cb_green, mark size=0.75pt] coordinates {(-0.8, 16)};

\end{axis}

\end{tikzpicture}
		\caption{
		Plot of the clipped Chebyshev Measure (\ref{legend:clipped}) and corresponding reweighted leverage function (\ref{legend:clipped_levs}) for \(t\in[0.5,1]\) and \(d=6\), $p=1$.
		As proven in \theoremref{chebyshev:ratio}, these functions are within a constant factor for all \(t\), so the clipped measure approximates the $L_p$ Lewis weights for $p \leq 2$. 
		We also visualize the ``spike'' polynomial $q(t)$ (\ref{legend:spike_poly}) and the upper bound (\ref{legend:upper_bound}) used in the proof of \theoremref{chebyshev:ratio}. 
		}
		\label{fig:clipped_lev_scores}
	\end{minipage}
\end{figure}

That is, we can \emph{exactly} characterize this Chebyshev-reweighted leverage function in terms of Jacobi polynomials. 
Further, because Jacobi polynomials are well-studied in the orthogonal polynomial literature, we can appeal to prior work on uniformly upper bounding these polynomials to bound the above sum of squares.
Overall, in \sectionref{lp-regression} we prove:
\begin{align}
    \frac1C v(t) \leq \tau[\cV^{\frac12-\frac1p}\cP](t) \leq C v(t)
    \hspace{1cm}
    \text{for all ~} \abs{t} \leq 1-\frac{1}{d^2}
    \label{eq:fixpoint-guarantee}
\end{align}
This is very close to what we need to show, but unfortunately the almost-Lewis weight property does not hold for large \(|t|>1-\frac1{d^2}\).
\figureref{cheby_levs_scores} shows what goes wrong: the Chebyshev density \(v(t)\) diverges to \(+\infty\) while the weighted leverage function $\tau[\cV^{\frac12-\frac1p}\cP](t)$ remains bounded. To resolve this issue, we adjust our proposed Lewis weight function, and instead consider \(w(t) \defeq \max\{c_1 (d+1)^2, v(t)\}\), which clips the Chebyshev density so that it cannot diverge to \(+\infty\).
We can then show the following core theorem for small $p$:
\begin{theorem}
\label{thm:chebyshev:ratio}
There are fixed constants $c_1,c_2,c_3$ such that, letting $w(t)=\min\left\{c_1(d+1)^2,v(t)\right\}$ be the clipped Chebyshev measure on $[-1,1]$ and letting $\cW$ be the corresponding diagonal operator with $[\cW f](t) = w(t) \cdot f(t)$, for any $p \in [\frac23,2]$ and $t \in [-1,1]$,
\[
    \frac{c_2}{\log^3 d} w(t) \leq \tau[\cW^{1/2-1/p}\cP](t) \leq c_3 w(t)
\]
\end{theorem}
\theoremref{chebyshev:ratio} shows that the clipped Chebyshev density gives a set of \(O(\log^3 d)\)-almost Lewis weights for the polynomial operator. So we can upper bound the true Lewis weights by the clipped measure, and only gain a \(\polylog(d)\) factor in the final sample complexity in comparison to exact Lewis weight sampling. Moreover, we can obtain the same bound via sampling by the Chebyshev measure itself, which tightly upper bounds the clipped measure after scaling (i.e., it has the same integral on $[-1,1]$ up to a constant factor).
We also reiterate that when \(p\in[1,2]\) we will directly appeal to this theorem for this value of \(p\), but when \(p>2\) we will appeal to this theorem for a different \(p'\in[\frac23,2]\), which is why we prove \theoremref{chebyshev:ratio} for some values of \(p<1\).

We prove \theoremref{chebyshev:ratio} by separately considering the case when \(\abs{t} \leq 1-\frac1{d^2}\) and when \(\abs{t} > 1-\frac1{d^2}\). The first case is easier: we show that for such values of $t$, the reweighted leverage function corresponding to the clipped Chebyshev measure -- i.e.  \(\tau[\cW^{1/2-1/p}\cP](t)\) -- very closely approximates the reweighted leverage function corresponding to the unclipped measure. We can then directly appeal to \equationref{fixpoint-guarantee}. The second case is more challenging: when \(\abs{t} > 1-\frac1{d^2}\), the density at $t$ is different in the clipped and unclipped measure, so the reweighted leverage scores differ significantly. To deal with this hard case, we separately prove an upper and lower bound as follows:
\begin{description}
        \item[Upper Bound:] Because \(w(t)\) itself is bounded, we can bound \(\tau[\cW^{1/2-1/p}\cP](t) \leq \tau[\cP](t)\), and we use the Markov Brothers' inequality to bound \(\tau[\cP](t) \leq O(d^2)\).
        \item[Lower Bound:] Because \(\tau[\cW^{1/2-1/p}\cP](t)\) is a maximization over degree \(d\) polynomials, we can prove a lower bound by exhibiting a specific ``spike'' polynomial \(q(t)\) which has \({([\cW^{\frac12-\frac1p} q](t))^2}/{\normof{\cW^{\frac12-\frac1p}q}_2^2} = \Omega(\frac{d^2}{\log^3 d})\).
\end{description}
The detailed proof can be found in \sectionref{lp-regression}. The final result of \theoremref{chebyshev:ratio} is visualized in \figureref{clipped_lev_scores}.

\subsection{Active \texorpdfstring{\(L_p\)}{Lp} Regression via Chebyshev Sampling}
\label{sec:two-stage}

Now that we have now bounded the \(L_p\) Lewis weight function of the polynomial operator \cP by the Chebyshev density for \(p\in[\frac23,2]\), in order to prove \theoremref{main} for \(p\in[1,2]\), we can almost directly apply existing Lewis weight sampling guarantees for active \(\ell_p\) regression \cite{MuscoMWY21,ChenD21}.
However, there remains an outstanding challenge.
Na\"ive Lewis weight sampling for \(\ell_p\) regression on an \(m \times d\) matrix incurs a \(\log(m)\) dependence in the sample complexity\footnote{
    Some work on Lewis weight sampling, including by Cohen and Peng \cite{CohenP15}, implicitly assumes \(\log m = O(\log d)\).
    This is reasonable in the finite matrix setting, but does not apply when \(m\) is infinite.
}.
This rules out directly applying Lewis weight sampling to our infinite operator \cP, for which \(m\) is infinite (recall \figureref{operator-as-matrix}).

We address this challenge with a simple observation: sampling rows of \cP by the Chebyshev measure is essentially equivalent\footnote{
    Two subtleties emerge here.
    First, we say ``essentially equivalent'' since this two-stage sampling scheme is \(O(\frac{1}{\poly(d)})\) close to our actual Chebyshev sampling in total variation, so these schemes are indistinguishable but not the same.
    Second, analyzing the two-stage procedure will require a random choice of the sample number \(n\) -- see the footnote on \theoremref{main}.
} to collecting a large \emph{uniform sample} of rows of \cP and then subsampling those rows according to the Chebyshev measure.
We visualize this ``two-stage'' decomposition of our sampling method in \figureref{two-stage}, and emphasize that we do not algorithmically generate the first uniformly sampled matrix\footnote{
    In principal, we \emph{could} algorithmically generate the uniform subsampled matrix and numerically compute its \(\ell_p\) Lewis weights, although this would incur a much higher polynomial runtime dependence on $d$ than our simpler approach of sampling directly from the Chebyshev measure.
}.
Instead, so long as this hypothetical two-stage algorithm is correct, by the equivalence of these sampling schemes, we know that our actual algorithm is correct.

Proving correctness requires two key ingredients.
Let \(\mA\in\bbR^{n_0 \times d+1}\) be this matrix created by uniformly sampling \(n_0\) rows of \cP.
First, we show that taking \(n_0 = \frac{\poly(d)}{\poly(\eps)}\) suffices to recover a \((1+\eps)\) error solution to the full regression problem on \cP.
Second, we prove that the Chebyshev measure evaluated at \mA's rows tightly upper bounds \mA's Lewis weight distribution. So, by prior work \cite{MuscoMWY21,ChenD21}, this can be used to show that sampling by the measure suffices to obtain a \((1+\eps)\) error solution to the regression problem involving \mA.
This Lewis weight sampling stage only has a dependence on \(\log(n_0) = \log(\frac d\eps)\), avoiding the \(\log(m)\) issue.
Overall, combining the error guarantees of both stages ensures that our hypothetical two-stage algorithm samples rows of \cP in the same way as \algorithmref{chebyshev-const:Lp} and with the same sample complexity as \theoremref{main}.

To prove the first point, that uniform sampling a large number of rows preserves a near-optimal solution, we turn to a different tool from the matrix sampling literature: \(L_p\) \emph{sensitivity sampling}.
The \(L_p\) sensitivities are a natural generalization of the \(L_2\) leverage scores, defined as
\[
    \psi_p[\cP](t)
    \defeq \max_{\vx\in\bbR^{d+1}} \frac{\abs{[\cP\vx](t)}^p}{\normof{\cP\vx}_p^p}
    = \max_{\text{deg}(q)\leq d} \frac{\abs{q(t)}^p}{\normof{q}_p^p}
\]
The value of using \(L_p\) sensitivity sampling is that standard concentration bounds and an \(\eps\)-net argument show that sampling \(n_0 = \frac{\poly(d)}{\poly(\eps)}\) rows proportionally to their sensitivities suffices to recover a \((1+\eps)\) error solution to the full \(L_p\) regression problem.
While the dependence on \(d\) is polynomially worse than that of Lewis weight sampling, it has no dependence on \(m\).
Since we want to sample rows of \cP uniformly, we will need to show a uniform bound on \(\psi_p[\cP](t)\) (i.e., an upper bound that does not depend on \(t\)).
Using a classical result on the smoothness of polynomials (specifically the Markov brothers' inequality), we can indeed show \(\psi_p[\cP](t) \leq d^2(p+1)\), which in turn implies that \(n_0 = \frac{\poly(d)}{\poly(\eps)}\) uniform samples suffice.

\begin{figure}
    \centering
    \includegraphics[width=0.75\columnwidth]{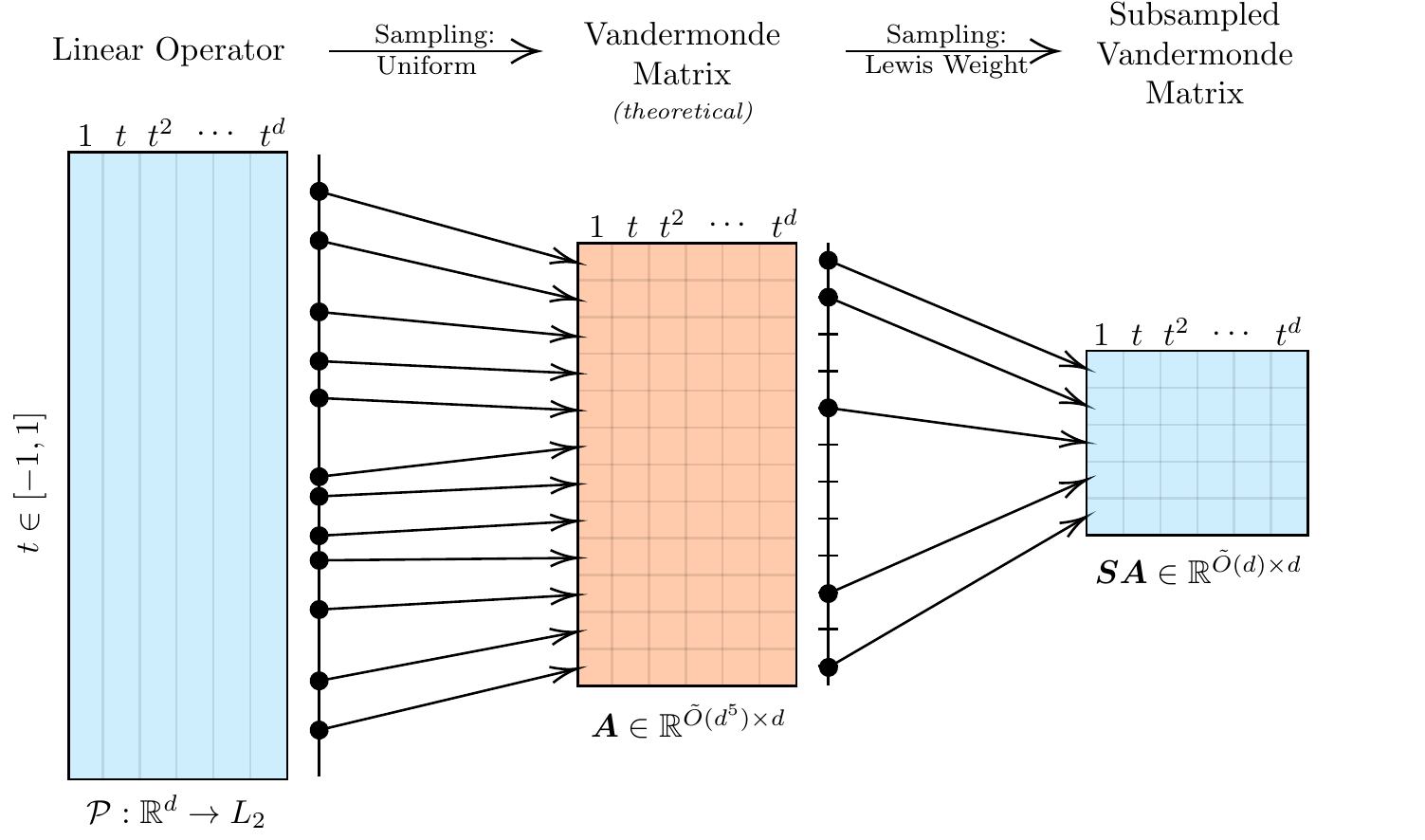}
            \parbox{0.975\textwidth}{
        \caption{
            Sketch of the two-stage proof technique described in Section \ref{sec:two-stage}.
            We show that the Chebyshev measure sampling of \algorithmref{chebyshev-const:Lp} is equivalent to a hypothetical two stage sampling procedure that first 
                                    uniformly samples \(O\left (\frac{\poly(d)}{\poly(\varepsilon)}\right )\) query points from \([-1,1]\) to form Vandermonde matrix \mA, and then further samples the rows of $\bv A$ by the Chebyshev measure, which approximates $\bv A$'s Lewis weight distribution. Since we can uniformly bound the {$L_p$ sensitivities} of the original regression problem by $\poly(d)$, we can argue that both stages of sampling preserve the solution of the $L_p$ regression problem, and thus that our final solution gives a $(1+\varepsilon)$ approximation to the optimal.
                                                        }
        \label{fig:two-stage}
    }
\end{figure}

To prove the second point, we need to show that the Chebyshev measure upper bounds $\mA$'s Lewis weights. To do so, we prove that the clipped Chebyshev measure, which is an almost-Lewis weight measure for \cP, is also an almost-Lewis weight distribution for \mA. Again the proof mostly follows from standard concentration results combined with an \(\eps\)-net argument, although we also need to use the fact that the clipped Chebyshev measure is bounded.

We visualize the structure of our two-stage proof in \figureref{two-stage}. Overall, the arguments above complete the analysis of \(L_p\) polynomial regression for \(p\in[1,2]\).

\subsection{Near-Linear Sample Complexity for \texorpdfstring{$p > 2$}{p>2}}

The next challenge is to extend our results to \(p>2\).
We could use a similar approach as in \sectionref{overview-bounding-lewis-weights} and \sectionref{two-stage}, but doing so would lead to suboptimal sample complexity.
In particular, \(\ell_p\) matrix Lewis weight sampling algorithms have a very different sample complexity for \(p\leq2\) and \(p>2\).
For \(p\in[0,2]\), Lewis weight sampling requires \(\tilde O(d)\) samples.
For \(p>2\), Lewis weight sampling requires \(\tilde O(d^{p/2})\) samples, and there are worst-case matrices that necessitate this sample complexity.
So to achieve \(\tilde O(d)\) sample complexity, we require a novel analysis of \(\ell_p\) Lewis weight sampling for active regression that leverages the structure of the polynomial operator \cP.
Concretely, within the framing of \sectionref{two-stage}, we keep the uniform sensitivity sampling stage but provide a new analysis for the second Lewis weight sampling stage.

We start by describing a simple approach for achieving constant factor error (but not \((1+\epsilon)\) factor) which follows from an observation in \cite{meyer2021}.
In particular, if we only want constant factor error, it suffices to find a subsampling matrix \(\mS\) that satisfies an \(\ell_p\) \emph{subspace embedding} property.
Specifically, we need that  for all \(\vx\in\bbR^{d+1}\), \(\normof{\mS\mA\vx}_p^p \approx \normof{\mA\vx}_p^p\).
We argue that such a matrix can be constructed with a number of rows linear in \(d\) (for any constant \(p\)) as follows:
Let \(f\) be a degree \(d\) polynomial, and let \(r\) be an integer such that \(q \defeq \frac{p}{r} \in [1,2]\).
Then, \(t \mapsto (f(t))^r\) is some degree \(rd\) polynomial.
So, if \(\mA\in\bbR^{n_0 \times d+1}\) is our Vandermonde matrix resulting from uniform sampling, we can let \(\mB\in\bbR^{n_0 \times rd+1}\) be another Vandermonde matrix generated by the same time samples but describing polynomials of degree \(rd\).
Then, for all \(\vx\in\bbR^{d+1}\) there exists some \(\vy\in\bbR^{rd+1}\) such that \((\mA\vx)^r = \mB\vy\), where we define the exponentiation elementwise.
In particular, we have \(\normof{\mA\vx}_p^p = \normof{\mB\vy}_{q}^{q}\).
Therefore, if we know that \mS provides an \(\ell_q\) norm subspace embedding for \mB, so that \(\normof{\mS\mB\vy}_{q}^{q} \approx \normof{\mB\vy}_{q}^{q}\) for all \(\vy\in\bbR^{rd+1}\), we also know that \mS is a subspace embedding for \(\mA\): \(\normof{\mS\mA\vx}_p^p \approx \normof{\mA\vx}_p^p\) for all \(\vx\in\bbR^{d+1}\).
Since \(\mB\) is exactly the Vandermonde matrix we would have generated from uniformly sampling in \sectionref{two-stage} with degree \(rd\) and \(\ell_{q}\) norm, we know that the Chebyshev measure bounds the Lewis weights of \(\mB\), and that the Lewis weight subsampling matrix \mS is a subspace embedding for \mB, and therefore also for \mA.

Achieving \((1+\eps)\) error regression is harder but takes a similar approach.
In order to have Lewis weight sampling imply \((1+\eps)\) error regression, a subspace embedding does not suffice and a more detailed argument is needed \cite{MuscoMWY21}.
A crucial step in this analysis is showing an \emph{affine embedding}: that \(\normof{\mS(\mA\vx-\vb)}_p \approx \normof{\mA\vx-\vb}_p\) for all \(\mA\vx\) with small \(\ell_p\) norm.
\cite{bourgain1989approximation} and \cite{MuscoMWY21} provide a way to prove this affine embedding via a \emph{compact rounding argument}, which designs a structured set of \(\eps\)-nets which allow for a tight  \(\tilde O(d^{\max\{1,p/2\}})\) sample complexity to be obtained from Lewis weight sampling.
To obtain a linear dependence in \(d\) for all \(p\), we reduce from the \(\ell_p\) case to the \(\ell_q\) case for \(q\leq 2\), as discussed above, but in a less direct way.
In particular, we show that a compact rounding for the range of \mB can be directly transformed to construct a compact rounding of the same size for the range of \mA.

This approach is elaborated on in Section \ref{sec:affine-embed}. Critically, we will now enforce that \(r\) is also an odd integer, so that we not only get \((\mA\vx)^r = \mB\vy\) but also have \(\mA\vx = (\mB\vy)^{1/r}\).
This does not hold when \(r\) is even since negative entries of \(\mA\vx\) get turned positive.
For \(p\geq 3\), we let \(r\) be the largest odd integer smaller than \(p\), so that \(q = \frac pr \in [1,2]\).
For \(p\in(2,3)\), this would pick \(r=1\) which would not be helpful, so we instead take \(r=3\), so that \(q = \frac pr \in [\frac23,1]\).
Once we construct this compact rounding, we find that sampling the rows of \mA by the \(\ell_q\) Lewis weights of \mB achieves the affine embedding with sample complexity linear in \(d\).
And since \sectionref{overview-bounding-lewis-weights} bounds the Lewis weights of \mB by the Chebyshev measure, we conclude that \algorithmref{chebyshev-const:Lp:eps} achieves \theoremref{main} for all \(p > 2\).

\subsection{Lower Bounds and \texorpdfstring{\(L_\infty\)}{L∞} Polynomial Approximation}
The linear dependence on $d$ in \theoremref{main} cannot be improved: when $f$ is exactly equal to a degree $d$ polynomial, if we do not take at least $d+1$ samples it is not possible to recover a zero-error approximation to the function.
A natural question is if the $1/\epsilon^{O(p)}$ dependence in the theorem is also tight -- i.e., is it necessary for the accuracy to depend exponentially on $p$? 

We answer this question in the affirmative with the lower bound of \theoremref{lb}, which has a short and direct proof.
For any algorithm that queries \(f\) at most \(n \leq O(\frac1{\eps^{p-1}})\) times, there must exist an interval $\cI\subset [-1,1]$ of width \(\eps^{p-1}\) such that none of the algorithm's queries lie in \cI with probability \(\frac23\).
We then randomly select a function \(f\) that is either $+1$ or $-1$ on \cI with equal probability, and $0$ elsewhere. To obtain a $1+\epsilon$ approximation in the $L_p$ norm, the algorithm must distinguish between these two cases, but with probability \(\frac23\), it does not even obtain a sample from the non-zero region.

Finally, we note that our techniques can be extended to give a constant factor approximation to the $L_\infty$ polynomial approximation problem with $O(d\polylog(d))$ samples.
Details are discussed in \sectionref{L_infty}, where we relate the \(L_\infty\) problem to the \(L_p\) problem with \(p = O(\log d)\).
Results for the \(L_\infty\) norm were already shown in \cite{KaneKP17} using a different approach but the same Chebyshev measure sampling distribution.

\begin{figure*}[!thb]
\centering
\begin{tikzpicture}[scale=1.25]

\node at (0,0){Markov Brothers'};
\node at (0,-0.4){Inequality};
\draw[dashed] (-1.2,-0.7) rectangle+(2.4,1); 

\node at (3.75,0){Matrix Lewis weights};
\node at (3.75,-0.4){/ leverage scores};
\draw[dashed] (3.75-1.5,-0.7) rectangle+(3,1);

\node at (8,0){Orthogonal polynomials};
\node at (8,-0.4){(Chebyshev / Jacobi)};
\draw[dashed] (8-1.6,-0.7) rectangle+(3.2,1);

\draw[->] (0.6,-0.7) -- (0.8,-1.2);
\draw[->] (-0.45,-1.7) -- (0.45,-2.7);

\draw[->] (4.5,-0.7) -- (4.5,-1.2);
\draw[->] (2.7,-1.7) -- (3.3,-1.7);

\draw[->] (5.25,-0.2) to[out=-70,in=150] (7,-1.2);
\draw[->] (7.5,-0.7) -- (7.5,-1.2);

\node at (-1.4,-1.5+0){Reduction};
\node at (-1.4,-1.5-0.4){to \(\ell_q\) space};
\draw[dashed] (-1.55-0.8,-1.5-0.7) rectangle+(1.9,1);

\node at (1.5,-1.5+0){$L_p$ Sensitivity};
\node at (1.5,-1.5-0.4){Bounds, $p\ge 1$};
\draw (1.5-1.2,-1.5-0.7) rectangle+(2.4,1);

\node at (4.5,-1.5+0){Uniform Sampling};
\node at (4.5,-1.5-0.4){(two stage)};
\draw (4.5-1.2,-1.5-0.7) rectangle+(2.4,1);

\node at (7.9,-1.5+0){Almost Lewis Weight};
\node at (7.9,-1.5-0.4){Bounds, $p\in\left[\frac{2}{3},2\right]$};
\draw (7.9-1.6,-1.5-0.7) rectangle+(3.2,1);

\draw[->] (8,-2.2) to[out=-160,in=20] (1.5,-2.7);
\draw[->] (-0.2,-3.2) -- (0.45,-3.2);
\draw[->] (2.05,-3.2) -- (2.65,-3.2);
\draw[->] (9,-2.2) -- (9,-2.7);
\draw[->] (4.5,-2.2) -- (4.5,-2.7);
\draw[->] (5.7,-1.7) to[out=-70,in=150] (7.5,-2.7);

\node at (-1,-1.5*2+0){Compact};
\node at (-1,-1.5*2-0.4){Rounding};
\draw[dashed] (-1-0.8,-1.5*2-0.7) rectangle+(1.6,1);

\node at (1.25,-1.5*2+0){Affine};
\node at (1.25,-1.5*2-0.4){Embedding};
\draw (1.25-0.8,-1.5*2-0.7) rectangle+(1.6,1);

\node at (4.25,-1.5*2+0){$(1+\eps)$-approximation,};
\node at (4.25,-1.5*2-0.4){$p\ge2$ (one stage)};
\draw[color=purple] (4.25-1.6,-1.5*2-0.7) rectangle+(3.2,1);

\node at (8,-1.5*2+0){$(1+\eps)$-approximation,};
\node at (8,-1.5*2-0.4){$p\in[1,2]$ (one stage)};
\draw[color=purple] (8-1.6,-1.5*2-0.7) rectangle+(3.2,1);

\end{tikzpicture}
\caption{Flowchart of proofs: dashed rectangles represent existing results, solid rectangles represent our technical contributions.}
\label{fig:flowchart}
\end{figure*}
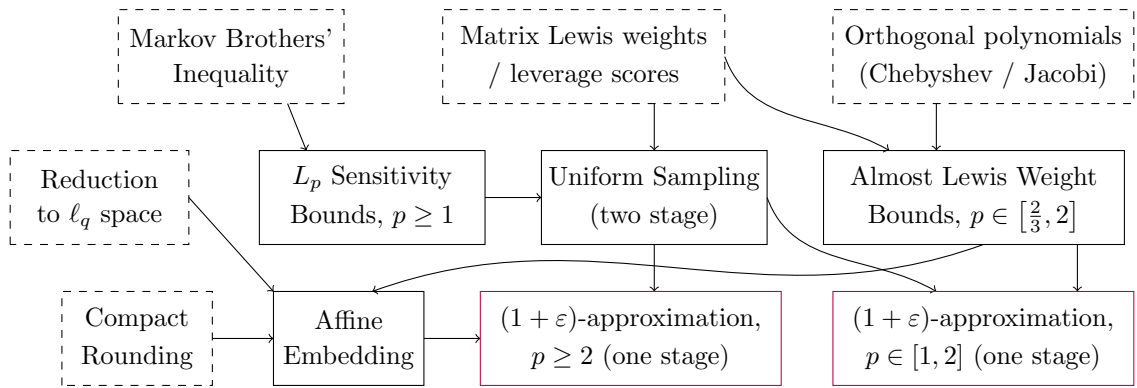

\paragraph{Organization of the rest of the paper.}
We first consider the $L_p$ regression problem for $p\in[1,2]$ in \sectionref{lp-regression}.
Specifically, we start by relating the Chebyshev density to the $L_p$ Lewis weights for all $p\in[\frac23,2]$.
We first outline the proof for $p=1$ in \sectionref{one:lewis:ratio} and defer the proof for general $p\in[\frac23,2]$ to \sectionref{lp-lewis-weights} and \sectionref{clipped-analysis}. 
We then prove correctness of \algorithmref{chebyshev-const:Lp} for $p\in [1,2]$ in \sectionref{smallp:regression:constant} and \sectionref{smallp:regression:eps}.

We handle $p>2$ in \sectionref{lp-regression:largep}.
We first prove the correctness of constant-factor regression in \sectionref{large-p-const-factor}, prove the majority of \((1+\eps)\) error analysis in \sectionref{largep:regression:eps}, and prove a core technical claim for \(p>2\) in \sectionref{affine-embed}.
We present the lower bound \theoremref{lb} in \sectionref{largep:regression:eps}. 
Finally, we address $L_\infty$ regression in \sectionref{L_infty}. 
A summary of our high-level ideas and their dependencies is shown in \figureref{flowchart}.

\section{Preliminaries}
For an integer $n>0$, we use $[n]$ to denote the set $\{1,\ldots,n\}$. 
We use $\poly(n)$ to denote a constant degree polynomial in $n$
and $\polylog(n)$ to denote a polynomial in $\log n$. 

Throughout this paper, unbold lowercase letters are scalars or functions, bold lowercase letters are vectors, bold uppercase letters are matrices, and calligraphic uppercase letters are linear operators.
The norm \(\normof{\cdot}_p\) will interchangeably refer to the vector norm, defined by \(\normof{\vx}_p^p = \sum_{i=1}^d |x_i|^p\), and the continuous norm \(\normof{f}_p^p = \int_{-1}^1 \abs{f(t)}^p dt\).
We say that a matrix \mA is a subspace embedding for another matrix or linear operator \cA if for all \vx we have \(\frac{1}{\alpha} \normof{\cA\vx}_p^p \leq \normof{\mA\vx}_p^p \leq \alpha \normof{\cA\vx}_p^p\) for some constant \(\alpha \geq 1\).
More broadly, if two scalars \(x\) and \(y\) have \(\frac1\alpha x \leq y \leq \alpha x\), then we write \(x \approx_\alpha y\).
For instance, the subspace embedding guarantee can be written as \(\normof{\cA\vx}_p^p \approx_\alpha \normof{\mA\vx}_p^p\) for all \vx.
We use brackets for indexing on both vectors and functions.

The \(i^{th}\) entries of the vectors \vx and \(\mA\vx\) are denoted \(\vx(i)\) and \([\mA\vx](i)\).
The \(\ell_2\) leverage score of the \(i^{th}\) row of matrix \mA is denoted \(\tau[\mA](i)\).
The \(\ell_p\) Lewis weight of the \(i^{th}\) row of matrix \mA is denoted \(w_p[\mA](i)\).
The \(\ell_p\) sensitivities of the \(i^{th}\) row of matrix \mA is denoted \(\psi_p[\mA](i)\).
We similarly denote the leverage function, Lewis weight function, and sensitivity of an operator \cA at time \(t\) as \(\tau[\cA](t)\), \(w_p[\cA](t)\), and \(\psi_p[\cA](t)\).

Let \cP denote the polynomial operator of degree \(d\):
\[
    \cP:\bbR^{d+1} \rightarrow L_p([-1,1])
    \hspace{2cm}
    [\cP\vx](t) \defeq \sum_{i=0}^d x_i t^i
\]
Note that the maximum degree of a polynomial is \(d\), but that the \emph{rank} of \cP is \(d+1\) because of the constant degree-0 polynomial.

We recall the Markov brothers' inequality that bounds the magnitude of the derivative of a polynomial of degree $d$ whose magnitude inside the interval $[-1,1]$ is bounded by $1$. 
\begin{theorem}[Markov brothers' inequality, e.g., Theorem 2.1 in \cite{GovilM99}]
\label{thm:markov:bros}
Suppose $q(t)$ is a polynomial of degree at most $d$ such that $|q(t)|\le 1$ for $t\in[-1,1]$. 
Then for all $t\in[-1,1]$, $|q'(t)|\le d^2$. 
\end{theorem}

Throughout this paper, we will be analyzing \algorithmref{chebyshev-const:Lp}, and showing that this algorithm satisfies \theoremref{main}.

\section{Active \texorpdfstring{\(L_p\)}{Lp} Regression for \texorpdfstring{$p\in[1,2]$}{p in [1,2]}}
\label{sec:lp-regression}

In this section, we start with the definition of leverage scores and prove that the \(L_1\) Lewis weights for the polynomial operator are bounded by the Chebyshev measure.
In particular, this section shows the relationship between Lewis weights and uniform bounds on orthogonal polynomials.
We then use this Lewis weight bound to show that \(\tilde O(d)\) samples suffice for robust \(L_1\) regression.

\subsection{Warm Up: Bounding the Leverage Scores for \texorpdfstring{\(p=2\)}{p=2}}
\label{sec:l2-leverage-bound}
We first start with leverage scores, which are a key building block underpinning Lewis weights.
Before discussing Lewis weights, we will look at bounding the leverage scores of \cP, which relates to solving \(L_2\) regression.
We first look at the properties of Leverage Scores for matrices:
\begin{definition}
For a matrix \(\mA\in\bbR^{n \times d}\), the leverage score for row \(i\in[n]\) is:
\[
    \tau[\mA](i) \defeq \max_{\vx\in\bbR^d,\,\norm{\mA\vx}_2 >0} \frac{([\mA\vx](i))^2}{\normof{\mA\vx}_2^2}
\]
\end{definition}
\noindent The leverage scores of a matrix are well studied, and we will rely on two of their properties:
\begin{enumerate}
    \item
    Leverage Scores are invariant to change of basis: for full-rank \(\mU\in\bbR^{d \times d}\), we have \(\tau[\mA\mU](i)=\tau[\mA](i)\).
    \item
    If \mA has orthonormal columns, then \(\tau[\mA](i)=\normof{\va_i}_2^2\) where \(\va_i\) is the \(i^{th}\) row of \mA.
\end{enumerate}
So, if we can find a matrix \mU such that \(\mA\mU\) has orthonormal columns, then we can compute \(\tau_i(\mA)=\normof{[\mA\mU](i)}_2^2\).
We can use this argument to bound the \textit{Leverage Function} of the polynomial operator:
\begin{definition}
For an operator \(\cA:\bbR^{d+1}\rightarrow L_2([-1,1])\), the leverage function for \cA at time \(t\in[-1,1]\) is
\[
    \tau[\cA](t) \defeq \max_{\vx\in\bbR^{d+1},\,\norm{\cA\vx}_2 >0} \frac{([\cA\vx](t))^2}{\normof{\mathcal{A}\vx}_2^2}
\]
\end{definition}
We can easily see that the leverage function is also rotationally invariant.
As shown in \figureref{operator-as-matrix}, \cP has columns that represent the first degree \(d\) monomials.
That is, we think of the \(i^{th}\) column of \cP as the polynomial \(p_i(t) = t^{i-1}\).
Since \(\int_{-1}^1 p_i(t) p_j(t) dt \neq 0\) in general, these columns are not orthogonal.

While the first degree \(d\) monomials are not orthogonal, the Legendre polynomials are.
So, we can find a change-of-basis matrix \mU such that the columns of \(\cP\mU\) are Legendre polynomials instead.
Under this basis, we have \(\normof{\cP\mU\vx}_2^2 = \normof{\vx}_2^2\), which lets us simplify the leverage function.
Letting \(L_i(t)\) denote the degree \(i\) Legendre polynomial, normalized so that \(\int_{-1}^1 (L_i(t))^2 dt=1\), we have
\begin{align}
    \tau[\cP](t)
        = \max_{\vx\in\bbR^{d+1}} \frac{([\cP\mU\vx](t))^2}{\normof{\cP\mU\vx}_2^2}
        = \max_{\normof\vx_2=1} ([\cP\mU\vx](t))^2
    = \max_{\normof\vx_2=1} \left(\sum_{i=0}^d x_i L_i(t) \right)^2
    = \sum_{i=0}^d (L_i(t))^2
    \label{eq:leverage-row-norm}
\end{align}
The last equality follows because \(\max_{\normof\vx_2=1} (\va^\top\vx)^2 = \normof\va_2^2\) for any \va.
If we view \(\cP\mU\) as an infinite matrix whose rows correspond to \(t\in[-1,1]\) and whose columns correspond to the Legendre polynomials, then \equationref{leverage-row-norm} shows that \(\tau[\cP](t)\) equals the row-norm-squared of this matrix, matching the second property we mentioned for matrix leverage scores.

So, to bound the leverage function for \cP, we now need to bound the sum-of-squared Legendre polynomials.
Here we appeal to existing uniform bounds on orthogonal polynomials.
For instance, Lorch proved in 1983 that \(\abs{L_i(t)} \leq \sqrt{\frac{2}{\pi\sqrt{1-t^2}}}\) for all \(t\in[-1,1]\) \cite{lorch1983alternative}.
So we conclude the bound
\[
    \tau[\cP](t)
    = \sum_{i=0}^d (L_i(t))^2
    \leq \sum_{i=0}^d \frac{2}{\pi\sqrt{1-t^2}}
    = \frac{2(d+1)}{\pi\sqrt{1-t^2}}
    = 2v(t)
\]
That is, the leverage function is upper bounded by the Chebyshev measure, which intuitively implies that \(O(d \log d)\) samples from the Chebyshev measure suffice to recover a polynomial for \(L_2\) regression.
Formally, for \(L_2\) regression, this technique can be analyzed using the tools in \cite{ChenPrice:2019a} or \cite{RauhutWard:2012}.

\subsection{Bounding the Lewis Weights for \texorpdfstring{\(p=1\)}{p=1}}
\label{sec:one:lewis:ratio}

Having covered the \(L_2\) case, we now focus on \(p=1\), where the leverage function is no longer sufficient.
We turn to Lewis weights, and start by considering the standard matrix setting:
\begin{definition}
    Let \(\mA\in\bbR^{n \times d}\), and \(p \ge 0\).
    Then the \(\ell_p\) Lewis weights for \mA are the unique weights \(w_p[\mA](1),\ldots,w_p[\mA](n)\) such that
    \[
        \tau\big[\overline{\mW}^{\frac12-\frac1p}\mA\big](i) = w_p[\mA](i)
    \]
    for all \(i\in[n]\), where \(\overline\mW\in\bbR^{n \times n}\) is the corresponding diagonal matrix with \(\overline\mW_{i,i} = w_p[\mA](i)\).
\end{definition}
\noindent \cite{CohenP15} show several important properties of Lewis weights:
\begin{enumerate}
    \item When \(p\in[1,2]\), sampling \(O(d \log d)\) rows of \mA with respect to its Lewis weights suffice to recover an \(\ell_p\) subspace embedding.
    \item If some other weights \(w_1,\ldots,w_n\) have \(\frac1C \leq \frac{\tau[\mW^{\frac12-\frac1p}\mA](i)}{w_i} \leq C\) for all \(i\in[n]\) and some constant \(C\), where \(\mW_{i,i}=w_i\), then \(w_1,\ldots,w_n\) are close to the true Lewis weights.
\end{enumerate}
In particular, if we can find any such \(w\)'s, then we can sample \(O(d \log d)\) rows of \mA with respect to \(w_1,\ldots,w_n\) and still get an \(\ell_p\) subspace embedding, which suffices to recover a near-optimal solution to \(\ell_p\) regression.
This motivates our approach, where we show that the Chebyshev Measure \(v(t)\) nearly satisfies this guarantee.

We start by defining Lewis weights for operators:
\begin{definition}
    For an operator \(\cA:\bbR^{d+1}\rightarrow L_1([-1,1])\), a Lewis weight function for \cA satisfies
    \[
        w_p[\cA](t) = \tau[\overline{\cW}^{\frac12-\frac1p}\cA](t)
    \]
    for all \(t\in[-1,1]\), where \(\overline\cW\) is the corresponding diagonal operator such that \([\overline{\cW} x](t) = w_p[\cA](t) \cdot x(t)\) for any function \(x\).
\end{definition}
The Chebyshev Measure will not satisfy this strict equality criteria, so we instead consider the approximate criteria:
\begin{definition}
    For an operator \(\cA:\bbR^{d+1}\rightarrow L_1([-1,1])\), a function \(w(t)\) is a \emph{\(C-\)Almost Lewis Weight Function} for \cA if
    \[
        \frac1C \leq \frac{\tau[\cW^{\frac12-\frac1p}\cA](t)}{w(t)} \leq C
    \]
    for all \(t\in[-1,1]\), where \(\cW\) is the corresponding diagonal operator such that \([\cW x](t) = w(t) \cdot x(t)\) for any function \(x\).
    We often refer to \(\frac{\tau[\cW^{\frac12-\frac1p}\cA](t)}{w(t)}\) as the \emph{Lewis Weight Fixpoint Ratio}.
\end{definition}
Similarly to the \(L_2\) case, we relate the leverage function to a class of orthogonal polynomials.
However, for \(p\neq2\), the Legendre polynomials do not make the columns of \(\cW^{\frac12-\frac1p}\cA\) orthogonal.
For \(p=1\), we turn to Chebyshev Polynomials of the Second Kind, denoted \(U_i(t)\), which satisfy \(\int_{-1}^1 U_i(t) U_j(t) \sqrt{1-t^2} dt = \frac\pi2\mathbbm1_{[i=j]}\).
\begin{theorem}
    \label{thm:cheby-lewis-weight-eq}
    Let \(v(t)\defeq \frac{d+1}{\pi\sqrt{1-t^2}}\), \(\cV\) be the diagonal operator for \(v(t)\), and \(U_i(t)\) be the degree \(i\) Chebyshev polynomial of the second kind.
    Then,
    \[
        \frac{\tau[\cV^{-\frac12}\cP](t)}{v(t)} = 1 + \frac{1-U_{2(d+1)}(t)}{2(d+1)}
    \]
\end{theorem}

\begin{proof}
Let \(\mU\) be the change-of-basis matrix such that \(\cP\mU\) has columns that are Chebyshev polynomials of the second kind.
We first verify the orthogonality by simplifying the denominator of \(\tau[\cV^{-\frac12}\cP](t) = \max_{\vx}\frac{([\cV^{-\frac12}\cP\mU\vx](t))^2}{\normof{\cV^{-\frac12}\cP\mU\vx}_2^2}\):
\begin{align*}
	\normof{\cV^{-\frac12}\cP\mU\vx}_2^2
	&= \int_{-1}^1 \left(\sum_{i=0}^d x_i U_i(s) \frac1{\sqrt{v(s)}}\right)^2 ds \\
	&= \frac{\pi}{d+1} \sum_{i=0}^d\sum_{j=0}^d x_ix_j \int_{-1}^1 U_i(s)U_j(s) \sqrt{1-s^2} ds \\
	&= \frac{\pi^2}{2(d+1)} \normof\vx_2^2
\end{align*}
With this orthogonality, we can rewrite the rescaled leverage scores as a squared row-norm:
\begin{align*}
	\tau[\cV^{-\frac12}\cP](t)
	&= \max_{\vx\in\bbR^{d+1}}
		\frac{
			\frac{1}{v(t)}([\cP\mU\vx](t))^2
		}{
			\normof{\cV^{-\frac12}\cP\mU\vx}_2^2
		} \\
	&= \frac{2(d+1)}{\pi^2v(t)} \max_{\vx\in\bbR^{d+1}}
		\frac{
			([\cP\mU\vx](t))^2
		}{
			\normof{\vx}_2^2
		} \\
	&= \frac{2(d+1)}{\pi^2v(t)} \max_{\normof{\vx}_2=1} \left(\sum_{i=0}^d x_i U_i(t)\right)^2 \\
	&= \frac{2(d+1)}{\pi^2v(t)} \sum_{i=0}^d (U_i(t))^2
\end{align*}
We now simplify this sum-of-squares term by using the specialized trigonometric structure of the Chebyshev polynomials of the second kind.
Letting \(\theta\defeq\cos(t)\), note that \(U_i(t) = \frac{\sin((i+1)\theta)}{\sqrt{1-t^2}}\) and the Chebyshev polynomials of the \textit{first} kind have \(T_i(t) = \cos(i\theta)\).
Then,
\[
    (U_i(t))^2
    = \frac{\sin^2((i+1)\theta)}{1-t^2}
    = \frac{\frac12 - \frac12 \cos(2(i+1)\theta)}{1-t^2}
    = \frac{\frac12 - \frac12 T_{2(i+1)}(\theta)}{1-t^2}
    = \frac{1}{2(1-t^2)} \cdot (1 - T_{2(i+1)}(t))
\]
so that \(\sum_{i=0}^d (U_i(t))^2 = \frac{1}{2(1-t^2)} \left((d+1) - \sum_{i=0}^d T_{2(i+1)}(t)\right)\).
Using the relation \(U_k(t) = 2 \sum_{\text{even } j=1}^k T_{j}(t) - 1\) for even \(k\), we simplify this summation as \(\sum_{i=0}^d T_{2(i+1)}(t) = \frac12 U_{2(d+1)}(t) + \frac12 - T_0(t)\).
Since \(T_0(t)=1\), \(\sum_{i=0}^d T_{2(i+1)}(t) = \frac12 U_{2(d+1)}(t) - \frac12\).
Returning to the rescaled leverage function,
\begin{align*}
    \tau[\cV^{-\frac12}\cP](t)
    &= \frac{2(d+1)}{\pi^2v(t)} \sum_{i=0}^d (U_i(t))^2 \\
    &= \frac{2(d+1)}{\pi^2v(t)} \cdot \frac{d+1}{2(1-t^2)} \left(1 + \frac{1 - U_{2(d+1)}(t)}{2(d+1)} \right) \\
    &= v(t) \left(1 + \frac{1 - U_{2(d+1)}(t)}{2(d+1)} \right),
\end{align*}
which completes the proof.
\end{proof}

Recall that for \(v(t)\) to be almost Lewis weights for \cP, we need \(\frac{\tau[\cV^{-\frac12}\cP](t)}{v(t)} = \Theta(1)\) for all \(t\in[-1,1]\).
Since \(\frac{-1}{\sqrt{1-t^2}} \leq U_i(t) \leq \frac{1}{\sqrt{1-t^2}}\), we can see that \theoremref{cheby-lewis-weight-eq} satisfies this criteria for almost all \(t\):
\begin{corollary}
\label{corol:cheby-lp-ratio}
For \(\abs{t} = 1 - O(\frac{1}{d^2})\), we have \(\frac1C \leq \frac{\tau[\cV^{-\frac12}\cP](t)}{v(t)} \leq C\) for some constant \(C\).
\end{corollary}
We prove this formally in \sectionref{tauV:vt}.
For \(\abs{t}\rightarrow1\), we know that \(\abs{U_{2(d+1)}(t)}\rightarrow2(d+1)\), so that \(\frac{\tau[\cV^{-\frac12}\cP](t)}{v(t)} \rightarrow 0\), meaning that the almost Lewis weight property does not hold.
So, while the Chebyshev measure seems to match the Lewis weights for most \(t\), it is wrong for \(t\) close to the ``endcaps'' at \(-1\) and \(1\).

To understand why the Chebyshev measure fails at the endcaps, we note an important property of the leverage function.
By the Markov Brother's Inequality, the leverage function is at most \(O(d^2)\) for all \(t\in[-1,1]\).
However, the Chebyshev measure is unbounded as \(\abs{t} \rightarrow 1\).
So, there must be a gap between these two distributions.

To resolve this gap, we analyze the \emph{Clipped Chebyshev Measure} \(w(t)\), shown in \figureref{clipped_cheby}, which lies below the true Chebyshev measure \(v(t)\), and which only differs in this endcap region:
\begin{definition}
The Clipped Chebyshev Measure is the function \(w(t) \defeq \min\{C(d+1)^2, \frac{(d+1)}{\pi\sqrt{1-t^2}}\}\).
\end{definition}
With a more involved analysis relegated to \sectionref{clipped-analysis}, we show that 1) \(\tau[\cW^{-\frac12}\cP](t) = \tilde\Theta(d^2)\) in the endcaps and 2) \(\tau[\cW^{-\frac12}\cP](t) = \Theta(\tau[\cV^{-\frac12}\cP](t))\) for \(\abs{t} \leq 1 - O(\frac1{d^2})\).
This final step completes our first major technical claim:
\begin{lemma}[\theoremref{chebyshev:ratio} for $p=1$]
\label{lem:chebyshev:ratio:one}
There are fixed constants $c_1,c_2$ such that, for $p=1$ and $t \in [-1,1]$,
\[
	\frac{c_1}{\log^3 d}\le\frac{\tau[\cW^{1/2-1/p}\cP](t)}{w(t)}\le c_2.
\]
\end{lemma}
The full proof of \theoremref{chebyshev:ratio} for general $p\in[\frac23,2]$, is discussed next, in \sectionref{lp-lewis-weights}.

\subsection{Bounding the Lewis Weights for \texorpdfstring{\(p\in(\frac23,2)\)}{p in (2/3,2)}}
\label{sec:lp-lewis-weights}
To generalize the Lewis weight analysis for \(p=1\), we find a different orthogonal polynomial that nearly achieves the \(C-\)almost Lewis weight property.
We turn to Jacobi Polynomials:
\begin{definition}
The normalized Jacobi Polynomial of degree \(d\) with parameters \(\alpha\) and \(\beta\), denoted \(J_d^{(\alpha,\beta)}\), defines the polynomials orthogonal with \(\int_{-1}^{1} J_i^{(\alpha,\beta)}(t) J_j^{(\alpha,\beta)}(t) (1-t)^\alpha (1+t)^\beta = \mathbbm1_{[i=j]}\).
\end{definition}
In particular, we look at the subclass of \emph{Gegenbauer/Ultraspherical polynomials} which have \(\alpha=\beta\), so we use the truncated notation \(J_d^{(\alpha)}\) and note they are orthogonal with\(\int_{-1}^{1} J_i^{(\alpha)}(t) J_j^{(\alpha)}(t) (1-t^2)^\alpha = \mathbbm1_{[i=j]}\).
Note that Legendre polynomials coincide with \(\alpha=0\), while Chebyshev polynomial of the second kind coincide with \(\alpha=\frac12\), so this class of polynomials certainly interpolates between the \(p=1\) and \(p=2\) orthogonal polynomials.
We now show that Gegenbauer polynomials are the correct orthogonal polynomial for \(L_p\) Lewis weights:
\begin{theorem}
For all \(p \in [\frac23,2]\) and \(\abs{t} \leq 1-O(\frac1{d^2})\), we have \(\frac1{C_0} \leq \frac{\tau[\cV^{\frac12-\frac1p}\cP](t)}{v(t)} \leq C_0\) for some universal constant \(C_0\).
\end{theorem}
\begin{proof}
We first show that fixing \(\alpha=\frac1p-\frac12\) and letting \mU be the corresponding change-of-basis matrix makes \(\cV^{\frac12-\frac1p}\cP\mU\) have orthogonal columns:
\begin{align*}
	\normof{\cV^{\frac12-\frac1p}\cP\mU\vx}_2^2
	&= \int_{-1}^1 \left( \left(\frac{d+1}{\pi\sqrt{1-s^2}}\right)^{\frac12 - \frac1p} \sum_{i=0}^d x_i J_i^{(\alpha)}(s) \right)^2 ds \\
	&= \left(\frac{d+1}{\pi}\right)^{1-\frac2p} \int_{-1}^1 \left(\sum_{i=0}^d x_i J_i^{(\alpha)}(s) \left((1-s^2)^{-\frac12}\right)^{\frac12 - \frac1p}\right)^2 ds \\
	&= \left(\frac{d+1}{\pi}\right)^{1-\frac2p} \sum_{i=0}^d \sum_{j=0}^d x_ix_j \int_{-1}^1 J_i^{(\alpha)}(s) \ J_j^{(\alpha)}(s) \ (1-s^2)^{(\frac1p-\frac12)} ds \\
	&= \left(\frac{d+1}{\pi}\right)^{1-\frac2p} \sum_{i=0}^d \sum_{j=0}^d x_ix_j \mathbbm1_{[i=j]} \\
	&= \left(\frac{d+1}{\pi}\right)^{1-\frac2p} \normof\vx_2^2
\end{align*}
and so we can reduce \(\tau[\cV^{\frac12-\frac1p}\cP](t)\) to a squared row-norm:
\begin{align*}
	\tau[\cV^{\frac12 - \frac1p}\cP](t)
	&= \max_{\vx\in\bbR^{d+1}} \frac{([\cV^{\frac12 - \frac1p} \cP\vx](t))^2}{\normof{\cV^{\frac12-\frac1p}\cP\vx}_2^2} \\
	&= (\tsfrac{\pi}{d+1})^{1-\frac2p} \max_{\normof\vx_2=1} ([\cV^{\frac12 - \frac1p} \cP\vx](t))^2 \\
	&= (\tsfrac{\pi}{d+1})^{1-\frac2p} (\tsfrac{d+1}{\pi\sqrt{1-t^2}})^{1-\frac2p} \max_{\normof\vx_2=1} ([\cP\vx](t))^2 \\
	&= (1-t^2)^{-(\frac12 - \frac1p)} \sum_{i=0}^d (J_i^{(\alpha)}(t))^2 \\
\end{align*}
Unlike the \(p=1\) case, we are not aware of any way to simplify this sum of squares exactly, so we instead provide nearly matching upper and lower bounds.
For the upper bound, Theorem 1 from \cite{nevai1994generalized} says that \((J_i^{(\alpha)}(t))^2 \leq \frac{C_\alpha}{\pi} \cdot (1-t^2)^{-(\alpha+\frac12)}\).
We then bound
\begin{align*}
    \tau[\cV^{\frac12 - \frac1p}\cP](t)
	&\leq (1-t^2)^{-(\frac12 - \frac1p)} \sum_{i=0}^d \frac{C_\alpha}{\pi} (1-t^2)^{-(\alpha + \frac12)} \\
	&= (1-t^2)^{-(\frac12 - \frac1p)} \sum_{i=0}^d \frac{C_\alpha}{\pi} (1-t^2)^{-\frac1p} \\
	&= (1-t^2)^{-\frac12} (d+1) \frac{C_\alpha}{\pi} \\
	&= C_\alpha \frac{d+1}{\pi\sqrt{1-t^2}} \\
	&= C_\alpha v(t)
\end{align*}
To achieve the lower bound, we appeal to a different form of an orthogonal polynomial guarantee.
We rephrase \(\tau[\cV^{\frac12 - \frac1p}\cP](t)\) in terms of the Generalized Christoffel Function \(\lambda_d(z,2,t) \defeq \min_{q:\text{deg}(q)\leq d} \frac{\int_{-1}^1 (q(s))^2 z(s) ds}{(q(t))^2}\), where \(z(s)\defeq(1-s^2)^{\frac1p-\frac12}\), as defined in Equation 1.5 of \cite{erdelyi1992generalized}.
\begin{align*}
	\tau[\cV^{\frac12 - \frac1p}\cP](t)
	&= \max_{\vx\in\bbR^{d+1}} \frac{([\cV^{\frac12-\frac1p}\cP\vx](t))^2}{\normof{\cV^{\frac12-\frac1p}\cP\vx}_2^2} \\
	&= (\tsfrac\pi{d+1})^{1-\frac2p} (v(t))^{1-\frac2p} \max_{q:\text{deg}(q)\leq d} \frac{(q(t))^2}{\int_{-1}^1 \left(q(s)\right)^2 z(s) ds} \\
	&= (\tsfrac\pi{d+1})^{1-\frac2p} (\tsfrac{d+1}{\pi})^{1-\frac2p} (1-t^2)^{\frac{-1}2(1-\frac2p)} \max_{q:\text{deg}(q)\leq d} \frac{(q(t))^2}{\int_{-1}^1 \left(q(s)\right)^2 z(s) ds} \\
	&= (1-t^2)^{\frac1p - \frac12} \frac1{\min_{q:\text{deg}(q)\leq d} \frac{\int_{-1}^1 \left(q(s)\right)^2 z(s) ds}{(q(t))^2}} \\
	&= (1-t^2)^{\frac1p - \frac12} \frac1{\lambda_d(z,2,t)}
\end{align*}
In \appendixref{christoffel-bound} we adapt Theorem 2.1 of \cite{erdelyi1992generalized} to show that \(\lambda_d(z,2,t) \leq \frac{C}{d-1}(1-t^2)^{\frac1p}\) for some universal constant \(C\) when \(\abs{t}\leq1-O(\frac1{d^2})\).
With this bound, we get \(\tau[\cV^{\frac12 - \frac1p}\cP](t) \geq (1-t^2)^{-\frac12} \frac{d-1}{C}\), so we can show the lower bound required by the almost Lewis weight property:
\[
	\frac{\tau[\cV^{\frac12 - \frac1p}\cP](t)}{v(t)}
	\geq \frac{(1-t^2)^{- \frac12}  \frac{d-1}{C}}{(1-t^2)^{-\frac12} \frac{d+1}{\pi}}
	= \frac{\pi(d-1)}{C(d+1)}
	\geq \frac{\pi}{3C}
\]
And so, we find that \(\frac{\pi}{3C} \leq \frac{\tau[\cV^{\frac12 - \frac1p}\cP](t)}{v(t)} \leq C_\alpha\), completing the proof.
\end{proof}
Again, we see that the Chebyshev measure satisfies the almost Lewis weight property for most \(t\in[-1,1]\), but this does not work in the endcaps.
To remedy this issue, we again appeal to the clipped Chebyshev measure, resulting in \theoremref{chebyshev:ratio}
\begin{reptheorem}{chebyshev:ratio}
There are universal constants $c_1,c_2$ such that, for all $p\in[\frac23,2]$ and $t \in [-1,1]$,
\[
	\frac{c_1}{\log^3 d}\le\frac{\tau[\cW^{1/2-1/p}\cP](t)}{w(t)}\le c_2.
\]
\end{reptheorem}
\noindent
The full proof using this clipped measure is deferred to \sectionref{clipped-analysis}.

\subsection{Constant-Factor Approximation}
\label{sec:smallp:regression:constant}

In order to achieve a constant-factor approximation to the \(L_p\) polynomial regression problem, we want to use \theoremref{chebyshev:ratio} to create a subspace embedding guarantee.
However, as discussed in \sectionref{two-stage}, Lewis weight guarantees have a logarithmic dependence on the number of rows of the full matrix, which is infinite for \cP.

Beyond Lewis weight sampling, it is known that matrix \(L_p\) sensitivity sampling can be done with a suboptimal dependence on the dimension \(d\), but without any dependence on the number \(m\) of rows within the analysis.
So, we bound the \(L_p\) sensitivity function of \cP, showing that \(\tilde O(d^5)\) samples drawn \emph{uniformly} from \([-1,1]\) creates a subspace embedding from the \cP operator to a tall-and-skinny matrix \mA.
With this sensitivity result, we can solve the problem in \theoremref{main} with \(\tilde O(d^5)\) samples:

\begin{definition}[\(L_p\) sensitivity function]
For an operator \(\cA:\bbR^{d+1}\rightarrow L_p([-1,1])\), the \(L_p\) sensitivity function for \cA at time \(t\in[-1,1]\) is
\[
    \psi_p[\cA](t) \defeq \max_{\vx\in\bbR^{d+1}} \frac{\abs{[\cA\vx](t)}^p}{\normof{\cA\vx}_p^p}.
\]
\end{definition}
We show that the sensitivities of \(L_p\) regression are bounded.
\begin{lemma}[Uniform sensitivity bound]
\label{lem:uniform-sensitivity-bound}
For all \(t\in[-1,1]\) and \(p \geq 1\), we have \(\psi_p[\cP](t) \leq d^2(p+1)\)
\end{lemma}
\begin{proof}
Note that \(\psi_p[\cP] \defeq \max_{\vx\in\bbR^{d+1}} \frac{\abs{[\cP\vx](t)}^p}{\normof{\cP\vx}_p^p} = \max_{q:\text{deg}(q)\leq d} \frac{\abs{q(t)}^p}{\int_{-1}^1 \abs{q(s)}^p ds}\)
Without loss of generality we take \(q(t)=1\).
Let \(C \defeq \max_{s\in[-1,1]} \abs{q(x)}\) and \(s^* \defeq \argmax_{s\in[-1,1]} \abs{q(x)}\).
By the Markov brothers' inequality, we have \(\abs{q(s^* + s)} \geq C - Cd^2s \geq 0\) for any \(\abs{s} \leq \frac1{d^2}\).
Then we can lower bound the integral in the denominator of \(\psi_p\) by
\[
    \int_{-1}^1 \abs{q(s)}^p\,ds
    \geq \int_{0}^{\frac1{d^2}} (C - C d^2 s)^p\,ds
    = \frac{-1}{Cd^2(p+1)}(C-Cd^2x)^{p+1} \bigg|_0^{1/d^2}
    \geq \frac{1}{d^2(p+1)}
\]
so that
\[
    \psi_p[\cP](t)
    = \frac{\abs{q(t)}^p}{\int_{-1}^1\abs{q(x)}^p\,dx}
    \leq d^2(p+1)
\]
\end{proof}
Next we show that since uniform sampling is oversampling with respect to the sensitivities, we can get an \(L_p\) subspace embedding with \(\tilde O(d^5)\) samples:
\begin{theorem}
\label{thm:sensitivity-correctness-Lp}
Let \(p \geq 1\) and suppose \(s_1,\ldots,s_{n_0}\) are drawn uniformly from \([-1,1]\). 
Let \(\mA\in\bbR^{n_0 \times (d+1)}\) be the associated Vandermonde matrix, so that \(\mA_{i,j} = s_i^{j-1}\).
Let \(\vb\in\bbR^{n_0}\) be the evaluations of \(f\), so that \(\vb(i) = f(s_i)\).
For \(n_0 = O\left(d^5 p^2 2^p \log d\right)\), there exists a universal constant \(c\) such that the sketched solution \(\vx_c = \argmin_{\vx} \normof{\mA\vx-\vb}_p\) satisfies
\[
    \normof{\cP\vx_c-f}_p \leq c \min_{\vx\in\bbR^{d+1}}\normof{\cP\vx-f}_p
\]
with probability at least \(\frac{11}{12}\). 

\noindent Further, let \(\eps\in(0,1)\) and suppose \(\normof f_p \leq C \, \min_{\vx}\normof{\cP\vx-f}_p\).
If \(n_0 = O\left(\frac{1}{\eps^{O(p^2)}}\,d^5p^{O(p)}\log\frac{d}{\eps}\right)\), then
\[
    \normof{\cP\hat\vx - f}_p^p \leq (1+\eps) \min_{\vx} \normof{\cP\vx-f}_p^p
\]
with probability at least \(\frac {11}{12}\).
In particular, suppose \(\vx_c\) is computed from sampling \(f\) uniformly at least \(O(d^5 p^2 2^p \log(d))\)times, we let \(\hat f(t) \defeq f(t) - [\cP\vx_c](t)\), and compute \(\hat\vx\) by sampling \(\hat f\) uniformly at least \(O\left(\frac{1}{\eps^{O(p^2)}}\,d^5p^{O(p)}\log\frac{d}{\eps}\right)\) times.
Then, if we let \(\tilde \vx \defeq \vx_c + \hat\vx\), we have
\[
    \normof{\cP\tilde\vx - f}_p^p \leq (1+\eps) \min_{\vx\in\bbR^{d+1}} \normof{\cP\vx-f}_p^p
\]
\end{theorem}
The proof of this theorem is a standard sensitivity sampling analysis combined with our bounds on the $L_p$ sensitivities, so it is deferred to \appendixref{sensitivity-sampling}.

To decrease this sample complexity further, we apply Lewis weight subsampling to the matrix \mA.
Since the rows of \mA are drawn \emph{uniformly} from \([-1,1]\), we can show that the Lewis weights of \mA closely approximate the Lewis weights of \cP.
So, by \theoremref{chebyshev:ratio}, we know that the Chebyshev measure upper bounds the Lewis weights of \mA.
That is, \textit{we can bound the Lewis weights of \mA without ever even building the matrix}.
Formally, we give the following guarantee:

\begin{theorem}
\label{thm:sens:lewis:approx}
Let \mA, and \(n_0\) as in either part of \theoremref{sensitivity-correctness-Lp}.
Then, with probability \(\frac{11}{12}\), for all \(i\in[n_0]\), the \(\ell_p\) Lewis weight of \mA at row \(i\) is at most \(\frac{1}{n_0} v(s_i) \polylog(d)\) and at least \(\frac{1}{n_0 \polylog(d)} w(s_i)\).
\end{theorem}
\begin{proof}
Let \(\mW\in\bbR^{n_0 \times n_0}\) be a diagonal matrix that represents our candidate \(\ell_p\) Lewis weights for \mA, with \(\mW_{ii} \defeq \gamma w(s_i)\), where \(\gamma \defeq \frac{2}{n_0}\) is a rescaling factor.
In \appendixref{reweighted-subspace-embedding} we use a standard $\epsilon$-net argument to show the spectral approximation
\[
    \frac12 \cP^\top\cW^{1-\frac2p}\cP \preceq \gamma^{-\frac2p}\mA^\top\mW^{1-\frac2p}\mA \preceq 2\cP^\top\cW^{1-\frac2p}\cP
\]
holds with probability \(\frac{11}{12}\).
We condition on this event.

Then note the inner-product form of the leverage scores: \(\tau[\mA]_i = \va_i^\top(\mA^\top\mA)^{-1}\va_i\) where \(\va_i\) is the \(i^{th}\) row of \mA, and \(\tau[\cP](t) = \vp_t^\top(\cP^\top\cP)^{-1}\vp_t\) where \(\vp_t \defeq [1~t~t^2~\ldots~t^d]\) is the row of \cP at time \(t\) (Theorem 5 from \cite{AvronKapralovMusco:2019} or Lemma 1 from \cite{meyer2022basic}).
Then we can examine the rescaled leverage scores:
\begin{align*}
    \tau[\mW^{\frac12-\frac1p}\mA](i)
    &= (\mW_{ii})^{1-\frac2p} \va_i^\top(\mA^\top\mW^{1-\frac2p}\mA)^{-1}\va_i \\
    &\leq 2 (\mW_{ii})^{1-\frac2p} \gamma^{\frac2p} \va_i^\top(\cP^\top\cW^{1-\frac2p}\cP)^{-1}\va_i \\
    &= 2 (\gamma w(s_i))^{1-\frac2p} \gamma^{\frac2p} \vp_{s_i}^\top(\cP^\top\cW^{-1}\cP)^{-1}\vp_{s_i} \\
    &= 2 \gamma\ \tau[\cW^{-1}\cP](s_i)
\end{align*}
and we can similarly show that \(\tau[\mW^{\frac12-\frac1p}\mA](i) \geq \frac12 \, \gamma \tau[\cW^{\frac12-\frac1p}\cP](s_i)\).
So now we can use \theoremref{chebyshev:ratio} to show the almost Lewis weight property holds on \mA:
\[
    \frac{\tau[\mW^{\frac12-\frac1p}\mA](i)}{\mW_{ii}}
    \leq \frac{2\gamma\ \tau[\cW^{\frac12-\frac1p}\cP](s_i)}{\gamma\ w(s_i)}
    = 2\frac{\tau[\cW^{\frac12-\frac1p}\cP](s_i)}{w(s_i)}
    \leq \log^3(d)
\]
and similarly we can show the lower bound \(\frac{\tau[\mW^{\frac12-\frac1p}\mA](i)}{\mW_{ii}} \geq \log^3(d)\).
Therefore, \(\mW_{ii} = \frac{2}{n_0} w(t)\) are \(\ell_p\) almost Lewis weights for \mA.
Further, since \(v(t) \geq w(t)\), we have that \(\frac{C}{n_0} v(t)\) upper bound the \(\ell_p\) Lewis weights for \mA for some constant \(C\).
\end{proof}

This na\"ively suggests an $\tilde{O}(d^5)$ runtime algorithm to pick \(O(d \polylog d)\) samples that give optimal \(L_p\) regression: sample $n_0 = O(d^5\log d)$ times uniformly from \([-1,1]\), and for each sample, throw it away with probability \(1-\min\{\frac{1}{n_0} v(s_i) \polylog(d), 1\}\).
Then, with high probability, \(O(d)\) samples remain and the resulting subsampled matrix is an \(L_p\) subspace embedding. 
Formally, this argument uses the following result from \cite{CohenP15}:

\begin{theorem}[Theorem 7.1 from \cite{CohenP15}\footnote{This is stated with a slightly different sampling method as in \cite{CohenP15}, but the theorem holds by applying standard methods to their Rademacher analysis}]
\label{thm:lewis-weight-accept-reject-l1}
Let \(\mA\in\bbR^{n_0 \times d+1}\) and \(p\in[1,2]\).
Let \(w_p[\mA](1),\ldots,w_p[\mA](n_0)\) be the \(\ell_p\) Lewis weights of \mA, and let \(\tilde w_i \geq C w_p[\mA](i)\) for all \(i\) such that \(\sum_i \tilde w_i = \tilde O(d)\).
Define probabilities \(p_i \defeq \min\{1, \frac{m}{n_0} \tilde w_i\}\), and build the diagonal matrix \(\mS\in\bbR^{n_0 \times n_0}\) such that \(\mS_{ii}\) takes value \(\frac{1}{(p_i)^{1/p}}\) with probability \(p_i\) and is \(0\) otherwise.
Remove the rows of \mS that are all zero.
Suppose we pick \(m\), the expected number of remaining rows, to be \(m = O(d \polylog(d))\).
Then with probability \(\frac{11}{12}\), for all \(\vx\in\bbR^{d+1}\), we have \(\normof{\mS\mA\vx}_p \approx_2 \normof{\mA\vx}_p\).
\end{theorem}

This \(\tilde{O}(d^5)\) time algorithm certainly suffices to give the near-optimal sample complexity for constant \(\eps\), but we can improve the time complexity.
In particular, since we exactly know the distribution of \(s_1,\ldots,s_{n_0}\) and the probabilities of the coins \(p_1,\ldots,p_{n_0}\), we can directly compute the marginal distribution of times that result from both sampling procedures:

\begin{lemma}
\label{lem:stage-squashing}
Suppose \(n_0\) time samples are drawn uniformly from \([-1,1]\), and each sample is thrown away with probability \(1-\min\{\frac{m}{n_0} \frac{1}{\sqrt{1-s_i^2}}, 1\}\).
Let \(n\) denote the number of remaining samples.
Then \(n\) is distributed as \(B(n_0, O(\frac{m}{n_0}))\), and with probability \(\frac{99}{100}\) the resulting samples cannot be distinguished from iid samples from the Chebyshev measure.
\end{lemma}

This short lemma is proven in \appendixref{stage-squashing}.
Taking \(n_0 = O(d^5 \polylog d)\) and \(m = O(d \polylog d)\), we get \(n \sim B(n_0, 1/\tilde{O}(d^4))\) so that \(n=d\,\polylog d\) with very high probability.
So, this lemma tells us that instead of sampling \(\tilde{O}(d^5)\) times uniformly, we can just sample \(d\,\polylog(d)\) samples from the Chebyshev distribution.
In summary, we arrive at the following:
\begin{corollary}
\label{corol:Lp-subspace-embedding}
Let \mA, and \(n_0\) as in either part of \theoremref{sensitivity-correctness-Lp}.
Let \(m = O(d \polylog d)\).
Suppose an algorithm samples \(n \sim B(n_0, O(\frac{m}{n_0}))\) and runs \algorithmref{chebyshev-const:Lp}.
Then, the matrix \(\mS\mA\) on line 4 of the algorithm is a subspace embedding for \(\cP\): \(\frac1C\normof{\cP\vx}_p^p \leq \normof{\mS\mA\vx}_p^p \leq C \normof{\cP\vx}_p^p\) for all \(\vx\in\bbR^{d+1}\).
\end{corollary}
We now state the overall correctness of the algorithm for constant factor approximation for \(p\geq1\):
\begin{theorem}
\label{thm:constant:approx}
Let \(p\geq1\) and \(n_0 = O\left(d^5 p^2 2^p \log d\right)\).
Suppose an algorithm samples \(n \sim B(n_0, 1/\tilde{O}(d^4))\) and runs \algorithmref{chebyshev-const:Lp}.
Then, with probability \(\frac23\), the resulting polynomial \(\hat q\) satisfies
\[
    \normof{\hat q - f}_p^p \leq O(1) \min_{q:\text{deg}(q)\leq d} \normof{q-f}_p^p
\]
\end{theorem}
The correctness of this theorem follows from combining \corolref{Lp-subspace-embedding} with Lemma A.1 from \cite{meyer2021}, which says that unbiased subspace embedding suffices for constant-factor error in regression.
While there is randomness in the sample complexity, we have that with very high probability \(n = O(d \polylog d)\).
Finally, we emphasize that \theoremref{constant:approx} holds for all $p\ge 1$ due to the result from \cite{meyer2021}. 
Thus we will ultimately also use this algorithm as a subroutine for $L_p$ polynomial regression for $p\ge 2$. 

\subsection{\texorpdfstring{$(1+\eps)$}{1+eps}-Approximation}
\label{sec:smallp:regression:eps}

Given the constant factor approximation in the previous section, we can now build an algorithm that outputs a $(1+\eps)$-approximation for the $L_p$ regression problem when $p\in[1,2]$. 
First, we recall an algorithm from \cite{MuscoMWY21} that samples \(d \poly(\log d, \frac1\eps)\) rows of a matrix by almost-Lewis weights, reads the corresponding coordinates in the measurement vector \vb, and solves the subsampled \(\ell_p\) matrix regression problem twice, giving a \((1+\eps)\) error solution.
Since we know that the Chebyshev density describes the almost-Lewis weights of \mA, we can directly appeal to this result.
In particular, they prove that \algorithmref{constant:active:Lp:eps} gives the following guarantee:

\begin{theorem}
\label{thm:active:lewis:norep}
Let \(\mA\in\bbR^{m \times d+1}\), \(\vb\in\bbR^m\), and \(p \geq 1\).
Then, with probability \(0.98\), \algorithmref{constant:active:Lp:eps} with \(n = O(d^{\max(1,p/2)} \frac{\log^2(d) \log(m)}{\eps^{\min(2p+5,p+7)}})\) returns a vector \(\tilde\vx\in\bbR^{d+1}\) such that \(\normof{\mA\tilde\vx-\vb}_p \leq (1+\eps) \min_{\vx}\normof{\mA\vx-\vb}_p\).
\end{theorem}

We remark that although \theoremref{active:lewis:norep} matches the guarantee given by Theorem \(3.4\) in \cite{MuscoMWY21}\footnote{This is following the first version of \cite{MuscoMWY21} uploaded to arXiv, which uses an analysis which makes especially simple to see how Bernstein suffices for either sampling scheme.}, \algorithmref{constant:active:Lp:eps} does not quite match the corresponding Algorithm 2 given by \cite{MuscoMWY21}.
Observe that each row is sampled without replacement with probability proportional to its Lewis weight in \algorithmref{constant:active:Lp}, whereas a fixed number of rows are sampled by \cite{MuscoMWY21}, so that each row is sampled with replacement with probability proportional to its Lewis weight.
However, the correctness of \algorithmref{constant:active:Lp} follows from the analysis of Theorem \(3.4\) in \cite{MuscoMWY21} by zooming into Claim \(3.14\) and just using a sampling matrix \mS defined by without-replacement sampling instead of the with-replacement matrix used.
None of the concentrations actually change at the end of the day.
We show an example of such a Bernstein bound later in this paper, in the proof of \lemmaref{preserve-compact-roundings}.

\begin{algorithm}[t]
\caption{Constant factor active \(\ell_p\) matrix regression}
\label{alg:constant:active:Lp}
\begin{algorithmic}[1]
\Require{Vandermonde matrix \(\mA\in\bbR^{n_0 \times d+1}\), response vector \(\vb\in\bbR^{n_0}\), target number of samples \(m\)}
\Ensure{Approximate solution \(\hat\vx\in\bbR^{d+1}\) to \(\min_{\vx}\normof{\mA\vx-\vb}_p\)}
\State{Let \(p_i = \min\{1, \frac{m}{n_0} \frac{1}{\sqrt{1-s_i^2}}\}\) where \(s_i\in[-1,1]\) is the time associated with row \(i\) of \mA}
\State{Let \(\mS\in\bbR^{n_0 \times n_0}\) be a diagonal matrix with \(S_{ii} = \frac{1}{(p_i)^{1/p}}\) with probability \(p_i\), and \(S_{ii}=0\) otherwise}
\State{\Return \(\hat\vx=\argmin_{\vx}\normof{\mS\mA\vx-\mS\vb}_p\)}
\end{algorithmic}
\end{algorithm}

\begin{algorithm}[t]
\caption{Relative error active \(\ell_p\) matrix regression}
\label{alg:constant:active:Lp:eps}
\begin{algorithmic}[1]
\Require{Matrix \(\mA\in\bbR^{n_0 \times d+1}\), response vector \(\vb\in\bbR^{n_0}\), target number of samples \(n\)}
\Ensure{Approximate solution \(\tilde\vx\in\bbR^{d+1}\) to \(\min_{\vx}\normof{\mA\vx-\vb}_p\)}
\State{Run \algorithmref{constant:active:Lp} on vector \(\vb\) with \(\frac n2\) samples to get vector \(\vx_c\)}
\State{Let \(\vz \defeq \vb-\mA\hat\vx_c\)}
\State{Run \algorithmref{constant:active:Lp} on vector \(\vz\) with \(\frac n2\) samples to get vector \(\hat\vx\)}
\State{\Return \(\tilde\vx=\vx_c + \hat\vx\)}
\end{algorithmic}
\end{algorithm}

Overall, \theoremref{active:lewis:norep} show that \algorithmref{constant:active:Lp:eps} finds a near-optimal solution to the uniform-sampled problem for \(p\in[1,2]\).
By the reduction from two-stage to one-stage sampling, this then implies that \algorithmref{chebyshev-const:Lp:eps} finds a near-optimal solution to the \(L_p\) polynomial regression problem.
So, we have now proven our \(L_p\) polynomial approximation guarantee for \(p\in[1,2]\):

\begin{reptheorem}{main}[For $1\le p\le 2$]
For any degree $d$, $p \in[1,2]$, and accuracy parameter $\eps\in(0,1)$, there is an algorithm that queries $f$ at $n = O(\frac{d}{\eps^{2p+5}} \polylog(\frac d\eps))$ points $t_1, \ldots, t_n$, each selected independently at random according to the Chebyshev density on $[-1,1]$, and outputs a degree $d$ polynomial $\hat{q}(t)$ such that, with probability at least $0.9$,
\[\|\hat{q}(t)-f(t)\|_p^p\le(1+\eps)\cdot\min_{q:\deg(q)\le d}\|q(t)-f(t)\|_p^p.\]
\end{reptheorem}

\section{Active \texorpdfstring{\(L_p\)}{Lp} Regression for \texorpdfstring{$p>2$}{p>2}}
\label{sec:lp-regression:largep}

In this section, we analyze $L_p$ regression for $p>2$. Our analysis differs significantly from the case of $p\in[1,2]$. In particular, while we still analyze sampling by the Chebyshev measure, in contrast to $p\in [1,2]$, we are not able to argue that the measure approximates the $L_p$ Lewis weights of the polynomial operate $\cP$. Moreover, even if we could bound them, sampling by $L_p$ Lewis weights requires \(O(d^{p/2})\) samples in the worst case to approximate a $p$-norm regression problem \cite{MuscoMWY21}.
There are matrices which require this rate, so to get sample complexity linear in \(d\), we will leverage special structure of polynomials that lets us avoid these worst-case instances.

We start with a simple but useful observation from \cite{meyer2021}. Ssuppose \(f(t)\) is a polynomial of degree \(d\), and let \(r \approx p\) be an integer with \(q \defeq \frac pr\in[1,2]\).
Then, we know that \(t \mapsto (f(t))^r\) is a degree \(rd\) polynomial.
Since \mA is a Vandermonde matrix, and letting \vx be the coefficient vector for \(f\), we thus have that
\begin{align*}
	\normof{\mA\vx}_p^p = \normof{\mB\vy}_q^q,
\end{align*}
where \(\mB\in\bbR^{n_0 \times rd+1}\) is a Vandermonde matrix generated by the same time points as \mA but with more columns, and where \(\vy\) is the coefficient vector for the degree \(rd\) polynomial \(t \mapsto (f(t))^r\).
This simple observation implies that if some sampling procedure preserves the \(\ell_q\) norm of all degree \(rd\) polynomials, then that sampling procedure also preserves the \(\ell_p\) norm of all degree \(d\) polynomials.
In other words, it suffices to use a sampling matrix \mS that samples rows of \mB with probability proportional to upper bounds on the \(\ell_q\) Lewis weights of \mB.
By \theoremref{sens:lewis:approx} we already know those Lewis weights are bounded by the Chebyshev measure.
So \algorithmref{constant:active:Lp}, which samples rows of \mA by the Chebyshev measure, preserves the \(\ell_p\) norm of \(\mA\vx\) for all \vx \emph{because} it is sampling rows by the \(\ell_q\) Lewis weights of \mB.

This argument suffices to get prove a subspace embedding result -- i.e., that the matrix \mS from \algorithmref{constant:active:Lp} satisfies \(\normof{\mS\mA\vx}_p^p \approx_C \normof{\mA\vx}_p^p\).
This is sufficient to get a constant-factor regression solution, and we formally work through this in \sectionref{large-p-const-factor}.
To achieve error \((1+\eps)\), we need a more refined analysis that builds on the first version of \cite{MuscoMWY21} uploaded to arXiv\footnote{While that version is available on arXiv at time of publishing, since it is unpublished, we include a (slightly shortened and corrected) copy of everything we use in \appendixref{compact-rounding}.}. Our approach still reduces from the general $p > 2$ case to some $q \leq 2$, but in a less direct way than described above. An edge case of our analysis requires that when \(p\in(2,3)\), we use \(r=3\) so that \(q = \frac pr \in [\frac23,1]\).
This is the case where we use the \(\ell_q\) Lewis weight bounds for \(q < 1\).

\subsection{Constant Factor Approximation for \texorpdfstring{\(p>2\)}{p > 2}}
\label{sec:large-p-const-factor}

We start by showing that running \algorithmref{constant:active:Lp} as done in line 1 of \algorithmref{constant:active:Lp:eps} achieves a constant-factor regression guarantee.
Formally, we rely on the following result from \cite{meyer2021}, where \(\psi_p[\mA](i) \defeq \max_{\vx} \frac{\abs{[\mA\vx](i)}^p}{\normof{\mA\vx}_p^p}\) is the \(\ell_p\) sensitivity score of \mA at row \(i\):

\begin{theorem}
\label{thm:vander:p:q:sens}
\cite{meyer2021}
Given \(p>2\), let \(r\) be any integer such that \(q \defeq \frac pr\) is in \([\frac23,2]\).
Given a Vandermonde matrix \(\mA\in\bbR^{n_0\times d+1}\), let \mB be the Vandermonde matrix \(\mA\) extended to have \(rd+1\) columns. 
Then for every vector \(\vx\in\bbR^{d+1}\), there exists a vector \(\vy\in\bbR^{rd+1}\) such that \(\abs{[\mA\vx](i)}^p=\abs{[\mB\vy](i)}^q\). 
Thus if \(\psi_p[\mA](i)\) denotes the \(\ell_p\)-sensitivity of the \(i\)-th row of \mA and \(\psi_q[\mB](i)\) denotes the \(\ell_q\)-sensitivity of the \(i\)-th row of \mB, then \(\psi_p[\mA](i) \leq \psi_q[\mB](i)\). 
\end{theorem}

For the constant-factor approximation step, given \(p>2\), we let \(r\) be an integer such that \(r \leq p < 2r\), so that \(q \defeq \frac pr \in [1,2]\).
With this choice of \(r\), chosen such that \(\ell_q\) is a valid norm that satisfies the triangle inequality, we will show that \algorithmref{constant:active:Lp}, as run in the first line of \algorithmref{constant:active:Lp:eps}, returns a constant factor solution to \(\min_{\vx} \normof{\mA\vx-\vb}_p\).
Recall that \(\mA\in\bbR^{n_0 \times d+1}\) is a Vandermonde matrix obtained by uniformly sampling \(n_0 = \poly(d,p^p,\frac1{\eps^p})\) points from \([-1,1]\).
Then let \(\mB\in\bbR^{n_0 \times rd+1}\) be an \emph{expanded} Vandermonde matrix, built using the same uniform samples but with maximum degree \(rd\).
Let \(d_B \defeq rd+1\) be the number of columns in \mB.
We also let \(w_q[\mB](i)\) be the \(\ell_q\)-Lewis weight of \mB at row \(i\).
We will analyze sampling rows of \mA with respect to \(w_q[\mB](i)\).

We first show that the sampling matrix \mS from \algorithmref{constant:active:Lp} is a subspace embedding:
\begin{lemma}
\label{lem:lp-subspace-embed}
Let \mA and \mS be the matrices as in \algorithmref{constant:active:Lp}.
Then, with probability \(\frac{99}{100}\), so long as \(m = O(\frac{pd}{\eps^2} \polylog(d))\), we have that \mS is an \(\ell_p\) subspace embedding:
\[
	\normof{\mS\mA\vx}_p^p \in (1\pm\eps) \normof{\mA\vx}_p^p \hspace{2cm} \forall \vx \in \bbR^{d+1}
\]
\end{lemma}
\begin{proof}
Recall \theoremref{vander:p:q:sens}, in particular that for any \(\vx\in\bbR^{d+1}\), there exists a vector \(\vy\in\bbR^{rd+1}\) such that \(([\mB\vy](i))^q = ([\mA\vx](i))^p\) for all \(i\in[n_0]\).
We then expand the subspace embedding norm:
\[
	\normof{\mS\mA\vx}_p^p
	= \sum_{i=1}^{n_0} \mS_{ii}^p \abs{[\mA\vx](i)}^p
	= \sum_{i=1}^{n_0} \frac{1}{p_i} \abs{[\mB\vy](i)}^q
	= \normof{\bar\mS\mB\vy}_q^q
\]
where \(\bar\mS_{ii} = (\mS_{ii})^{p/q} = \frac{1}{(p_i)^{1/q}}\) is the sampling matrix we would use when sampling \mB by \(\ell_q\) Lewis weights.
So, we not only have \(\normof{\mA\vx}_p^p = \normof{\mB\vy}_q^q\), but also have \(\normof{\mS\mA\vx}_p^p = \normof{\bar\mS\mB\vy}_q^q\).
Then we are sampling by overestimates of the Lewis weights, since \(w_q[\mB](i) \leq \frac{1}{n_0} \frac{rd+1}{\sqrt{1-s_i^2}} \polylog(d) \leq \frac{m}{n_0} \frac{1}{\sqrt{1-s_i^2}} = p_i\), which holds for \(m \geq d \polylog(d)\).
So, by \theoremref{lewis-weight-accept-reject-l1}, we have that \(\mS\) is a \((1\pm\frac12)\) \(\ell_q\)-subspace embedding for \(\mB\) so long as \(m = O(\frac{rd}{\eps^2} \polylog(d))\), and therefore that \(\bar\mS\) is a \((1\pm\eps)\) \(\ell_p\)-subspace embedding for \(\mA\).
\end{proof}
\begin{lemma}
\label{lem:lp-const-factor}
The vector \(\vx_c\) returned by line 1 of \algorithmref{constant:active:Lp} is a constant-factor solution to the overall optimization problem, with probability \(\frac{99}{100}\):
\[
	\normof{\mA\vx_c-\vb}_p \leq C_z \min_{\vx\in\bbR^{d+1}} \normof{\mA\vx-\vb}_p
\]
For some universal constant \(C_z\).
In particular, this implies that \(\vz\) from line 2 of \algorithmref{constant:active:Lp} has \(\normof{\vz}_p \leq C_z \min_{\vx} \normof{\mA\vx-\vb}_p\).
\end{lemma}
\begin{proof}
Recall that \(\vx_c \defeq \argmin_{\vx} \normof{\mS\mA\vx-\mS\vb}_p\), and that \lemmaref{lp-subspace-embed} shows that \(\mS\) is an \(\ell_p\) subspace embedding for \(\mA\).
Let \(\vx^* \defeq \argmin_{\vx} \normof{\mA\vx-\vb}\) be the true optimal regression solution.
Then, by repeated use of the triangle inequality,
\begin{align*}
	\normof{\mA\vx_c - \vb}_p
	&\leq \normof{\mA\vx_c - \mA\vx^*}_p + \normof{\mA\vx^* - \vb}_p \\
	&\leq 2\normof{\mS\mA\vx_c - \mS\mA\vx^*}_p + \normof{\mA\vx^* - \vb}_p \\
	&\leq 2(\normof{\mS\mA\vx_c - \mS\vb}_p + \normof{\mS\mA\vx^* - \mS\vb}_p) + \normof{\mA\vx^* - \vb}_p \\
	&\leq 4\normof{\mS\mA\vx^* - \mS\vb}_p + \normof{\mA\vx^* - \vb}_p
\end{align*}
where the last line follows from the optimality of \(\tilde\vx\).
Then, since \(\E[\normof{\mS\mA\vx^* - \mS\vb}_p^p] = \normof{\mA\vx^*-\vb}_p^p\), by Markov's inequality we bound \(\normof{\mS\mA\vx^* - \mS\vb}_p^p \leq 200 \normof{\mA\vx^*-\vb}_p^p\), and we conclude that
\[
	\normof{\mA\vx_c - \vb}_p \leq 801 \normof{\mA\vx^*-\vb}_p
\]
\end{proof}

\subsection{Relative Error Approximation}
\label{sec:largep:regression:eps}

In this section, we show that the estimator \(\tilde\vx\) recovered on by \algorithmref{constant:active:Lp:eps} is a \((1+\eps)\)-optimal estimator for \(\normof{\mA\vx-\vb}_p^p\).
First, note we assume that \(\eps \leq \frac1p\) in this section, and prove that sampling \(\tilde O(\frac{d \, 2^{O(p)}}{\eps^{6p+2}})\) rows suffices to recover a near-optimal estimator.
If \(\eps > \frac1p\), then we can just run the algorithm when \(\eps = \frac1p\), which yields a \(\tilde O(dp^{O(p)})\) sample complexity, so the sample complexity we promise in \theoremref{main}suffices across all possible \(\eps\in(0,1)\) and \(p \geq 2\).

Much of this section very closely tracks the proof of Theorem \(3.4\) in the first version of \cite{MuscoMWY21} uploaded to arXiv, with the main difference being \lemmaref{trunc} which uses \theoremref{vander:p:q:sens} to define the vector \(\bar{\vz}\) with respect to the \(\ell_q\) Lewis weights of \mB, where the original analysis uses the \(\ell_p\) Lewis weights of \mA.
The core of the novel analysis is used to prove \theoremref{affine-embed}.
While we state and use \theoremref{affine-embed} in this section, we do not prove it until later, in \sectionref{affine-embed}.

Most of this section analyzes the second call to \algorithmref{constant:active:Lp}, from the line 3 of \algorithmref{constant:active:Lp:eps}.
As such, we explicitly write down the notation that will be used throughout most of this section:

\begin{setting}
\label{setting:large-p-eps}
\(\mA\in\bbR^{n_0 \times d+1}\) is a Vandermonde matrix formed by sampling \(n_0 = O(\frac{1}{\eps^{O(p^2)}} d^5 p^{O(p^2)} \log \frac d\eps)\) times \(s_1,\ldots,s_{n_0}\) uniformly at random from \([-1,1]\).
\(r\) is an integer such that \(\frac12p \leq r < \frac32 p\), and \(q \defeq \frac{p}{r} \in [\frac23,2]\).
\(\mB\in\bbR^{n_0 \times d_B}\) is a Vandermonde matrix formed from the same time samples \(s_1,\ldots,s_{n_0}\), but with \(d_B \defeq rd+1\) columns.
\(w_q[\mB](i)\) denotes the \(\ell_q\) Lewis Weight of \mB at row \(i\), and \(\psi_p[\mA](i) \defeq \max_{\vx} \frac{\abs{[\mA\vx](i)}^p}{\normof{\mA\vx}_p^p}\) denotes the \(\ell_p\) sensitivity of row \(i\) of \mA.
\(\vz \defeq \vb - \mA\vx_c\) is the vector generated by line 2 of \algorithmref{constant:active:Lp}.
By \lemmaref{lp-const-factor}, \(\normof{\vz}_p \leq C_z OPT\), where \(OPT = \min_{\vx} \normof{\mA\vx-\vb}_p\).
\(\bar\vz\) is equal to \(\vz\) except that it has several entries zeroed out:
\[
	\bar\vz(i) \defeq \begin{cases}
		\vz(i) & \abs{\vz(i)} \leq \tsfrac{OPT}{\eps} (w_q[\mB](i))^{1/p} \\
		0 & \text{otherwise}
	\end{cases}
\]
Let \(\mS\in\bbR^{n_0 \times n_0}\) be the sample-and-rescale matrix generated in step 3 of \algorithmref{constant:active:Lp:eps} with \(m = O(\frac{d}{\eps^{6.5p+2}} \polylog(\frac d\eps))\).
\(C_0 \defeq 400C_z\) is a large enough constant.
\end{setting}

Note that \(r\) in this section might not be the same value of \(r\) taken in the constant factor analysis of \sectionref{large-p-const-factor}.
We explain this new choice of \(r\) in \sectionref{affine-embed} in full detail, but at a high level, we will eventually want \(r\) to be odd for this analysis to go through, which will sometimes require \(q \in [\frac23,1]\), for instance.

In the majority of this proof, we constrict ourselves to looking at vectors in the range of \mA which are not too much larger than \(OPT\), defining a sort of ``reasonable range of \mA'' to focus on.
Rigorously, this means the upcoming lemmas will only look at vectors in the range of \mA with \(\normof{\mA\vx}_p \leq C_0 OPT\).
We will eventually ensure that both \(\vx^* = \argmin_{\vx} \normof{\mA\vx-\vz}_p\) and \(\hat\vx = \argmin_{\vx} \normof{\mS\mA\vx-\mS\vz}_p\) lie within this reasonable range.

We first examine the vector \(\bar\vz\) defined in \settingref{large-p-eps}.
Intuitively, we say that the entries of \(\vz\) that get zeroed out are so large that the reasonable range of \(\mA\) cannot fit them.
So, we can approximate the true error by \(\normof{\mA\vx-\vz}_p^p \approx \normof{\mA\vx-\bar\vz}_p^p + \normof{\vz-\bar\vz}_p^p\).
That is, minimizing \(\normof{\mA\vx-\vz}_p\) is effectively equivalent to minimizing \(\normof{\mA\vx-\bar\vz}_p\).

We define the zeroing-out procedure in terms of the \(\ell_q\) Lewis weights of \mB here, so this is one place where we adapt the prior work to use the special structure of Vandermonde matrices.
Roughly, the \(\ell_p\) sensitivity \(\psi_p[\mA](i)\) measures how spikey a vector in the reasonable range can be.
The Vandermonde structure lets us bound the sensitivity of \mA with the sensitivity of \mB, since \(\psi_p[\mA](i) \leq \psi_q[\mB](i)\).
Then, we use the fact that all matrices have their \(\ell_q\) sensitivities bounded by their Lewis weights for \(q \leq 2\).
So, we can bound the spikeyness of the \(i^{th}\) entry of a vector in the reasonable range by the \(\ell_q\) Lewis weight of \(\mB\) at row \(i\).
For general matrices, the \(\ell_p\) sensitivity \(\psi_p[\mA](i)\) can be \(d^{\frac p2-1}\) times larger than the \(\ell_p\) Lewis weight, and this way of bounding the entries of \vz is one central step to avoiding the \(\tilde O(d^{p/2})\) dependence.

\begin{lemma}
\label{lem:trunc}
Consider \settingref{large-p-eps}, and let 
\[
	\cB=\left\{ i\in[n] \,:\, \abs{\vz(i)} > \tsfrac{OPT}{\eps} (w_q[\mB](i))^{1/p} \right\}.
\]
So that \(\bar{\vz}\in\bbR^{n_0}\) is equal to \vz but with all entries in \cB set to zero.
Then for all \(\vx\in\bbR^{d+1}\) with \(\normof{\mA\vx}_p \leq C_0 \, OPT\),
\[
	\left|
		\normof{\mA\vx-\vz}_p^p - \normof{\mA\vx-\bar{\vz}}_p^p - \normof{\vz-\bar{\vz}}_p^p
	\right| \leq C_1 \eps \cdot OPT^p
\]
where \(C_1\) is a constant that depends only on \(C_0, C_z\), and \(p\).
\end{lemma}
\begin{proof}
For any \(\vx\in\bbR^{d+1}\), by the definition of \(\ell_p\) sensitivity,
\[
	\frac{\abs{[\mA\vx](i)}^p}{\normof{\mA\vx}_p^p} \leq \psi_p[\mA](i)
\]
From the relationship of \(\ell_p\) sensitivities and \(\ell_q\) sensitivities for Vandermonde matrices, i.e., \theoremref{vander:p:q:sens}, we have
\[
	\frac{\abs{[\mA\vx](i)}^p}{\normof{\mA\vx}_p^p} \leq \psi_p[\mA](i) \leq \psi_q[\mB](i)
\]
Next, by Lemma 2.5 from \cite{MuscoMWY21}, which says that for \(q\in[0,2]\), the \(\ell_q\) sensitivities lower bound the \(\ell_q\) Lewis weights, we have
\[
	\frac{\abs{[\mA\vx](i)}^p}{\normof{\mA\vx}_p^p} \leq \psi_q[\mB](i) \leq w_q[\mB](i)
\]
Thus for \(i\in\cB\) we have
\[
	\abs{[\mA\vx](i)}^p \leq w_q[\mB](i) \cdot \normof{\mA\vx}_p^p \leq \eps^p \normof{\mA\vx}_p^p \cdot \frac{\abs{\vz(i)}^p}{OPT^p}
\]
Since \(\normof{\mA\vx}_p^p \leq C_0^p \, OPT^p\) by assumption, it follows that 
\(
	\abs{[\mA\vx](i)}^p \leq C_0^p \eps^p \cdot|\vz(i)|^p,
\)
and thus \(\abs{[\mA\vx](i) - \vz(i)} \in (1 \pm C_0 \eps) \abs{\vz(i)}\). Using this fact and the fact that \(\bar\vz(i)=0\),
\begin{align*}
	\abs{[\mA\vx](i)-\vz(i)}^p - \abs{[\mA\vx](i)-\bar{\vz}(i)}^p
	&= \abs{[\mA\vx](i)-\vz(i)}^p - \abs{[\mA\vx](i)}^p \\
	&\in (1\pm C_0\eps)^p \abs{\vz(i)}^p \pm C_0^p \eps^p \abs{\vz(i)}^p \\
	&\subseteq (1\pm 3C_0 p\eps) \abs{\vz(i)}^p \pm C_0^p \eps^p \abs{\vz(i)}^p \tag{\lemmaref{binomial-approx}} \\
	&\subseteq (1\pm C_0'\eps) \abs{\vz(i)}^p \numberthis \label{eq:partial-zeroing-pointwise}
\end{align*}
where the last line sets \(C_0' \defeq 3C_0p+C_0^p\).
Then, summing over all \(i\in\cB\),
\[
	\sum_{i\in\cB} \abs{[\mA\vx](i)-\vz(i)}^p
	- \sum_{i\in\cB} \abs{[\mA\vx](i)-\bar{\vz}(i)}^p 
	- \sum_{i\in\cB} \abs{\vz(i)-\bar{\vz}(i)}^p
	\leq C_0'\eps \cdot \sum_{i\in\cB}|\vz(i)|^p
\]
We have by assumption that \(\sum_{i\in\cB}|\vz(i)|^p \leq \normof{\vz}_p^p \leq C_0^p \, OPT^p\). Finally, 
since \(\bar{\vz}(i)=\vz(i)\) for \(i\notin\cB\), we conclude that 
\[
	\left|
		\normof{\mA\vx-\vz}_p^p
		- \normof{\mA\vx-\bar{\vz}}_p^p
		- \normof{\vz-\bar{\vz}}_p^p
	\right|
	= C_0'C_z^p \eps \cdot OPT^p.
\]
\end{proof}

Next, we show the same intuition about \(\vz\) and \(\bar\vz\) holds when looking at the subsampled regression problem; that minimizing \(\normof{\mS\mA\vx-\mS\vz}_p\) is roughly equivalent to minimizing \(\normof{\mS\mA\vx-\mS\bar\vz}_p\).

\begin{lemma}
\label{lem:sens:sketch}
Consider \settingref{large-p-eps}.
Then with probability at least \(\frac{99}{100}\), \(\normof{\mS\vz}_p \leq C_s OPT\) and further for all \(\vx\in\bbR^{d+1}\) with \(\normof{\mA\vx}_p \leq C_0 OPT\), we get
\[
	\left|
		\normof{\mS\mA\vx-\mS\vz}_p^p
		- \normof{\mS\mA\vx-\mS\bar{\vz}}_p^p
		- \normof{\mS(\vz-\bar{\vz})}_p^p
	\right|
	\leq C_2\eps \cdot OPT^p
\]
where \(C_s\) and \(C_2\) are constants that depend only on \(C_0, p\), and \(C_z\).
\end{lemma}
\begin{proof}
The proof builds off of \lemmaref{trunc}.
For any \(\vx\in\bbR^{d+1}\) and \(i\in\cB\), by multiplying both sides of \equationref{partial-zeroing-pointwise} by \(S_{ii}^p\), we have that for all \(i \in \cB\),
\[ 
	\abs{[\mS\mA\vx](i)-[\mS\vz](i)}^p - \abs{[\mS\mA\vx](i)-[\mS\bar{\vz}](i)}^p \in (1\pm C_0'\eps) \abs{[\mS\vz](i)}^p
\]
For all \(i \notin \cB\),
\(
	\abs{[\mS\mA\vx](i)-[\mS\vz](i)}^p - \abs{[\mS\mA\vx](i)-[\mS\bar{\vz}](i)}^p = 0.
\) since \(\bar{\vz}(i)=\vz(i)\) for \(i\notin\cB\).
Summing over all \(i\in[n_0]\), we get
\[
	\normof{\mS\mA\vx-\mS\vz}_p^p - \normof{\mS\mA\vx-\mS\bar\vz}_p^p - \normof{\mS(\vz-\bar\vz)}_p^p \in \pm C_0'\eps \normof{\mS(\vz-\bar\vz)}_p^p
\]
Next, since \(\bar\vz\) is a partial zeroing of \(\vz\), and since \(\E[\normof{\mS\vz}_p^p] = \normof{\vz}_p^p\) we can use Markov's inequality to bound \(\normof{\mS(\bar\vz-\vz)}_p^p \leq \normof{\mS\vz}_p^p \leq 100 \normof{\vz}_p^p \leq 100C_z^p OPT^p\), with probability \(\frac{99}{100}\).
We conclude:
\[
	\normof{\mS\mA\vx-\mS\vz}_p^p - \normof{\mS\mA\vx-\mS\bar\vz}_p^p - \normof{\mS(\vz-\bar\vz)}_p^p \in \pm\, 100C_0'C_z\eps \cdot OPT
\]
\end{proof}

Next we state our core technical contribution: the \emph{Affine Embedding} guarantee.
While the prior work proves this same result, they require \(\tilde O(d^{p/2})\) samples to do so.
In \sectionref{affine-embed}, we show that Vandermonde matrices can do this by taking \(\tilde O(d)\) samples with probabilities proportional to the \(\ell_q\) Lewis weights of \mB.

\begin{theorem}[Affine Embedding]
\label{thm:affine-embed}
Consider \settingref{large-p-eps}.
Then with probability \(\frac{99}{100}\), for all \(\vx\in\bbR^{d+1}\) with \(\normof{\mA\vx}_p \leq C_0 OPT\), we have
\begin{align}
	\abs{
		\normof{\mS\mA\vx - \mS\bar\vz}_p^p - \normof{\mA\vx-\bar\vz}_p^p
	} \leq C_3\eps \cdot OPT^p
	\label{eq:affine-embed}
\end{align}
where \(C_3\) is a constant that depends only on \(C_0, C_z,\) and \(p\).
\end{theorem}
We prove \theoremref{affine-embed} later, in \sectionref{affine-embed}, and instead first show that this affine embedding suffices to prove the correctness of the overall algorithm.
\begin{theorem}
\label{thm:sens:two:stage:largep}
Consider \settingref{large-p-eps}.
Then, \algorithmref{constant:active:Lp} reads \(O(\frac{d}{\eps^{6.5p+2}} \polylog(\frac d\eps))\) entries of \(\vb\) and outputs a vector \(\tilde\vx\) such that with probability \(0.9\),
\[
	\normof{\mA\tilde{\vx}-\vb}_p \leq (1+\eps) \min_{\vx\in\bbR^{d+1}} \normof{\mA\vx-\vb}_p
\]
\end{theorem}
\begin{proof}
By \lemmaref{lp-const-factor}, we know that step 1 from \algorithmref{constant:active:Lp} returns a vector \(\vx_c\) such that \(\normof{\mA\vx_c-\vb}_p \leq C_z \cdot OPT\).
Recall that \(\vz \defeq \vb - \mA\vx_c\), so we equivalently have \(\normof{\vz}_p \leq C_z \cdot OPT\).
Let \(\bar\vz \in \bbR^{n_0}\) be the partially zeroed out copy of \vz as in \settingref{large-p-eps}.
Combining \lemmaref{trunc}, \lemmaref{sens:sketch}, and \theoremref{affine-embed}, for all \(\vx\) with \(\normof{\mA\vx}_p \leq C_0 \cdot OPT\), we get
\begin{align*}
	\normof{\mS\mA\vx-\mS\vz}_p^p
	&\in \normof{\mS\mA\vx - \mS\bar\vz}_p^p + \normof{\mS\vz-\mS\bar\vz}_p^p \pm C_2 \eps \cdot OPT^p \tag{\lemmaref{sens:sketch}} \\
	&\subseteq \normof{\mA\vx - \bar\vz}_p^p + \normof{\mS\vz-\mS\bar\vz}_p^p \pm (C_2+C_3) \eps \cdot OPT^p \tag{\theoremref{affine-embed}} \\
	&\subseteq \normof{\mA\vx - \vz}_p^p - \normof{\vz-\bar\vz}_p^p + \normof{\mS\vz-\mS\bar\vz}_p^p \pm (C_1+C_2+C_3) \eps \cdot OPT^p \tag{\lemmaref{trunc}} \\
	&= \normof{\mA\vx - \vz}_p^p - \hat C \pm (C_1+C_2+C_3) \eps \cdot OPT^p
\end{align*}
where \(\hat C \defeq \normof{\vz-\bar\vz}_p^p - \normof{\mS\vz-\mS\bar\vz}_p^p\) is independent of \vx.
Note that since \(\bar\vz\) is a partial zeroing of \vz, \(\normof{\vz - \bar\vz}_p \leq \normof{\vz}_p \leq C_z \cdot OPT\).
Similarly, \(\normof{\mS\vz-\mS\bar\vz}_p \leq \normof{\mS\vz}_p \leq C_s \cdot OPT\).
So, we have \(\hat C \leq (C_z^p + C_s^p) \, OPT\) and thus we can equivalently write this last bound as, for any \vx with \(\normof{\mA\vx}_p \leq C_0 \cdot OPT\),
\begin{align}
	\abs{\normof{\mS\mA\vx-\mS\vz}_p^p - (\normof{\mA\vx - \vz}_p^p + \hat C)} \leq C_4 \eps \cdot OPT^p
	\label{eq:affine-subspace-embed}
\end{align}
where \(C_4 \defeq C_1+C_2+C_3\).
We will apply \equationref{affine-subspace-embed} twice, once to \(\hat\vx = \argmin_{\vx} \normof{\mS\mA\vx-\mS\vz}_p\) and once to \(\vx^* \defeq \argmin_{\vx} \normof{\mA\vx-\vz}_p\).
To do so, we first have to verify that \(\normof{\mA\hat\vx}_p\) and \(\normof{\mA\vx^*}_p\) are small enough -- i.e. are at most \(C_0 OPT\).
We first bound \(\normof{\mA\vx^*}_p \leq \normof{\mA\vx^* - \bar\vz}_p + \normof{\bar\vz}_p \leq 2\normof{\bar\vz}_p \leq 2C_z OPT \leq C_0 OPT\).
Next, by \lemmaref{lp-subspace-embed}, we have that \mS is an \(\ell_p\) subspace embedding.
So, we have \(\normof{\mA\hat\vx}_p \leq 2\normof{\mS\mA\hat\vx}_p\) and by Markov's inequality, with probability \(\frac{99}{100}\), we have:
\[
	2\normof{\mS\mA\hat\vx}_p
	\leq 2\normof{\mS\mA\hat\vx-\mS\vz}_p + 2\normof{\mS\vz}_p
	\leq 2\normof{\mS\mA\vx^*-\mS\vz}_p + 2\normof{\mS\vz}_p
	\leq 200(\normof{\mA\vx^*-\vz}_p + \normof{\vz}_p)
	\leq C_0 OPT
\]
We proceed to apply \equationref{affine-subspace-embed} twice, to get
\begin{align*}
	\normof{\mA\hat\vx - \vz}_p^p
	&\leq \normof{\mS\mA\hat\vx - \mS\vz}_p^p - \hat C + C_4\eps OPT^p \tag{\equationref{affine-subspace-embed}}\\
	&\leq \normof{\mS\mA\vx^* - \mS\vz}_p^p - \hat C + C_4\eps OPT^p \tag{Optimality of \(\hat\vx\)}\\
	&\leq (\normof{\mA\vx^* - \vz}_p^p + \hat C) - \hat C + 2C_4\eps OPT^p \tag{\equationref{affine-subspace-embed}}\\
	&= \normof{\mA\vx^* - \vz}_p^p + 2C_4\eps OPT^p
\end{align*}
And so the overall predictor \(\tilde\vx = \vx_c + \hat\vx\) has
\begin{align*}
	\normof{\mA\tilde\vx-\vb}_p^p
	&= \normof{\mA\vx_c + \mA\hat\vx - \vz - \mA\vx_c}_p^p \\
	&= \normof{\mA\hat\vx - \vz}_p^p \\
	&\leq \min_{\vx} \normof{\mA\vx - \vz}_p^p + 2C_4\eps OPT^p \\
	&= \min_{\vx} \normof{\mA\vx - (\vb + \mA\vx_c)}_p^p + 2C_4\eps OPT^p \\
	&= \min_{\vx} \normof{\mA\vx - \vb}_p^p + 2C_4\eps OPT^p \\
	&= (1+2C_4\eps)\min_{\vx} \normof{\mA\vx - \vb}_p^p
\end{align*}
Note that our proof ensures that \theoremref{sens:two:stage:largep} holds with a fixed constant probability. 
\end{proof}

\subsection{Proving the Affine Embedding 
\texorpdfstring{(\theoremref{affine-embed})}{}}
\label{sec:affine-embed}
To prove \theoremref{affine-embed}, we want a bound over all vectors \(\mA\vx\) where \(\normof{\mA\vx}_p^p \leq C_0^p OPT^p\).
Since a na\"ive \(\eps\)-net argument would lead to a suboptimal dependence on \(d\), we follow the first arXiv version of \cite{MuscoMWY21}, and appeal to a more refined net analysis introduced in \cite{bourgain1989approximation}.
In that work, the authors construct a ``compact rounding'' for all vectors in the set \(\{\mA\vx : \normof{\mA\vx}_p^p \leq 1\}\).
In particular, they construct a series of nets \(\cD_0,\ldots,\cD_\ell\) (with different properties for each \(k\in\{0,\ldots,\ell\}\)), such that every \(\mA\vx\) with \(\normof{\mA\vx}_p^p\leq1\) can be approximated as \(\mA\vx \approx \sum_{k=0}^\ell \vd_k\), where each \(\vd_k \in \cD_k\).
After scaling these vectors by a factor of \(C_0 OPT\) and applying a union bound over the nets \(\cD_0, \ldots, \cD_\ell\), \cite{bourgain1989approximation} obtains a sampling result for \(\ell_p\) Lewis weights with an optimal \(d\) dependence of \(\tilde O(d^{\max\{1,p/2\}})\).

To avoid this large \(d\) dependence for \(p>2\), we return to the expanded Vandermonde matrix \(\mB\in\bbR^{n_0 \times d_B}\).
In \lemmaref{tensored-compact-rounding}, we show how to use the nets \(\cD_0,\ldots,\cD_\ell\) from the \(\ell_q\) compact rounding on \mB to create nets \(\cE_0,\ldots,\cE_\ell\) for an \(\ell_p\) compact rounding on \mA.
Each \(\ell_p\) net \(\cE_k\) will have the same cardinality as the corresponding \(\ell_q\) net \(\cD_k\), which makes it significantly smaller than the black-box net that would be created for Lewis weight sampling general matrices in the \(\ell_p\) norm.
Lastly, \lemmaref{tensored-compact-rounding} also uses a technique from \cite{bourgain1989approximation} to transform \(\cE_0,\ldots,\cE_\ell\), which approximate vectors of the form \(\mA\vx\), into new nets \(\cF_0,\ldots,\cF_\ell\), which have similar size and approximate vectors of the form \(\mA\vx-\bar\vz\).

To get started, we use the following compact rounding lemma, proven in the first version of \cite{MuscoMWY21} uploaded to arXiv, with a complete and simplified proof included in \appendixref{compact-rounding} for completeness.
Specifically, we state the result from \appendixref{compact-rounding} in the special case when \(\vv=0\):
\begin{lemma}[Compact Rounding,~\cite{MuscoMWY21}]
\label{lem:compact-rounding}
Let \(\mB\in\bbR^{n_0 \times d_B}\) and \(q\in[0,2]\).
Let \(\cN_{\eps_c}\) be an \(\eps_c\)-Net over \(\normof{\mB\vy}_q=1\) with \(\abs{\cN_\eps} \leq O(d \log(\frac1\eps))\).
Let \(\ell = \log_{1+\eps_c}((2d_B)^{1/q})\).
Then, there exists sets of vectors \(\cD_0,\ldots,\cD_\ell \subseteq \bbR^{n_0}\), such that:
For all \(\vu\in\cN_{\eps_c}\) we can pick \(\vd_0\in\cD_0,\ldots,\vd_\ell\in\cD_\ell\) to create a ``compact rounding'' \(\vu' = \sum_{k=0}^\ell \vd_{k}\) where:
\begin{enumerate}
	\item \(\abs{\vu(i) - \vu'(i)} \leq \eps_c \abs{\vu(i)}\) for all \(i\in[n_0]\)
	\item \(\abs{\vd_{k}(i)} \leq \frac{1}{\eps_c} (\frac12(\frac{w_q[\mB](i)}{d_B} + \frac1{n_0}))^{1/q} (1+\eps_c)^{k+2}\) for all \(i\in[n_0], k\in\{0,\ldots,\ell\}\)
	\item \(\vd_0,\ldots,\vd_\ell\) all have disjoints supports
\end{enumerate}
Further, we have that the sets \(\cD_0,\ldots,\cD_\ell\) are not too large:
\[
	\log\abs{\cD_k} \leq C_q \frac{d_B \log(n_0)}{\eps_c^{1+q}(1+\eps_c)^{qk}},
\]
where $C_q$ is a fixed constant depending only on $q$.
\end{lemma}
Note that we can upper bound \(\frac{w_q[\mB](i)}{d_B} + \frac1{n_0} \leq \frac{w_q[\mB](i)}{d_B} \polylog(d)\), so we instead have \(\abs{\vd_k(i)} \leq \frac{1}{\eps} (\frac{w_q[\mB](i)}{d_B})^{1/q} (1+\eps)^{k+2} \polylog(d)\).
To do this, note that by \theoremref{sens:lewis:approx}, \(w_q[\mB](i) \geq \frac1{n_0} w'(s_i) \cdot \frac{1}{\polylog(d)}\), where \(w'\) is the clipped Chebyshev measure for degree \(rd\).
Then, \(w_q[\mB](i) \geq \frac{d_B}{n_0} \cdot \frac{1}{\polylog(d)}\), so that \(\frac1{n_0} \leq \frac{w_q[\mB](i)}{d_B} \polylog(d)\), and so \(\frac{w_q[\mB](i)}{d_b} + \frac1{n_0} \leq \frac{w_q[\mB](i)}{d_B} \polylog(d)\).

\begin{lemma}[Vandermonde Compact Rounding]
\label{lem:tensored-compact-rounding}
Let \(\mA\in\bbR^{n_0 \times d+1}\) and \(p>2\).
Let \(\cN_\eps\) be an \(\eps\)-Net over \(\normof{\mA\vx}_p \leq C_0 OPT\), so that any \(\vx\) with \(\normof{\mA\vx}_p \leq C_0 OPT\) has some \(\vu\in\cN_\eps\) such that \(\normof{\mA\vx-\vu}_p \leq \eps OPT\).
Then we can pick an odd integer \(r\) such that \(\frac12 p \leq r \leq \frac32 p\), and let \(q = \frac{p}{r} \in [\frac23,2]\).
Let \(\ell = \log_{1+\eps}((2d_B)^{1/q})\).
There exists sets of vectors \(\cF_0,\ldots,\cF_\ell \subseteq \bbR^{n_0}\), such that:
For any \(\vu\in\cN_\eps\) we let \(\vr \defeq \vu - \bar\vz\) and we can pick \(\vf_0\in\cF_0,\ldots,\vf_\ell\in\cF_\ell\) to create a ``compact rounding'' \(\vr' = \sum_{k=0}^\ell \vf_{k}\) where:
\begin{enumerate}
	\item \(\abs{\vr(i) - \vr'(i)} \leq \eps \max\{\abs{[\mA\vx](i)}, \abs{\bar\vz(i)}\}\) for all \(i\in[n_0]\)
	\item \(\abs{\vf_{k}(i)} \leq \frac{OPT}{\eps^2} (\frac{w_q[\mB](i)}{d_B} \polylog(d))^{1/p} (1+\eps^r)^{\frac{k+2}{r}}\) for all \(i\in[n_0], k\in\{0,\ldots,\ell\}\)
	\item \(\vf_0,\ldots,\vf_\ell\) all have disjoints supports
\end{enumerate}
Further, we have that the sets \(\cF_0,\ldots,\cF_\ell\) are not too large:
\[
	\log\abs{\cF_k} \leq C_q \frac{d_B \log(n_0)}{\eps^{r(1+q)}(1+\eps^r)^{qk}}
\]
\end{lemma}
\begin{proof}
Depending on the value of \(p\), we will pick \(q\) differently.
If \(p\in(2,3)\), we let \(r=3\) and \(q \defeq \frac pr \in [\frac23,1]\).
If \(p \geq 3\), we let \(r\) be an odd integer such that \(r \leq p < 2r\), and let \(q \defeq \frac pr \in [1,2]\).
In both cases \(r\) is an odd integer, we have \(p=qr\), we know that the \(\ell_q\) Lewis weights of \mB are close to the Chebyshev measure, and \lemmaref{compact-rounding} accepts this value of \(q\).
The rest of this paper will not distinguish between the \(p\in(2,3)\) and the \(p\geq3\) cases.
Notably, the compact rounding requires being given \(\cN_q\), an \(\ell_q\) net over \(\{\mB\vy \,:\, \normof{\mB\vy}_q \leq 1\}\).
But we want to make sure that all \(\mA\vx\in\cN_\eps\) have \((\mA\vx)^r \in \cN_q\).
So, formally, let \(\cN_{q,0}\) be an arbitrary \(\ell_q\) net for \(\{\mB\vy \,:\, \normof{\mB\vy}_q \leq 1\}\), and let \(\cN_{induced} \defeq \{\mB\vy \,:\, ([\mA\vx](i))^r = [\mB\vy](i) \text{ for all } i\}\) be the mapping of \(\cN_\eps\) to the range of \(\mB\).
By Lemma 2.4 of \cite{bourgain1989approximation}, we have that both \(\cN_{q,0}\) and \(\cN_{\eps}\) have cardinality at most \((\frac3\eps)^{d}\).
We then apply \lemmaref{compact-rounding} on the net \(\cN_q \defeq \cN_{q,0} \bigcup \cN_{induced}\) and with \(\eps_c = \eps^r\).

Also, note that the vectors in \lemmaref{compact-rounding} formally require \(\normof{\mB\vy}_q \leq 1\), while we have \(\normof{\mB\vy}_q = \normof{\mA\vx}_p^r \leq (C_0 OPT)^r\).
So, we scale up the vectors \(\vd_k\) returned by \lemmaref{compact-rounding} by a factor of \((C_0 OPT)^r\), so that \(\abs{\vd_k(i)} \leq \frac{(C_0 OPT)^r}{\eps^r} (\frac{w_q[\mB](i)}{d_B} + \frac1{n_0})^{1/q} (1+\eps^r)^{k+2}\).

With this in place, now we fix any \(\vu\in\cN_{induced}\), and let \(\sum_{k=0}^\ell \vd_k\) be the compact rounding of \vu as defined by \lemmaref{compact-rounding}.
Using the fact that \(qr=p\), we let \( \alpha_i \defeq \tsfrac{OPT}\eps (\tsfrac{w_q[\mB](i)}{d_B} \polylog(d))^{1/p} \) so that \(\abs{\vd_k(i)} \leq \alpha_i^r (1+\eps^r)^{k+2}\).
We now intuitively round \(\mA\vx \approx \sum_{k=0}^\ell (\vd_k)^{1/r}\).
We define \(\ve_0,\ldots,\ve_\ell\) such that:
\[
	\ve_k(i) \defeq (\vd_k(i))^{1/r}
\]
so that \(\abs{\ve_k(i)} \leq \alpha_i (1+\eps^r)^{\frac{k+2}{r}}\).
Using the fact that \(r\) is an odd integer, we have \(\sign(\ve_k(i)) = \sign(\vd_k(i))\).
Further, looking at the proof of the compact rounding in \appendixref{compact-rounding} with \(\vv=0\), we see from \lemmaref{ell:inf:error:bound} that \(\sign(\vd_k(i)) = \sign(\vu(i))\).
So, we have that \(\sign(\ve_k(i)) = \sign(\mA\vx(i))\).
This definition of \(\ve_k\) means that \(\abs{\ve_k(i)} \leq \alpha_i (1+\eps^r)^{\frac{k+2}{r}}\), and further that
\begin{align*}
	\abs{\mA\vx(i) - \sum_{k=0}^\ell \ve_k(i)}
	&= \abs{(\vu(i))^{1/r} - (\vd_k(i))^{1/r}} \\
	&= \abs{\abs{\vu(i)}^{1/r} - \abs{\vd_k(i)}^{1/r}} \tag{The signs are equal} \\
	&\leq \abs{\abs{\vu(i)} - \abs{\vd_k(i)}}^{1/r} \tag{\(\abs{x^{1/r} - y^{1/r}} \leq \abs{x-y}^{1/r}\) for all \(x,y\geq0\)}\\
	&\leq (\eps^r \abs{\vu(i)})^{1/r} \\
	&= \eps \abs{\mA\vx(i)}
\end{align*}
Also note that \(\ve_k\) has the same support as \(\vd_k\), so that all the properties of \lemmaref{compact-rounding} are preserved, just in estimating a slightly different vector.
We next examine rounding \(\mA\vx-\bar\vz\) to a compact rounding.
Borrowing a proof strategy from \appendixref{compact-rounding},
\begin{align*}
	B_{k,\vu}
	&\defeq \left\{
		i \in [n_0]
		\,:\,
		\ve_k(i) \neq 0
		\, , \,
		\eps \abs{\bar\vz(i)} \leq 2\alpha_i(1+\eps^r)^{\frac{k+2}{r}}
	\right\} \tag{for \(k\in\{0,\ldots,\ell\}\)} \\
	H_k
	&\defeq \left\{
		i \in [n_0]
		\,:\,
		2\alpha_i(1+\eps^r)^{\frac{k+1}{r}} < \eps \abs{\bar\vz(i)} \leq 2\alpha_i(1+\eps^r)^{\frac{k+2}{r}}
	\right\} \tag{for \(k\in\{1,\ldots,\ell\}\)}\\
	G_{k,\vu} &\defeq H_k \setminus \bigcup_{k' \geq k} \{i\in[n_0] \,:\, \ve_{k'}(i) \neq 0\}
	\tag{for \(k\in\{1,\ldots,\ell\}\)}
\end{align*}
Note that \(\abs{\bar\vz(i)} \leq \frac{OPT}{\eps}(w_q[\mB](i))^{1/p} < 2\alpha_i(1+\eps^r)^{\frac{\ell+2}{r}}\), so all entries of \(\bar\vz\) are covered by our disjoint sets.
The sets \(B_{0,\vu},\ldots,B_{\ell,\vu},G_{1,\vu},\ldots,G_{\ell,\vu}\) will define the support of the final compact rounding vectors we will return, so we first show that these sets partition \([n_0]\):
In the following cases, consider any \(k, k'\):
\begin{itemize}
	\item
	\(B_{k,\vu} \bigcap B_{k',\vu} = \emptyset\) since \(i\in B_{k,\vu}\) implies \(\ve_k(i)\neq0\) implies \(\ve_{k'}(i)=0\) implies \(i\notin B_{k',\vu}\).
	\item
	\(G_{k,\vu} \bigcap G_{k',\vu} \subseteq H_{k} \bigcap H_{k'} = \emptyset\) since \(H_k\) and \(H_{k'}\) have no intersection by definition.
	\item
	For \(k \geq k'\), \(B_{k,\vu} \bigcap G_{k',\vu} = \emptyset\) since \(i\in B_{k,\vu}\) means \(\ve_k(i) \neq 0\) so \(i \in \bigcup_{k' \geq k} \{i\in[n_0] \,:\, \ve_{k'}(i) \neq 0\}\) so \(i \notin G_{k',\vu}\).
	\item
	For \(k < k'\), \(B_{k,\vu} \bigcap G_{k',\vu} = \emptyset\) since \(k'\geq k+1\) and \(i\in H_{k'}\) means \(\eps\abs{\bar\vz(i)} > 2\alpha(1+\eps^r)^{\frac{k'+1}{r}} \geq 2\alpha(1+\eps^r)^{\frac{k+2}{r}}\), which contradicts \(i\in B_{k,\vu}\).
\end{itemize}
So, we can now define the vectors \(\vf_0,\ldots,\vf_\ell\) as
\[
	\vf_k(i) \defeq \begin{cases}
		\ve_k(i) - \bar\vz(i) \qquad & i \in B_{k,\vu} \\
		- \bar\vz(i) \qquad & i \in G_{k,\vu} \\
		0 \qquad & \textit{otherwise}
	\end{cases}
\]
Now, we show that \(\vr' \defeq \sum_{k=0}^\ell \vf_k\) satisfies the guarantees of \lemmaref{tensored-compact-rounding}.
Fix any \(i\in[n_0]\) and let \(k\in\{0,\ldots,\ell\}\) be the index\footnote{Technically, we don't guarantee that all \(i\in[n_0]\) are associated with some \(k\in\{0,\ldots,\ell\}\). But the relative error guarantee from \lemmaref{compact-rounding} and definitions of \(B_{k,\vu}\) and \(G_{k,\vu}\) imply that if \(\vu(i)\neq0\) or \(\bar\vz(i)\neq0\) then such a \(k\) exists, which suffices to prove our error guarantee.} where \(\vf_k(i) \neq 0\).
Then, recalling that \(\vr = \mA\vx-\bar\vz\),
\begin{align*}
	\abs{\vf_k(i) - \vr(i)}
	&= \abs{\ve_k(i) - [\mA\vx](i)} \tag{when \(i\in B_{k,\vu}\)} \leq \eps \abs{[\mA\vx](i)} \\[0.75em]
	\abs{\vf_k(i) - \vr(i)}
	&= \abs{[\mA\vx](i)} \tag{when \(i\in G_{k,\vu}\)} \\
	&\leq (1+\eps)\abs{\ve_{k'}(i)} \tag{for some \(k' < k\), by def of \(G_{k,\vu}\)} \\
	&\leq 2 \alpha_i (1+\eps^r)^{\frac{k'+2}{r}} \tag{\(\abs{\ve_{k'}(i)} \leq \alpha_i (1+\eps)^\frac{k'+2}{r}\)} \\
	&\leq \eps \abs{\bar\vz(i)} \tag{def of \(H_{k}\)}
\end{align*}
And so we find \(\abs{\vr'(i) - \vr(i)} \leq \eps \max\{\abs{[\mA\vx](i)}, \abs{\bar\vz(i)}\}\).
We also have that \(\vf_0,\ldots,\vf_\ell\) have disjoint supports because \(B_{0,\vu},\ldots,B_{\ell,\vu},G_{1,\vu},\ldots,G_{\ell,\vu}\) have disjoint supports.

Next, we bound the size of entries of \(\vf_k\).
We have \(\abs{\vf_k(i)} \leq \abs{\ve(i)} + \abs{\bar\vz(i)} \leq (1+\frac{2}{\eps})\alpha_i (1+\eps^r)^{\frac{k+2}{r}} \leq \frac{OPT}{\eps^2} (\frac{w_q[\mB](i)}{d_B} \polylog(d))^{1/p} (1+\eps^r)^{\frac{k+2}{r}}\).

To bound the number of possible \(\vf_k\) vectors, note that \(\vf_k\) is a deterministic function in \(B_{k,\vu}\) and \(G_{k,\vu}\).
So, let \(B_{k} \defeq \{B_{k,\vu} : \vu\in\cN_\eps\}\) be the set of all possible ``B'' index sets generated at layer \(k\), and similarly let \(G_k \defeq \{B_{k,\vu} : \vu\in\cN_\eps\}\).
Then, looking across all possible fixings of \(\vu\in\cN_\eps\), each \(\vf_k\) is deterministic in some \(\cS_1 \in B_k\) and some \(\cS_2 \in G_k\).
So, the number of possible \(\vf_k\) is at most
\[
	\abs{\cF_k} 
	= \abs{\{\vf_k \,:\, \vu\in\cN_\eps\}}
	\leq \abs{\{(\cS_1,\cS_2) \,:\, \cS_1\in B_k,~ \cS_2\in G_k\}}
	= \abs{B_k} \cdot \abs{G_k}
\]
Next, since \(B_{k,\vu} \subseteq \{i\in [n_0] \,:\, \ve_k(i) \neq 0\}\), and since \(\ve_k\) are a simple bijection with \(\vd_k\in\cD_k\), we have \(\abs{B_k} \leq \abs{\cD_k}\).
The same holds for \(G_k\), so \(\abs{G_k} \leq \abs{\cD_k}\).
We conclude that
\[
	\log\abs{\cF_k} \leq \log\abs{B_k} + \log\abs{G_k} \leq 2\log\abs{\cD_k} = 2C_r \frac{d_B \log(n_0)}{\eps^{r(1+q)}(1+\eps^r)^{qk}}
\]
\end{proof}

\begin{lemma}
\label{lem:preserve-compact-roundings}
Let \(p_i \defeq \min\{1, \frac{m}{n_0} \frac{1}{\sqrt{1-s_i^2}}\}\), where \(s_1,\ldots,s_{n_0}\) are times samples uniformly at random from \([-1,1]\), and where \(m = O(\frac{d}{\eps^{6.5p+2}}\log(d))\).
Consider the diagonal sampling matrix \(\mS\in\bbR^{n_0 \times n_0}\) which takes \(\mS_{ii}^p = \frac{1}{p_i}\) with probability \(p_i\) and \(\mS_{ii} = 0\) otherwise.
Then consider the set of all possible rounding vectors \(\vr'\) created by \lemmaref{tensored-compact-rounding}.
With probability \(\frac{98}{100}\), all such \(\vr'\) have \(\normof{\mS\vr'}_p^p \in \normof{\vr'} \pm \eps^p OPT^p\).
\end{lemma}
\begin{proof}
First, we simplify the probabilities \(p_i\).
We know by \lemmaref{max-unif-cheby} that \(\max_i \frac{1}{\sqrt{1-s_i^2}} \leq C_c \sqrt{n_0}\) with probability \(\frac{99}{100}\).
So,
\[
	\frac{m}{n_0} \frac{1}{\sqrt{1-s_i^2}}
	\leq \frac{m C_c}{\sqrt{n_0}}
	\leq 1
\]
Where the last inequality holds so long as \(m \leq O(\sqrt{n_0}) = \tilde O(d^{2.5} p^{O(p)} \frac{1}{\eps^{O(p^2)}})\), which is satisfied by our choice of \(m\).
This means that \(p_i = \min\{1, \frac{m}{n_0} \frac{1}{\sqrt{1-s_i^2}}\}\) can be simplified to just \(p_i = \frac{m}{n_0} \frac{1}{\sqrt{1-s_i^2}}\).

Now, we move onto proving the correctness of \(\normof{\mS\vr'}_p^p\).
Fix any compact rounding \(\vr' = \sum_{k=0}^\ell \vf_k\) created by \lemmaref{tensored-compact-rounding}.
Then, since \(\vf_0,\ldots,\vf_\ell\) have disjoint support,
\[
	\normof{\mS\vr}_p^p = \sum_{k=0}^\ell \normof{\mS\vf_k}_p^p
\]
So it suffices to just prove that \(\normof{\mS\vf_k}_p^p \in \normof{\vf_k}_p^p \pm \frac{\eps^p}{\ell+1} OPT^p\) for all \(\vf_k\in\cF_k\) for all \(k\in\{0,\ldots,\ell\}\).
The rest of this proof shows this concentration across all \(\vf_k\) vectors.

Fix any \(\vf_k\in\cF_k\) for any \(k\in\{0,\ldots,\ell\}\).
Then, we have:
\begin{align*}
	\abs{\vf_k(i)}^p
	&\leq \frac{w_q[\mB](i)}{d_B} \cdot \frac{(1+\eps^r)^{q(k+2)}}{\eps^{2p}} \,OPT^p \polylog(d)  \tag{\lemmaref{tensored-compact-rounding}}\\
	\frac{1}{p_i} \abs{\vf_k(i)}^p
	&\leq \frac{1}{p_i} \frac{w_q[\mB](i)}{d_B} \cdot \frac{(1+\eps^r)^{q(k+2)}}{\eps^{2p}} \,OPT^p \polylog(d) \\
	&= \frac{n_0 \sqrt{1-s_i^2}}{m d_B} \cdot w_q[\mB](i) \cdot \frac{(1+\eps^r)^{q(k+2)}}{\eps^{2p}} \,OPT^p \polylog(d) \tag{\(p_i = \frac{m}{n_0 \sqrt{1-s_i^2}}\)} \\
	&= \frac{n_0 \sqrt{1-s_i^2}}{m d_B} \cdot \left(\frac{d_B}{n_0} \frac{1}{\sqrt{1-s_i^2}} \polylog(d)\right) \cdot \frac{(1+\eps^r)^{q(k+2)}}{\eps^{2p}} \,OPT^p \polylog(d) \tag{\theoremref{sens:lewis:approx}} \\
	&= \frac{(1+\eps^r)^{q(k+2)}}{m \eps^{2p}} \,OPT^p \polylog(d)
\end{align*}
Next, we will let \(X_i \defeq \mS^p_{ii} \abs{\vf_k(i)}^p - \abs{\vf_k(i)}^p\), which are mean-zero random variables such that \(\sum_{i=1}^{n_0} X_i = \normof{\mS\vf_k}_p^p - \normof{\vf_k}_p^p\).
Letting \(B(n,p)\) be the binomial distribution, we then bound
\begin{align*}
	\E[X_i^2] &= \Var[\mS_{ii}^p \abs{\vf_k(i)}^p] = \tsfrac{1}{p_i^2}\abs{\vf_k(i)}^{2p} \Var[B(1,p_i)] \leq \tsfrac{1}{p_i} \abs{\vf_k(i)}^{2p} \\
	\sum_{i=1}^{n_0} \E[X_i^2] &\leq \sum_{i=1}^{n_0} \frac{1}{p_i} \abs{\vf_k(i)}^2 \leq \normof{\vf_k}_p^p \cdot \max_{i\in[n_0]} \tsfrac{1}{p_i}\abs{\vf_k(i)}^p
\end{align*}
And so, by Bernstein's Inequality (\importedtheoremref{bernstein-ineq}) and since \(\abs{X_i} \leq \max_{i} \tsfrac{1}{p_i}\abs{\vf_k(i)}^p\), we get the concentration
\begin{align*}
	\Pr[\abs{\normof{\mS\vf_k}_p^p - \normof{\vf_k}_p^p} \leq \gamma \,OPT^p]
	&= \Pr[\abs{\textstyle{\sum_{i=1}^{n_0} X_i}} \leq \gamma \,OPT^p] \\
	&\leq 2\exp\left(-\frac{\frac12\gamma^2 \,OPT^{2p}}{(\normof{\vf_k}_p^p + \frac{\gamma}{3} \,OPT^{p}) \cdot \max_{i} \tsfrac{1}{p_i}\abs{\vf_k(i)}^p}\right) \\
\intertext{
	Since \(\gamma = \frac{\eps^p}{\ell+1} \leq 1\) and \(\normof{\vf_k}_p^p \leq \normof{\mA\vx-\bar\vz}_p^p \leq (C_0 + C_z)^p \,OPT^p\), and letting \(C_{B} = 2((C_0+C_z)^p+1)\):
}
	&\leq 2\exp\left(-\frac{\gamma^2 \,OPT^{2p}}{C_B \,OPT^p \cdot \max_{i} \tsfrac{1}{p_i}\abs{\vf_k(i)}^p}\right) \\
	&\leq 2\exp\left(-\frac{\gamma^2 \,OPT^{p}}{C_b \cdot \frac{(1+\eps^r)^{q(k+2)}}{m \eps^{2p}} \,OPT^{p} \polylog(d) }\right) \\
	&= 2\exp\left(-m \frac{\gamma^2 \eps^{2p}}{(1+\eps^r)^{q(k+2)} \polylog(d) }\right) \\
	&\leq \delta
\end{align*}
This is less than \(\delta\) for \(m = \frac{(1+\eps^r)^{q(k+2)}}{\eps^{2p} \gamma^2} \polylog(d) \log(\frac{2}{\delta})\).
Union bounding over all \(\vf_k\in\cF_k\), we get
\begin{align*}
	m
	&= \frac{(1+\eps^r)^{q(k+2)}}{\eps^{2p} \gamma^2} \cdot C_r \frac{d_B \log(n_0)}{\eps^{r(1+q)}(1+\eps^r)^{qk}} \cdot \polylog(d) \log(\tsfrac{2}{\delta}) \\
	&= d_B \frac{(1+\eps^r)^{2q}}{\eps^{3p+r} \gamma^2} \cdot \log(n_0) \polylog(d) \log(\tsfrac{2}{\delta}) \\
\end{align*}
Note \((1+\eps^r)^{2q} \leq 2^{2q} \leq 2^4\).
Lastly, we union bound over all \(k\in[\ell]\), where \(\ell = O(\frac{\log(d)}{\eps})\), so that \(\gamma = \frac{\eps^p}{\ell+1} = O(\frac{\eps^{p+1}}{\log(d)})\), and also recall that \(n_0 = O(\frac{1}{\eps^{O(p^2)}} d^5 p^{O(p^2)} \log(\frac{d}{\eps}))\) so that \(\log(n_0) = O(p^2 \log(\frac{pd}{\eps}))\), and that \(r \leq \frac32 p\), so we conclude that
\[
	m = \frac{d_B}{\eps^{6.5p+2}} \cdot \polylog(\tsfrac{pd}{\eps\delta})
\]
samples suffice to achieve the embeddings for all \(\vf_k\) and therefore for all \(\vr'\).
\end{proof}

\begin{lemma}
\label{lem:compact-rounding-suffices}
Let \(\cN_\eps\) be an \(\eps\)-Net on \(\{\mA\vx \,:\, \normof{\mA\vx}_p \leq C_0 OPT\}\), so that for any \(\mA\vx\) in this set there exists some \(\vu\in\cN_\eps\) such that \(\normof{\mA\vx-\vu}_p \leq \eps OPT\).
Consider the set of possible residual vectors \(\vr = \vu - \bar\vz\) for all \(\vu\in\cN_\eps\), and the corresponding roundings \(\vr'\) created by \lemmaref{tensored-compact-rounding}.
Suppose the sampling matrix \mS ensures that \(\normof{\mS\vr'}_p^p \in \normof{\vr'}_p^p \pm \eps OPT\).
Then, \(\normof{\mS\vr}_p^p \in \normof{\vr}_p^p \pm C_{\cN} \eps^p \cdot OPT^p\), where \(C_{\cN}\) is a constant that depends only on \(C_0, C_z\), and \(p\).
\end{lemma}
\begin{proof}
We start with a triangle inequality to show three individual terms we need to bound:
\[
	\abs{\normof{\mS\vr}_p - \normof{\vr}_p}
	\leq \abs{\normof{\mS\vr'}_p - \normof{\vr'}_p} + \normof{\vr-\vr'}_p + \normof{\mS\vr-\mS\vr'}_p
\]
For two numbers \(b \geq a \geq 0\), we have \((b-a)^p \leq (b-a) b^{p-1} \leq b^{p} - ab^{p-1} \leq b^p - a^p\).
So, our given assumption on \(\normof{\mS\vr'}_p^p\) implies that \(\abs{\normof{\mS\vr'}_p - \normof{\vr'}_p}^p \leq \abs{\normof{\mS\vr'}_p^p - \normof{\vr'}_p^p} \leq \eps^p OPT^p\).
That is, the first term above is bounded by \(\eps OPT\).
The second term relies on the first property of \lemmaref{compact-rounding}, which bounds \(\abs{\vr(i) - \vr'(i)} \leq \eps \max\{\abs{\vu(i)}, \abs{\bar\vz(i)}\}\).
From there, we get
\begin{align*}
	\abs{\vr(i) - \vr'(i)}
	&\leq \eps \max\{\abs{\vu(i)}, \abs{\bar\vz(i)}\} \\
	\abs{\vr(i) - \vr'(i)}^p
	&\leq \eps^p \max\{\abs{\vu(i)}^p, \abs{\bar\vz(i)}^p\} \\
	&\leq \eps^p (\abs{\vu(i)}^p + \abs{\bar\vz(i)}^p) \numberthis\label{eq:compact-loss-powered} \\
	\normof{\vr - \vr'}_p^p
	&\leq \eps^p (\normof{\vu}_p^p + \normof{\bar\vz}_p^p) \\
	&\leq \eps^p (C_0^p OPT^p + C_z^p OPT^p) \\
	\normof{\vr - \vr'}_p
	&\leq (C_0^p + C_z^p)^{1/p} \eps OPT
\end{align*}
We lastly have to bound \(\normof{\mS\vr-\mS\vr'}_p\).
Recall that \mS is a diagonal matrix.
This lets us expand
\begin{align*}
	\normof{\mS(\vr-\vr')}_p^p
	&= \sum_{i=1}^n \mS_{ii} \abs{\vr(i)-\vr'(i)}^p \\
	&\leq \eps^p \sum_{i=1}^n \mS_{ii} (\abs{\vu(i)}^p + \abs{\bar\vz(i)}^p) \tag{By \equationref{compact-loss-powered}} \\
	&= \eps^p (\normof{\mS\vu}_p^p + \normof{\mS\bar\vz}_p^p) \\
	&\leq \eps^p (2^p\normof{\vu}_p^p + 100\normof{\bar\vz}_p^p) \tag{Subspace Embedding on \vu and Markov's Inequality on \(\bar\vz\)}\\
	&\leq \eps^p (2^pC_0^pOPT^p + 100 C_z^p OPT^p) \\
	\normof{\mS(\vr-\vr')}_p
	&\leq (2^pC_0^p + 100C_z^p)^{1/p} \eps OPT
\end{align*}
Which means we can overall bound
\begin{align*}
	\abs{\normof{\mS\vr}_p - \normof{\vr}_p}
	&\leq \abs{\normof{\mS\vr'}_p - \normof{\vr'}_p} + \normof{\vr-\vr'}_p + \normof{\mS\vr-\mS\vr'}_p \\
	&\leq (1 + 2(C_0^p + C_z^p)^{1/p} + 2(2^pC_0^p + 100C_z^p)^{1/p})\eps \cdot OPT
\end{align*}
\end{proof}

\begin{lemma}
\label{lem:net-suffices}
Let \(\cN_\eps\) be an \(\eps\)-Net on \(\{\mA\vx \,:\, \normof{\mA\vx}_p \leq C_0 OPT\}\), so that for any \(\mA\vx\) in this set there exists some \(\vu\in\cN_\eps\) such that \(\normof{\mA\vx-\vu}_p \leq \eps OPT\).
Consider the set of possible residual vectors \(\vr = \vu - \bar\vz\) for all \(\vu\in\cN_\eps\).
Suppose the sampling matrix \mS ensures that \(\normof{\mS\vr}_p^p \in \normof{\vr}_p^p \pm C_{\cN} \eps \cdot OPT\).
Then, for  all \vx with \(\normof{\mA\vx}_p \leq C_0 OPT\), \(\normof{\mS(\mA\vx-\bar\vz)}_p^p \in \normof{\mA\vx-\bar\vz}_p^p \pm C_3 \eps \cdot OPT^p\), where \(C_3\) is a constant that depends only on \(C_0, C_z, C_{\cN},\) and \(p\).
\end{lemma}
\begin{proof}
Fix any \(\vx\) with \(\normof{\mA\vx}_p \leq C_0 OPT\).
Let \(\vu\in\cN_\eps\) such that \(\normof{\mA\vx-\vy}_p \leq \eps OPT\).
Then, by triangle inequality
\begin{align*}
	\abs{\normof{\mS(\mA\vx-\bar\vz)}_p - \normof{\mA\vx-\bar\vz}_p}
	&\leq \abs{\normof{\mS(\vu-\bar\vz)}_p - \normof{\vu-\bar\vz}_p} + \normof{\mS(\mA\vx-\vu)}_p + \normof{\mA\vx-\vu}_p \\
	&\leq C_{\cN}\eps OPT + 3\normof{\mA\vx-\vu}_p \tag{\mS is a subspace embedding} \\
	&\leq (C_{\cN} + 3) \eps \cdot OPT
\end{align*}
Note that for \(a,b\in[0,\frac12]\) and \(p\geq2\), we have \(\abs{a^p - b^p} \leq \abs{a-b}\).
Therefore, for any \(c,d>0\), by setting \(a=\frac{c}{2\max\{c,d\}}\) and \(b=\frac{d}{2\max\{c,d\}}\) and simplifying, we get \(\abs{c^p - d^p} \leq (2\max\{c,d\})^{p-1} \abs{c-d}\).
In our setting, we note that \(\normof{\mA\vx-\bar\vz}_p \leq \normof{\mA\vx}_p + \normof{\bar\vz}_p \leq (C_0 + C_z) OPT\).
Further, \(\normof{\mS(\mA\vx-\bar\vz)}_p \leq \normof{\mA\vx-\bar\vz}_p + \eps (3+C_{\cN}) OPT \leq (C_0 + 3 + C_{\cN}) OPT\).
So, letting \(C_s \defeq (C_0 + C_z + 3 + C_{\cN})\), we have \(\max\{\normof{\mS(\mA\vx-\bar\vz)}_p, \normof{\mA\vx-\bar\vz}_p\} \leq C_s OPT\), and so
\begin{align*}
	\abs{\normof{\mS(\mA\vx-\bar\vz)}_p^p - \normof{\mA\vz-\bar\vz}_p^p}
	\leq (2C_s OPT)^{p-1} \abs{\normof{\mS(\mA\vx-\bar\vz)}_p - \normof{\mA\vz-\bar\vz}_p}
	\leq (2C_s)^{p-1}C' \eps \cdot OPT^p
\end{align*}
\end{proof}
This concludes the proof of \theoremref{affine-embed}.

\subsection{Lower Bounds for \texorpdfstring{\(L_p\)}{Lp} Regression} 
We now show that \((1+\eps)\)-approximation for \(L_p\) regression requires reading at least \(\Omega\left(\frac{1}{\eps^{p-1}}\right)\) entries of the function \(f\).
Later, in \sectionref{L_infty}, we show that even \(2\)-approximation for \(L_\infty\) regression requires reading \(\Omega(n)\) entries of \(f\).

\begin{theorem}
Fix \(p > 1\).
Any algorithm that can output a \((1+\eps)\) approximation to \(L_p\) polynomial regression with probability at least \(\frac23\) must use \(n=\Omega(\frac{1}{\eps^{p-1}})\) queries.
\end{theorem}
\begin{proof}
Suppose an algorithm uses \(n\leq\frac{1}{4\eps^p}\) queries.
Then there must exist an interval \(\cI\subset[-1,1]\) of width \(\frac1{4n}\) such that none of the algorithm's queries land within \cI with probability \(\frac23\).
We then define two functions:
\[
    f_+(t) \defeq \begin{cases} +\frac{2^{1/p}}{\eps} & t\in\cI \\ 0 & t\notin\cI \end{cases}
    \hspace{1.5cm}
    f_-(t) \defeq \begin{cases} -\frac{2^{1/p}}{\eps} & t\in\cI \\ 0 & t\notin\cI \end{cases}
\]
Both \(f_+\) and \(f_-\) have \(\normof{f_+}_p^p=\normof{f_-}_p^p=\frac{1}{4n} \cdot \frac{2}{\eps^p}\).
Let \(C \defeq 2^{-1/p} - \frac12 \in (0, \frac12)\).
Then both functions have \(\min_{\deg(q)\leq d}\normof{q-f}_p^p \leq (1-C\eps) \normof{f}_p^p\), since the polynomials \(q_+(t)\defeq1\) and \(q_-(t)\defeq-1\) achieve this \(L_p\) norm:
\begin{align*}
    \normof{q_+-f_+}_p^p
    &= \tsfrac{1}{4n} (\tsfrac{2^{1/p}}{\eps} - 1)^p + (1-\tsfrac1{4n})(0-1)^p \\
    &\leq \tsfrac{1}{4n} \left( (\tsfrac{2^{1/p}}{\eps} - 1)^p + (4n-0) \right) \\
    &= \tsfrac{1}{4n} \left( \tsfrac{2}{\eps^p} (1 - \tsfrac{\eps}{2^{1/p}})^p + 4n \right) \\
    &\leq \tsfrac{1}{4n} \left( \tsfrac{2}{\eps^p} (1 - \tsfrac{\eps}{2^{1/p}}) + \tsfrac1{\eps^p} \right) \\
    &= \tsfrac{1}{4n} \cdot \tsfrac2{\eps^p} \left( 1 - (2^{-1/p} - \tsfrac12)\, \eps \right) \\
    &= (1-C\eps) \normof{f_+}_p^p
\end{align*}
Or equivalently, \(\normof{f_+}_p^p \geq \frac{1}{1-C\eps} \min_{\deg(q)\leq d}\normof{q-f}_p^p > (1+C\eps)\min_{\deg(q)\leq d}\normof{q-f}_p^p\).
Now suppose some polynomial \(\hat q\) has \(\normof{\hat q - f_+}_p^p \leq (1-\gamma) \normof{f_+}_p^p\).
Since \(\normof{f_+-f_-}_p = 2\normof{f_+}_p\), we have
\begin{align*}
    \normof{\hat q - f_-}_p
    &\geq \normof{f_+-f_-}_p - \normof{\hat q - f_+}_p \\
    &= 2\normof{f_+}_p - (1-\gamma)^{1/p} \normof{f_+}_p \\
    &= (2-(1-\gamma)^{1/p}) \normof{f_-}_p \\
    &\geq (1+\tsfrac{\gamma}{p}\eps) \normof{f_-}_p \\
    \normof{\hat q - f_-}_p^p &\geq (1+\gamma) \normof{f_-}_p^p
\end{align*}
That is, if \(\hat q\) is a slightly good approximation to \(f_+\), then \(\hat q\) is a slightly bad approximation to \(f_-\).
By symmetry, the inverse claim also holds.

To complete the argument, suppose nature picks \(f_+\) or \(f_-\) uniformly at random.
Then with probability \(\frac23\) the algorithm returns some polynomial \(\hat q\) without knowing which function nature chose.
If \(\normof{\hat q - f_+}_p^p \leq \normof{f_+}_p^p\) then \(\normof{\hat q - f_-}_p^p \geq \normof{f_-}_p^p\), and otherwise \(\normof{\hat q - f_+}_p^p > \normof{f_+}_p^p\).
So, with probability \(\frac23\cdot\frac12 = \frac13\) the resulting polynomial has error
\[
    \normof{\hat q - f}_p^p
    \geq \normof{f}_p^p
    > (1+C\eps) \min_{\deg(q)\leq d}\normof{q-f}_p^p
\]
By adjusting the value of \(\eps\), we complete the proof.
\end{proof}

\section{Near-Optimal \texorpdfstring{\(L_\infty\)}{L∞} Regression}
\label{sec:L_infty}
We now demonstrate how to extend these guarantees from \(L_p\) polynomial regression into \(L_\infty\) polynomial regression. 
We remark that the sample complexity and approximation factor guarantees in this section were already shown in \cite{KaneKP17}, but with a different algorithm.

For a finite dimensional regression problem with $m$ rows, we could achieve a $(1+\eps$)-approximation to \(\ell_\infty\) regression by approximately solving \(\ell_p\) regression with \(p=\frac{\log m}{\eps}\) \cite{meyer2021}.
However, since polynomials lie within an infinite dimensional space, we cannot na\"{i}vely apply this argument. 
In fact, it can be shown that even with arbitrarily many observations it is impossible to solve polynomial \(L_\infty\) regression to better than a 2-factor approximation:

\begin{theorem}
\label{thm:lb:poly:infty}
There does not exist an algorithm that can output a $2$-approximation to $L_\infty$ polynomial regression with probability at least $\frac{2}{3}$. 
\end{theorem}
\begin{proof}
Consider an algorithm that observes at most a finite number, say $n<\infty$, of queries from \(f\). 
Then there exists some interval \(\cI\subset[-1,1]\) of nonzero width such that none of the algorithm's queries land within \cI with probability \(\frac23\). 
We then define two functions:
\[
    f_+(t) \defeq \begin{cases} +1 & t\in\cI \\ 0 & t\notin\cI \end{cases}
    \hspace{1.5cm}
    f_-(t) \defeq \begin{cases} -1 & t\in\cI \\ 0 & t\notin\cI \end{cases}
\]
Both \(f_+\) and \(f_-\) have \(\normof{f_+}_\infty=\normof{f_-}=1\), and both have \(\min_{\deg(q)\leq d} \normof{q-f}_\infty \leq \frac12\), since the polynomials \(q_+(t) \defeq \frac12\) and \(q_-(t)\defeq-\frac12\) achieve uniform error \(\frac12\).

To complete the argument, suppose nature picks \(f_+\) or \(f_-\) uniformly at random.
Then with probability \(\frac23\) the algorithm returns some polynomial \(\hat q\) without knowing which function nature chose.
If \(\hat q(t) \geq 0\) anywhere on \cI then \(\normof{q - f_-}_\infty \geq 1\), and if \(\hat q(t) \leq 0\) anywhere on \cI then \(\normof{q - f_+}_\infty \geq 1\).
So, with probability \(\frac23 \cdot \frac12 = \frac13\) the resulting vector has error \(\normof{\hat q - f} \geq 1 \geq 2 \min_{q:\deg(q)\leq d} \normof{q-f}_\infty\).
\end{proof}

In light of the lower bound in \theoremref{lb:poly:infty}, we aim to provide a constant-factor approximation for $L_\infty$ polynomial regression rather than $(1+\eps)$-approximation. 
This requires a slightly different algorithm than \algorithmref{chebyshev-const:Lp}, shown below in \algorithmref{chebyshev-const:L_infty}.
The only changes are that the rescaling matrix now has \(p\) in the numerator, and that \vx is computed by \(\ell_\infty\) matrix regression.

\begin{algorithm}[!tbh]
\caption{Chebyshev sampling for \(L_\infty\) polynomial regression}
\label{alg:chebyshev-const:L_infty}
\begin{algorithmic}[1]
\Require{Access to signal \(f\), parameter \(p \geq 1\), degree \(d\), number of samples \(n\)}
\Ensure{Degree \(d\) polynomial \(p(t)\)}
\State{Sample \(t_1,\ldots,t_n\in[-1,1]\) i.i.d.~from the pdf \(\frac{1}{\pi\sqrt{1-t^2}}\)}
\State{Observe signal samples \(b_i \defeq f(t_i)\) for all \(i\in[n]\)}
\State{Build \(\mA\in\bbR^{n\times(d+1)}\) and diagonal \(\mR\in\bbR^{n \times n}\) with \([\mA]_{i,j}=t_i^{j-1}\) and \([\mR]_{ii}=\big(\frac{dp}{n}\sqrt{1-t_i^2}\big)^{\nicefrac1p}\)}
\State{Compute \(\vx=\argmin_{\vx\in\bbR^{d+1}} \normof{\mR\mA\vx-\mR\vb}_\infty\)}
\State{Return \(p(t) = \sum_{i=0}^d x_i t^i\)}
\end{algorithmic}
\end{algorithm}

\begin{theorem}
\label{thm:main:infty}
Let \(B(n,r)\) denote the binomial distribution.
Let \(n_0 = O(d^5 \polylog d)\) and let \(p = O(\log d)\).
Suppose an algorithm samples \(n\sim B(n_0, 1/\tilde O(d^4))\) and runs \algorithmref{chebyshev-const:L_infty}.
Then, with probability \(\frac23\), the resulting polynomial \(\hat q\) satisfies
\[
    \normof{\hat q - f}_\infty \leq O(1) \min_{q:\deg(q)\leq d} \normof{q - f}_\infty
\]
\end{theorem}

We prove this by mirroring a known proof technique found in Appendix A of \cite{ParulekarPP21}, which says that having a subspace embedding suffices to constant-factor approximation guarantees in any normed space.
So, to apply this proof technique, we first have to have a subspace embedding in the \(L_\infty\) norm:
\begin{lemma}
Suppose an algorithm samples \(n \sim B(d^4, O(\frac1{d^3}))\) and runs \algorithmref{chebyshev-const:L_infty}.
Then, the matrix \(\mR\mA\) on line 4 of the algorithm is a subspace embedding for \cP: \(\frac1C \normof{\cP\vx}_\infty \leq \normof{\mR\mA\vx}_\infty \leq C\normof{\cP\vx}_\infty\) for all \(\vx\in\bbR^{d+1}\).
\end{lemma}

This is the conclusion of two shorter lemmas
\begin{lemma}
Let \(p>2\) be an integer.
Suppose an algorithm samples \(n \sim B(d^4, O(\frac1{d^3}))\) and runs \algorithmref{chebyshev-const:L_infty}.
Then, the matrix \(\mR\mA\) on line 4 of the algorithm is a subspace embedding for \cP: \(\frac1C \normof{\cP\vx}_p^p \leq \normof{\mR\mA\vx}_p^p \leq C\normof{\cP\vx}_p^p\) for all \(\vx\in\bbR^{d+1}\).
\end{lemma}
\begin{proof}
We start by using the same trick as \theoremref{vander:p:q:sens} in \sectionref{largep:regression:eps} to build a subspace embedding for large \(p\).
Let \(\cQ:\bbR^{dp+1}\rightarrow L_1([-1,1])\) be the extended polynomial operator, so that \([\cQ\vv](t) = \sum_{i=0}^{dp} x_i t^i\).

Notice that \((\cP\vx)^p\) is just some polynomial raised to integer power \(p\).
So, for any \(\vx\in\bbR^{d+1}\), there exists a \(\vv\in\bbR^{dp+1}\) such that \((\cP\vx)^p = \cQ\vv\).
Then, we can write
\[
    \normof{\cP\vx}_p^p
    = \int_{-1}^1 |[\cP\vx](t)|^p dt
    = \int_{-1}^1 |[\cQ\vv](t)| dt
    = \normof{\cQ\vv}_1
\]
We then apply \corolref{Lp-subspace-embedding} to the \(L_1\) norm for polynomials of degree \(dp\).
This tells us that diagonal \(\mS\in\bbR^{n \times n}\) with \([\mS]_{ii} = \frac{dp}{n} \sqrt{1-t_i^2}\) and Vandermonde \(\mB\in\bbR^{n \times (dr+1)}\) with \([\mB]_{ij} = t_i^j\) enjoy
\[
    \frac1C \normof{\cQ\vv}_1 \leq \normof{\mS\mB\vv}_1 \leq C \normof{\cQ\vv}_1
    \hspace{1cm}
    \text{for all }
    \vv\in\bbR^{dp+1}
\]
Since \mA and \mB are just Vandermonde matrices of degree \(d\) and \(dp\) respectively, we can use the same observation to equate \(\normof{\mS\mB\vv}_1 = \normof{\mR\mA\vx}_p^p\):
\[
    \normof{\mS\mB\vv}_1
    = \sum_{i=1}^n \mS_{ii} ~ |[\cQ\vv](t_i)|
    = \sum_{i=1}^n \left(\mR_{ii} ~ |[\cP\vx](t_i)|\right)^p
    = \normof{\mR\mA\vx}_p^p
\]
Where we use the fact that \([\mS]_{ii} = [\mR]_{ii}^p = \frac{dp}{n}\sqrt{1-t_i^2}\).
So, the subspace embedding guarantee is equivalent to
\[
    \frac1C \normof{\cP\vx}_p^p \leq \normof{\mR\mA\vx}_p^p \leq C \normof{\cP\vx}_p^p
    \hspace{1cm}
    \text{for all }
    \vx\in\bbR^{d+1}
\]
This complete the process of making a subspace embedding for large \(p\).
\end{proof}

Next, we take \(p = O(\log d)\) and show this creates an \(L_\infty\) subspace embedding.
\begin{lemma}
Let \(p = O(\log d)\) be an integer.
Suppose an algorithm samples \(n \sim B(d^5, \tilde O(\frac1{d^4}))\) and runs \algorithmref{chebyshev-const:L_infty}.
Then, the matrix \(\mR\mA\) on line 4 of the algorithm is a subspace embedding for \cP: \(\frac1C \normof{\cP\vx}_\infty \leq \normof{\mR\mA\vx}_\infty \leq C\normof{\cP\vx}_\infty\) for all \(\vx\in\bbR^{d+1}\).
\end{lemma}
\begin{proof}
We achieve this by showing \(\normof{\cP\vx}_p\approx_{O(1)}\normof{\cP\vx}_\infty\) and \(\normof{\mR\mA\vx}_p\approx_{O(1)}\normof{\mR\mA\vx}_\infty\) for all \vx.

This is simple to show in the finite dimensional case.
By standard finite dimensional \(\ell_p\) norm inequalities, \(\normof{\mR\mA\vx}_\infty \leq \normof{\mR\mA\vx}_p \leq n^{\frac1p} \normof{\mR\mA\vx}_\infty\).
Since \(n = \tilde O(d)\), having \(p = O(\log d)\) suffices for \(n^{\frac1p}\) to be \(O(1)\).

The infinite dimension case is more involved.
We need to show that for any polynomial \(h(t)\) of degree \(d\), we have \(\normof{h}_\infty \approx_c \normof{h}_p\).
One direction is simple to show:
\[
    \normof{h}_p^p = \int_{-1}^1 |h(t)|^p dt \leq 2 \normof{h}_\infty
\]
The other direction follows from the Markov Brothers' Inequality, using an argument similar to \lemmaref{uniform-sensitivity-bound}.
Without loss of generality assume that \(\normof{h}_\infty=1\), and that \(h(t_0)=1\) for some \(t_0<0\).
Then, by Markov Brothers', we have \(|h(t_0+x)|\geq 1-d^2x\) for any \(0 < x < \frac1{d^2}\).
In particular, we have \(|h(t)| > 1-\frac1d\) for \(t\in[t_0, t_0+\frac1{d^3}]\).
Then,
\begin{align*}
    \normof{h}_p
    &= \left(\int_{-1}^1 |h(t)|^p dt\right)^{1/p} \\
    &\geq \left(\frac1{d^3} (1-\tsfrac1d)^p\right)^{1/p} \\
    &= \frac1{d^{3/p}} (1-\tsfrac1d) \\
    &= \Omega(1)
\end{align*}
Where the last line follows from \(d\geq2\) and \(p = O(\log d)\), so that \(1-\frac1d\geq\frac12\) and \(d^{3/p} = O(1)\).
We conclude that \(\normof{h}_p = \Omega(1) = \Omega(1) \normof{h}_\infty\), and therefore that \(\normof{\cP\vx}_p \approx_C \normof{\cP\vx}_\infty\).

Then, we finally combine this with the subspace embedding from the prior lemma to get
\[
    \frac1C \normof{\cP\vx}_\infty \leq \normof{\mR\mA\vx}_\infty \leq C \normof{\cP\vx}_\infty
\]
\end{proof}

Now that we have a subspace embedding, we can complete the proof that \algorithmref{chebyshev-const:L_infty} is correct.
\begin{proof}
Let \(\vx^* \defeq \argmin_{\vx} \normof{\cP\vx-f}_\infty\) a true optimal solution.
We first bound \(\normof{\mR\mA\vx^* - \mR\vb}_\infty \leq C \normof{\cP\vx^*-f}_\infty\):
\begin{align*}
    [R]_{ii}
    &= (\textstyle{\frac{dp}{n} \sqrt{1-t_i^2}})^{1/O(\log d)} \\
    &\leq (\Theta(\textstyle{\frac{d\cdot\log d}{d \polylog d}}))^{1/O(\log d)} \\
    &= (\Theta(\textstyle{\frac{1}{\polylog d}}))^{1/O(\log d)} \\
    &= O(1) \\
    \normof{\mR\mA\vx^* - \mR\vb}_\infty
    &= \sup_{i\in[n]} \Big| R_{ii} \, [\mA\vx^* - \vb]_i\Big| \\
    &= \sup_{i\in[n]} \Big| R_{ii} \, ([\cP\vx^* - f](t_i)) \Big|  \\
    &\leq \sup_{t\in[-1,1]} O(1) \Big| [\cP\vx^*](t_i) - f(t_i) \Big|  \\
    & = O(1) \normof{\cP\vx^* - f}_\infty \numberthis\label{eq:l_infty_centering} \\
\end{align*}
And this bound suffices to prove our guarantee.
Let \(\hat\vx\defeq \argmin_{\vx} \normof{\mR\mA\vx-\mR\vb}_\infty\) be the solution returned in line 4 of \algorithmref{chebyshev-const:L_infty}.
Then,
\begin{align*}
    \normof{\cP\hat\vx-f}_\infty
    &\leq \normof{\cP\hat\vx-\cP\vx^*}_\infty + \normof{\cP\vx^* - f}_\infty \\
    &\leq C \normof{\mR\mA\hat\vx-\mR\mA\vx^*}_\infty + \normof{\cP\vx^* - f}_\infty \tag{Subspace Embedding} \\
    &\leq C (\normof{\mR\mA\hat\vx-\mR\vb}_\infty + \normof{\mR\mA\vx^*-\mR\vb}_\infty) + \normof{\cP\vx^* - f}_\infty \\
    &\leq 2C \normof{\mR\mA\vx^*-\mR\vb}_\infty + \normof{\cP\vx^* - f}_\infty \tag{Optimality of \(\hat\vx\)} \\
    &\leq O(1) \normof{\cP\vx^*-f}_\infty + \normof{\cP\vx^* - f}_\infty \tag{\equationref{l_infty_centering}}\\
    &= O(1) \normof{\cP\vx^*-f}_\infty
\end{align*}
Which completes the proof.
\end{proof}

\section{Analysis of the Clipped Chebyshev Measure}
\label{sec:clipped-analysis}

As mentioned in \sectionref{lp-regression}, the Chebyshev measure itself is not sufficient to achieve the approximate Lewis weight property for \cP, since the Chebyshev measure grows to infinity as \(\abs{t}\rightarrow1\) while the leverage function is bounded.
Thus we instead analyze the following clipped measure: \(w(t) \defeq \min\{C(d+1)^2,v(t)\} = \min\{C(d+1)^2,\frac{d+1}{\pi\sqrt{1-t^2}}\}\) and prove the following result:
\begin{reptheorem}{chebyshev:ratio}
There are fixed constants $c_1,c_2$ such that, for all $p\in[\frac23,2]$ and $t \in [-1,1]$,
\[
	\frac{c_1}{\log^3 d}
	\leq
	\frac{\tau[\cW^{\frac12-\frac1p}\cP](t)}{w(t)}
	\leq
	c_2.
\]
\end{reptheorem}
The basic flow of the proof is broken into two portions.
First, recall the overall shape of the rescaled leverage function:
\begin{align}
    \tau[\cW^{\frac12-\frac1p}\cP](t)
    = \max_{\vx\in\bbR^{d+1}} \frac{([\cW^{\frac12-\frac1p}\cP\vx](t))^2}{\normof{\cW^{\frac12-\frac1p}\cP\vx}_2^2}
    = (w(t))^{1-\frac2p} \max_{q:\text{deg}(q)\leq d} \frac{(q(t))^2}{\normof{\cW^{\frac12-\frac1p}q}_2^2}
    \label{eq:clipped-lewis-weight-expansion}
\end{align}
We need to show that this leverage function is close to \(w(t)\) for all \(t\in[-1,1]\).
We split this analysis into two parts:
\begin{enumerate}
    \item
    \textbf{The ``Middle Region''} with \(w(t)=v(t)\), so that  \(\abs{t} \leq 1-O(\frac1{d^2})\): \\
    We show in \sectionref{one:mid} that \(\normof{\cW^{\frac12-\frac1p}\cP\vx}_2^2\) and \(\normof{\cV^{\frac12-\frac1p}\cP\vx}_2^2\) are similar enough that \(\tau[\cW^{\frac12-\frac1p}\cP] \approx \tau[\cV^{\frac12-\frac1p}\cP]\) in this region, and so the analysis of \theoremref{chebyshev:ratio} is tight enough to ensure the almost Lewis weight property here.
    \item
    \textbf{The ``Endcap Region''} with \(w(t)=C(d+1)^2\), so that  \(\abs{t} \geq 1-O(\frac1{d^2})\): \\
    We know that \(w(t)\) and \(v(t)\) are very different here, so we use the fact that \(w(t)=C(d+1)^2\) is independent of \(t\).
    This endcap analysis also proceeds in two steps:
    \begin{itemize}
        \item
        \textbf{Upper bound \(\tau[\cW^{\frac12-\frac1p}\cP](t) \leq O(w(t)) = O(d^2)\):} \\
        In \sectionref{one:end:upper}, we note that \(w(t)\leq C(d+1)^2\) for all \(t\in[-1,1]\).
        We use this to lower bound \(\normof{\cW^{\frac12-\frac1p}\cP\vx}_2^2 \geq \frac1{C(d+1)^2} \normof{\cP\vx}_2^2\), and reduce the second form in \equationref{clipped-lewis-weight-expansion} to the unweighted leverage function for \cP.
        We appeal to our earlier bound on the leverage function for \cP from \sectionref{l2-leverage-bound}.
        \item
        \textbf{Lower bound \(\tau[\cW^{\frac12-\frac1p}\cP](t) \leq \Omega(w(t) \log^3 (d)) = \Omega(d^2\log^3(d))\):} \\
         In \sectionref{one:end:lower}, we plug in a spike polynomial that approximates \(t\mapsto t^{d^2}\) into the rightmost term in \equationref{clipped-lewis-weight-expansion}, and evaluate the numerator and denominator for that polynomial.
    \end{itemize}
    
\end{enumerate}
We again break up the analysis into the slightly more approachable \(p=1\) setting and the more complete \(p\in[\frac23,2]\) setting.
Additionally, in this section, we refer to the middle region as
\[
	\cI_{mid} \defeq \{t ~|~ w(t) = v(t)\} = {\textstyle \left[\sqrt{1-\frac{1}{\pi^2(d+1)^2C^2}}, \sqrt{1+\frac{1}{\pi^2(d+1)^2C^2}}\right]}
\]
and the endcap region as \(\cI_{cap} \defeq [-1,1] ~\backslash~ \cI_{mid}\).
We also often use the notation \(x \approx_{\alpha} y\) with \(\alpha \geq 1\) to mean that \(\frac1\alpha y \leq x \leq \alpha y\).
Lastly, to reduce the messiness of the analysis, we omit the change-of-basis matrix that was used in prior sections \mU.
For \(p=1\) analysis, Chebyshev polynomials of the second kind are used.
For \(p=2\) analysis, Legendre polynomials are used.
For \(p\in(\frac23,2)\) analysis, Ultraspherical (i.e. Jacobi) polynomials are used.

As an aside, when \(p>2\) this analysis breaks down in a few places since \(\frac12-\frac1p\) swaps from being negative to positive.
For instance, this means that \(t \mapsto (w(t))^{\frac12-\frac1p}\) is maximized in the middle region for \(p<2\) but is maximized in the endcap for \(p>2\).

\subsection{Middle Region Analysis for \texorpdfstring{\(p=1\)}{p=1}}
\label{sec:one:mid}

Our main goal in this section is to prove \lemmaref{mid}, which states that $\frac{\tau[\cW^{-\frac12}\cP](t)}{w(t)}=\Theta(1)$.
We first recall our bound on the leverage function in \sectionref{l2-leverage-bound}:
\begin{lemma}
\label{lem:unscaled-leverages}
The leverage function for \cP has \(\tau[\cP](t) \leq \frac{(d+1)^2}{2}\) for all \(t\in[-1,1]\).
\end{lemma}
\noindent We use this lemma to (1) analyze the behavior of $\frac{\tau[\cW^{-\frac12}\cP](t)}{w(t)}$ on $\cI_{mid}$ by showing that the leverage functions on the operators \(\tau[\cW^{-\frac12}\cP](t)\) and \(\tau[\cV^{-\frac12}\cP](t)\) are very similar inside this middle region $\cI_{mid}$ in \sectionref{tauw:tauv} and (2) upper and lower bound the ratio of $\frac{\tau[\cV^{-\frac12}\cP](t)}{v(t)}$ in \sectionref{tauV:vt}. 
Using these bounds, we then prove \lemmaref{mid} in \sectionref{complete:mid}.

\subsubsection{Relating \texorpdfstring{\(\tau[\cW^{-\frac12}\cP]\) to \(\tau[\cV^{-\frac12}\cP]\)}{tau(W) to tau(V)}}
\label{sec:tauw:tauv}
In this section, our main goal is to show in \corolref{oper:approx:mid} that $\frac{\tau[\cW^{-\frac12}\cP](t)}{w(t)}\approx_{\frac{2}{\pi^2C^2}}\frac{\tau[\cV^{-\frac12}\cP](t)}{v(t)}$ for $t\in\cI_{mid}$, where we recall that $\cI_{mid}$ is defined by
\[
	\cI_{mid} \defeq \{t ~|~ w(t) = v(t)\} = {\textstyle \left[\sqrt{1-\frac{1}{\pi^2(d+1)^2C^2}}, \sqrt{1+\frac{1}{\pi^2(d+1)^2C^2}}\right]},
\]
so that \(w(t)=v(t)\) for $t\in\cI_{mid}$. 
To this end, we first remark that it suffices to show that \(\normof{\cW^{-\frac12}\cP\vx}_2^2 \approx_{\frac{2}{\pi^2C^2}} \normof{\cV^{-\frac12}\cP\vx}_2^2\). 
To see why this suffices, consider the definitions of the leverage functions:
\[
	\tau[\cW^{-\frac12}\cP](t)
	= \max_{\vx} \frac{([\cW^{-\frac12}\cP\vx](t))^2}{\normof{\cW^{-\frac12}\cP\vx}_2^2}
	= \max_{\vx} \frac{([\cV^{-\frac12}\cP\vx](t))^2}{\normof{\cW^{-\frac12}\cP\vx}_2^2}
	\approx_{\frac{2}{\pi^2C^2}} \max_{\vx} \frac{([\cV^{-\frac12}\cP\vx](t))^2}{\normof{\cV^{-\frac12}\cP\vx}_2^2}
	= \tau[\cV^{-\frac12}\cP](t).
\]
Hence, we first show in \lemmaref{oper:norm:approx} that \(\normof{\cW^{-\frac12}\cP\vx}_2^2 \approx_{\frac{\pi^2C^2}{\pi^2C^2-1}} \normof{\cV^{-\frac12}\cP\vx}_2^2\). 

\begin{lemma}
\label{lem:oper:norm:approx}
For all $\vx\in\mathbb{R}^{d+1}$, we have
\[
	\normof{\cW^{-\frac12}\cP\vx}_2^2 \approx_{\frac{\pi^2C^2}{\pi^2C^2-1}} \normof{\cV^{-\frac12}\cP\vx}_2^2
\]
\end{lemma}
\begin{proof}
We start by looking at the difference between $\normof{\cW^{-\frac12}\cP\vx}_2^2$ and $\normof{\cV^{-\frac12}\cP\vx}_2^2$. 
\begin{align*}
	\bigg|\normof{\cW^{-\frac12}\cP\vx}_2^2 &- \normof{\cV^{-\frac12}\cP\vx}_2^2\bigg|
		= \abs{ \int_{-1}^1 ([\cW^{-\frac12}\cP\vx](t))^2 - ([\cV^{-\frac12}\cP\vx](t))^2 ~ dt } \\
		&= \abs{ \int_{\cI_{cap}} ([\cW^{-\frac12}\cP\vx](t))^2 - ([\cV^{-\frac12}\cP\vx](t))^2 ~ dt+\int_{\cI_{mid}} ([\cW^{-\frac12}\cP\vx](t))^2 - ([\cV^{-\frac12}\cP\vx](t))^2 ~ dt}.
		\end{align*}
Since $[\cW^{-\frac12}\cP\vx](t)=[\cV^{-\frac12}\cP\vx](t)$ for $t\in\cI_{mid}$, then $\int_{\cI_{mid}} ([\cW^{-\frac12}\cP\vx](t))^2 - ([\cV^{-\frac12}\cP\vx](t))^2 ~ dt=0$. 
Moreover, since $w(t)$ is the clipped Chebyshev measure, we have that $w(t)\le v(t)$ and thus \((w(t))^{-\frac12} \geq (v(t))^{-\frac12}\). 
Hence,
\begin{align*}
		\abs{\normof{\cW^{-\frac12}\cP\vx}_2^2 - \normof{\cV^{-\frac12}\cP\vx}_2^2}
	&= \abs{ \int_{\cI_{cap}} ([\cW^{-\frac12}\cP\vx](t))^2 - ([\cV^{-\frac12}\cP\vx](t))^2 ~ dt } \\
	&\leq \abs{\int_{\cI_{cap}} ([\cW^{-\frac12}\cP\vx](t))^2 ~ dt}=\int_{\cI_{cap}} ([\cW^{-\frac12}\cP\vx](t))^2 ~ dt.
	\end{align*}
Because \((w(t))^{-1} = \frac1{C(d+1)^2}\) on \(\cI_{cap}\), then
\begin{align*}
\abs{\normof{\cW^{-\frac12}\cP\vx}_2^2 - \normof{\cV^{-\frac12}\cP\vx}_2^2} &\le \frac{1}{C(d+1)^2} \int_{\cI_{cap}} ([\cP\vx](t))^2 dt 
\end{align*}
Since \lemmaref{unscaled-leverages} implies \(([\cP\vx](t))^2 \leq \frac{(d+1)^2}{2}\normof{\cP\vx}_2^2\), then
\begin{align*}
\abs{\normof{\cW^{-\frac12}\cP\vx}_2^2 - \normof{\cV^{-\frac12}\cP\vx}_2^2}
	&\leq \frac{1}{C(d+1)^2} \cdot \frac{(d+1)^2}{2} \normof{\cP\vx}_2^2 \int_{\cI_{cap}} dt=\tsfrac{\normof{\cP\vx}_2^2}{2C} \int_{\cI_{cap}} dt.
\end{align*}
To upper bound the length of the interval \(\cI_{cap}\), note that \(1-\sqrt{1-\frac{1}{x^2}} \leq \frac1{x^2}\) for $x^2\ge 1$. 
Hence,
\[
	\textstyle\int_{\cI_{cap}} dt
	= 2 \cdot \left(1 - \sqrt{1-\frac{1}{\pi^2(d+1)^2C^2}}\right)
	\leq \frac{2}{\pi^2(d+1)^2C^2},
\]
so that
\[
\abs{\normof{\cW^{-\frac12}\cP\vx}_2^2 - \normof{\cV^{-\frac12}\cP\vx}_2^2}\leq \frac{\normof{\cP\vx}_2^2}{2C} \cdot \frac{2}{\pi^2(d+1)^2C^2}
	= \frac{1}{\pi^2(d+1)^2C^3}\normof{\cP\vx}_2^2.
	\]
Next, we bound the norm \(\normof{\cP\vx}_2^2\) using the fact that \(w(t) \leq C(d+1)^2\) to say that \(1 \leq \sqrt C(d+1) \cdot (w(t))^{-\frac12}\), so that
\begin{align*}
	\normof{\cP\vx}_2^2
	&= \int_{-1}^1 (1 \cdot [\cP\vx](t))^2 dt \\
	&\leq \int_{-1}^1 \left(\sqrt C(d+1) \cdot (w(t))^{-\frac12} \cdot [\cP\vx](t)\right)^2 dt \\
	&= C(d+1)^2 \normof{\cW^{-\frac12}\cP\vx}_2^2.
\end{align*}
Therefore,
\[
	\abs{\normof{\cW^{-\frac12}\cP\vx}_2^2 - \normof{\cV^{-\frac12}\cP\vx}_2^2}
	\leq \tsfrac{C(d+1)^2}{\pi^2(d+1)^2C^3}\normof{\cW^{-\frac12}\cP\vx}_2^2
	= \tsfrac{1}{\pi^2C^2}\normof{\cW^{-\frac12}\cP\vx}_2^2.
\]
Rearranging this inequality,
\[
	\abs{1 - \frac{\normof{\cV^{-\frac12}\cP\vx}_2^2}{\normof{\cW^{-\frac12}\cP\vx}_2^2}} \leq \frac1{\pi^2C^2}
\]
or equivalently,
\[
	1 - \frac{1}{\pi^2C^2} \leq \frac{\normof{\cV^{-\frac12}\cP\vx}_2^2}{\normof{\cW^{-\frac12}\cP\vx}_2^2} \leq 1 + \frac1{\pi^2C^2}
	\leq \frac{1}{1-\frac{1}{\pi^2C^2}},
\]
for $C>\frac{1}{\pi}$. 
Since \(\frac{1}{1-\frac{1}{\pi^2C^2}} = \frac{\pi^2C^2}{\pi^2C^2-1}\), then we have the multiplicative error guarantee
\[
	\normof{\cW^{-\frac12}\cP\vx}_2^2 \approx_{\frac{\pi^2C^2}{\pi^2C^2-1}} \normof{\cV^{-\frac12}\cP\vx}_2^2
\]
for \(C > \frac1\pi \approx 0.312\).
\end{proof}

\noindent
We now complete the formal proof of \corolref{oper:approx:mid}. 
\begin{corollary}
\label{corol:oper:approx:mid}
$\frac{\tau[\cW^{-\frac12}\cP](t)}{w(t)}\approx_{\frac{\pi^2C^2}{\pi^2C^2-1}}\frac{\tau[\cV^{-\frac12}\cP](t)}{v(t)}$ for $t\in\cI_{mid}$.
\end{corollary}
\begin{proof}
By \lemmaref{oper:norm:approx}, we have that \(\normof{\cW^{-\frac12}\cP\vx}_2^2 \approx_{\frac{\pi^2C^2}{\pi^2C^2-1}} \normof{\cV^{-\frac12}\cP\vx}_2^2\) for all $\vx\in\mathbb{R}^{d+1}$. 
Since \(w(t)=v(t)\) for $t\in\cI_{mid}$, \lemmaref{oper:norm:approx} implies through the definition of the leverage functions that
\begin{align*}
\frac{\tau[\cW^{-\frac12}\cP](t)}{w(t)}&=\frac{1}{w(t)}\max_{\vx} \frac{([\cW^{-\frac12}\cP\vx](t))^2}{\normof{\cW^{-\frac12}\cP\vx}_2^2}\\
&=\frac{1}{v(t)}\max_{\vx} \frac{([\cV^{-\frac12}\cP\vx](t))^2}{\normof{\cW^{-\frac12}\cP\vx}_2^2}\\
&\approx_{\frac{\pi^2C^2}{\pi^2C^2-1}}\frac{1}{v(t)}\max_{\vx} \frac{([\cV^{-\frac12}\cP\vx](t))^2}{\normof{\cV^{-\frac12}\cP\vx}_2^2}\\
&=\frac{\tau[\cV^{-\frac12}\cP](t)}{v(t)},
\end{align*}
as desired.
\end{proof}

\subsubsection{Relating \texorpdfstring{\(\tau[\cV^{-\frac12}\cP]\) to \(v(t)\)}{tau(V) to v(t)}}
\label{sec:tauV:vt}
In this section, we relate \(\tau[\cV^{-\frac12}\cP]\) to $v(t)$ for $t\in\cI_{mid}$, which will ultimately allow us to relate \(\tau[\cW^{-\frac12}\cP]\) to $w(t)$ in \sectionref{complete:mid}, using \corolref{oper:approx:mid}.  
\begin{lemma}
\label{lem:v:ratio:mid}
For \(t\in\cI_{mid}\), we have that \(\tau[\cV^{-\frac12}\cP](t) \approx_{(\frac54+\frac{\pi C}2)}v(t)\).
\end{lemma}
\begin{proof}
Note that the claim is equivalent to the statement that \(\frac{\tau[\cV^{-\frac12}\cP](t)}{v(t)} \in (\frac1\gamma,\gamma)\) for \(\gamma\leq\frac54+\frac{\pi C}2\).
We will use the relationship \(\frac{-1}{\sqrt{1-t^2}} \leq U_i(t) \leq \frac{1}{\sqrt{1-t^2}}\) to prove this.

Specifically, we ensure the two traits
\begin{alignat*}{3}
	\frac{\tau[\cV^{-\frac12}\cP](t)}{v(t)} &\leq 1 + \frac{1+\frac{1}{\sqrt{1-t^2}}}{2(d+1)} &\leq \gamma \\
	\frac{\tau[\cV^{-\frac12}\cP](t)}{v(t)} &\geq 1 + \frac{1-\frac{1}{\sqrt{1-t^2}}}{2(d+1)} &\geq \frac1\gamma.
\end{alignat*}
Solving these two inequalities on the right hand side yields
\[
	\textstyle \abs{t} \leq \sqrt{1-\frac{1}{(2(d+1)(\gamma-1)-1)^2}}
	\hspace{1cm}
	\text{and}
	\hspace{1cm}
	\textstyle \abs{t} \leq \sqrt{1-\frac{1}{(2(d+1)(\frac1\gamma-1)-1)^2}},
\]
respectively. 
Observe that the guarantee on the left implies the guarantee on the right, so we just ensure that one trait. 
Rather, we should think of
\[
	\cI_{\gamma} \defeq \left\{t ~\bigg|~ \abs{t} \leq {\textstyle\sqrt{1-\frac{1}{(2(d+1)(\gamma-1)-1)^2}}}\right\}
\]
as the set of time points where we have \(\tau[\cV^{-\frac12}\cP](t) \approx_\gamma v(t)\).
We now ensure that this interval \(\cI_\gamma\) entirely contains the middle region \(\cI_{mid}\), i.e., $\cI_{mid}\subset\cI_\gamma$. 
Note that $t\in\cI_{mid}$ implies that 
\[t\le\sqrt{1-\frac{1}{\pi^2(d+1)^2C^2}}.\] 
For $\gamma=1 + \frac{\pi}{2}C + \frac{1}{2(d+1)}$, note that we have 
\[t\le\sqrt{1-\frac{1}{(2(d+1)(\gamma-1)-1)^2}},\]
as desired. 
Hence, $\cI_{mid}=\cI_\gamma$ for $\pi^2(d+1)^2C^2=(2(d+1)(\gamma-1)-1)^2$ or equivalently, $\gamma=1 + \frac{\pi}{2}C + \frac{1}{2(d+1)}$. 
Since $d\ge 1$ implies
\[1 + \frac{\pi}{2}C + \frac{1}{2(d+1)}\leq 1 + \frac{\pi}{2}C + \frac{1}{4} = \frac54 + \frac{\pi}{2}C,\]
then $\cI_{mid}\subset\cI_\gamma$ for $\gamma\le\frac54 + \frac{\pi}{2}C$. 
Therefore, the set \(\cI_\gamma\) where the leverage scores of \(\cV^{-\frac12}\cP\) are \(\gamma\)-close to \(v(t)\) covers the set of time-samples \textit{not} in the cap for \(\gamma \leq \frac54 + \frac{\pi}{2}C\). 
Equivalently, we have that \(\tau[\cV^{-\frac12}\cP](t) \approx_{(\frac54+\frac{\pi C}2)}v(t)\) for \(t\in\cI_{mid}\). 
\end{proof}

\subsubsection{Complete Result in the Middle}
\label{sec:complete:mid}
We now finally relate $\frac{\tau[\cW^{-\frac12}\cP](t)}{w(t)}$ by using \corolref{oper:approx:mid} and \lemmaref{v:ratio:mid}.  
\begin{lemma}
\label{lem:mid}
For $t\in\cI_{mid}$, we have
\[
	\frac{\tau[\cW^{-\frac12}\cP](t)}{w(t)}=\Theta(1).
\]
\end{lemma}
\begin{proof}
By \corolref{oper:approx:mid} and \lemmaref{v:ratio:mid}, we have that for $t\in\cI_{mid}$,
\[
	\tau[\cW^{-\frac12}\cP](t) \approx_\alpha v(t)
\]
where \(\alpha = \frac{\pi^2C^2}{\pi^2C^2-1} \cdot \left(\frac54 + \frac{\pi C}{2}\right)\) for some constant \(C > \frac1\pi\approx0.312\).
Furthermore, since \(v(t)=w(t)\) in the region $t\in\cI_{mid}$, this further implies \(\tau[\cW^{-\frac12}\cP](t) \approx_\alpha w(t)\), as desired.
\end{proof}

\subsection{Endcap Region Analysis for \texorpdfstring{\(p=1\)}{p=1}}
\label{sec:one:end}

We now turn to \(t\in\cI_{cap}\), and we will show that
\[
	\tau[\cW^{-\frac12}\cP](t) \approx_{\frac1{2C}} w(t)
\]
for \(t\in\cI_{cap}\). 
Thus it suffices to upper and lower bound the ratio $\frac{\tau[\cW^{-\frac12}\cP](t)}{w(t)}$.

\subsubsection{Upper Bounding the Ratio.}
\label{sec:one:end:upper}
In this section, we provide an upper bound on the ratio $\frac{\tau[\cW^{-\frac12}\cP](t)}{w(t)}$. 
Namely, we show in \lemmaref{cap:upper} that there exists an absolute constant $C$, the same constant $C>\frac{1}{\pi}$ in the definition of the clipped Chebyshev measure, such that $\frac{\tau[\cW^{-\frac12}\cP](t)}{w(t)}\le\frac{1}{2C}$ for all $t\in\cI_{cap}$. 
\begin{lemma}
\label{lem:cap:upper}
For $t\in\cI_{cap}$, we have
\[
	\frac{\tau[\cW^{-\frac12}\cP](t)}{w(t)}=O(1).
\]
\end{lemma}
\begin{proof}
Since $\tau[\cW^{-\frac12}\cP](t) = \max_{\vx}\frac{([\cW^{-\frac12}\cP\vx](t))^2}{\normof{\cW^{-\frac12}\cP\vx}_2^2}$ and $w(t)\le C(d+1)^2$ for all $t\in[-1,1]$, we first lower bound \(\normof{\cW^{-\frac12}\cP\vx}_2^2\) by
\begin{align*}
	\normof{\cW^{-\frac12}\cP\vx}_2^2
	&= \int_{-1}^1 \frac1{w(t)} ([\cP\vx](t))^2 \ dt \\
	&\geq \int_{-1}^1 \frac1{C(d+1)^2} ([\cP\vx](t))^2 \ dt = \frac1{C(d+1)^2} \normof{\cP\vx}_2^2. 
\end{align*}
Then we can directly tackle the leverage function:
\begin{align*}
	\tau[\cW^{-\frac12}\cP](t) &= \max_{\vx}\frac{([\cW^{-\frac12}\cP\vx](t))^2}{\normof{\cW^{-\frac12}\cP\vx}_2^2} = \frac1{w(t)} \max_{\vx}\frac{([\cP\vx](t))^2}{\normof{\cW^{-\frac12}\cP\vx}_2^2} \\
	&= \frac1{C(d+1)^2} \max_{\vx}\frac{([\cP\vx](t))^2}{\normof{\cW^{-\frac12}\cP\vx}_2^2}
\end{align*}
since $w(t)=C(d+1)^2$ for $t\in\cI_{cap}$. 
Thus, 
\begin{align*}
	\tau[\cW^{-\frac12}\cP](t) &\leq \frac1{C(d+1)^2} \max_{\vx}\frac{([\cP\vx](t))^2}{\frac{1}{C(d+1)^2}\normof{\cP\vx}_2^2} = \tau[\cP](t) \leq \frac{(d+1)^2}{2}.
\end{align*}
Then we can then conclude
\[
	\frac{\tau[\cW^{-\frac12}\cP](t)}{w(t)}
	\leq \frac{\frac{(d+1)^2}{2}}{C(d+1)^2}
	= \frac{1}{2C}.
\]
\end{proof}

\subsubsection{Lower Bounding the Ratio.}
\label{sec:one:end:lower}
In this section, we provide a lower bound on the ratio $\frac{\tau[\cW^{-\frac12}\cP](t)}{w(t)}$. 
Namely, we show in \lemmaref{cap:lower} that there exists an absolute constant $C'$ such that $\frac{\tau[\cW^{-\frac12}\cP](t)}{w(t)}\ge\frac{C'}{\log^3 d}$ for all $t\in\cI_{cap}$.  
We first require the following structural result from polynomial approximation theory.
\begin{theorem}[Low-degree approximation of high-degree polynomial, Theorem 3.3 in~\cite{SachdevaV14}]
\label{thm:poly:approx}
For any positive integers $s$ and $d$, there exists a degree $d$ polynomial $F$ such that
\[\sup_{t\in[-1,1]}|f(t)-t^s|\le2e^{-\frac{d^2}{s}}.\]
Moreover, for any $\delta>0$ and $d\ge\left\lceil\sqrt{2s\log\frac{2}{\delta}}\right\rceil$, there exists a polynomial $f$ of degree $d$ such that
\[\sup_{t\in[-1,1]}|f(t)-t^s|\le\delta.\]
\end{theorem}

\begin{lemma}
\label{lem:cap:lower}
For $t\in\cI_{cap}$, we have
\[
	\frac{\tau[\cW^{-\frac12}\cP](t)}{w(t)}=\Omega\left(\frac{1}{\log^3 d}\right).
\]
\end{lemma}
\begin{proof}
To lower bound the ratio $\frac{\tau[\cW^{-\frac12}\cP](t)}{w(t)}$, we first note that for \(t\in\cI_{cap}\), we have that $w(t)=C(d+1)^2$ and thus it suffices to lower bound 
\[
	\tau[\cW^{-\frac12}\cP](t)=\max_{\vx} \frac{([\cW^{-\frac12}\cP\vx](t))^2}{\normof{\cW^{-\frac12}\cP\vx}_2^2}
\]
by analyzing the quantity for a specific choice of $\vx\in\bbR^{d+1}$. 

Let $q=O\left(\frac{(d+1)^2}{\log d}\right)$ so that by \theoremref{poly:approx}, there exists a degree $d$ polynomial $f$ such that
\[
	\sup_{t\in[-1,1]}|f(t)-t^q|\le d^{-\gamma},
\]
for some constant $\gamma>0$. 
We set $\vx\in\bbR^{d+1}$ so that the operator $\cP\vx$ corresponds to $f(t)$ and lower bound $\frac{([\cW^{-\frac12}\cP\vx](t))^2}{\normof{\cW^{-\frac12}\cP\vx}_2^2}$. 

First, note that since $t^q=1$ at $t=1$, then we have $f(1)\ge 1-d^{-\gamma}$. 
Similarly, since $|t|\ge\sqrt{1-\frac{1}{\pi^2(d+1)^2C^2}}\ge1-\frac{1}{2\pi^2(d+1)^2C^2}$ for $t\in\cI_{cap}$, then we have $t^q\ge\frac{1}{4}$ since $q=O\left(\frac{(d+1)^2}{\log d}\right)$. 
Thus, we have $f(t)\ge\frac{1}{4}-d^{-\gamma}$ for all $t\in\cI_{cap}$. 
Since $w(t)=C(d+1)^2$ for all $t\in\cI_{cap}$, then
\[
	([\cW^{-\frac12}\cP\vx](t))^2\ge\frac{1}{8C(d+1)^2}.
\]

It remains to upper bound $\normof{\cW^{-\frac12}\cP\vx}_2^2$ when the operator $\cP\vx$ corresponds to $F(t)$. 
Since $\sup_{t\in[-1,1]}|f(t)-t^q|\le d^{-\gamma}$, then we have
\begin{align*}
\|\cW^{-\frac12} f\|_2^2&=\int_{-1}^1\frac{1}{w(t)}(f(t))^2\,dt\\
&\le2\int_{-1}^1\frac{1}{w(t)}d^{-2\gamma}\,dt+2\int_{-1}^1\frac{1}{w(t)}t^{2q}\,dt.
\end{align*}
Since $w(t)=\min\{C(d+1)^2,\ \tsfrac{d+1}{\pi\sqrt{1-t^2}}\}$, then $\frac{1}{w(t)}\le\frac{\pi}{d+1}$. 
Thus, 
\[
	\|\cW^{-\frac12} f\|_2^2\le\frac{4\pi d^{-2\gamma}}{d+1}+4\int_0^1\frac{1}{w(t)} t^{2q}\,dt.
\]
We decompose the interval $[0,1]$ into $\cI_1=\left[0,\sqrt{1-\frac{C^2\pi^2\log^2 d}{(d+1)^2}}\right)$ and $\cI_2=\left[\sqrt{1-\frac{C^2\pi^2\log^2 d}{(d+1)^2}},1\right]$. 
Note that for $t\in\cI_1$, we have $t\le 1-\frac{C^2\pi^2\log^2 d}{2(d+1)^2}$ and thus $t^{2q}\le\exp\left(-O\left(C^2\pi^2\log d\right)\right)$ for $q=O\left(\frac{(d+1)^2}{\log d}\right)$. 
Hence for sufficiently large $C>0$, we have that $t^{2q}\le\frac{1}{16\pi(d+1)^3}$ for all $t\in\cI_1$. 
Thus since $\frac{1}{w(t)}\le\frac{\pi}{d+1}$, then
\[
	4\int_{\cI_1}\frac{1}{w(t)} t^{2q}\,dt\le\frac{16\pi}{d+1}\int_{\cI_1}t^{2q}\,dt\le\frac{1}{(d+1)^4}.
\]
Note that $|\cI_2|\le\frac{C^2\pi^2\log^2 d}{2(d+1)^2}$ and $t^{2q}\le 1$ for $t\in\cI_2$. 
Moreover for $t\in\cI_2$, we have $\frac{d+1}{\pi\sqrt{1-t^2}}\ge\frac{C(d+1)^2}{\log d}$ so that $\frac{1}{w(t)}\le\frac{\log d}{C(d+1)^2}$. 
Hence, 
\begin{align*}
\int_{\cI_2}\frac{1}{w(t)} t^{2q}\,dt&\le\int_{\cI_2}\frac{\log d}{C(d+1)^2}\,dt\\
&\le\frac{\log d}{C(d+1)^2}\cdot\frac{C^2\pi^2\log^2 d}{2(d+1)^2}=\frac{C\pi^2\log^3 d}{2(d+1)^4}.
\end{align*}
Therefore in summary, we have
\begin{align*}
\|\cW^{-\frac12} f\|_2^2&\le\frac{4\pi d^{-2\gamma}}{d+1}+4\int_0^1\frac{1}{w(t)} t^{2q}\,dt\\
&=\frac{4\pi d^{-2\gamma}}{d+1}+4\int_{\cI_1}\frac{1}{w(t)} t^{2q}\,dt+4\int_{\cI_2}\frac{1}{w(t)} t^{2q}\,dt\\
&\le\frac{4\pi d^{-2\gamma}}{d+1}+\frac{1}{(d+1)^4}+\frac{C\pi^2\log^3 d}{2(d+1)^4}.
\end{align*}
Hence for sufficiently large $\gamma>0$, we have that
\[
	\|\cW^{-\frac12} f\|_2^2=O\left(\frac{\log^3 d}{d^4}\right).
\]
Combined with the previous bound of $([\cW^{-\frac12}\cP\vx](t))^2\ge\frac{1}{8C(d+1)^2}=\Omega\left(\frac{1}{d^2}\right)$, then
\[
	\frac{([\cW^{-\frac12}\cP\vx](t))^2}{\normof{\cW^{-\frac12}\cP\vx}_2^2}=\Omega\left(\frac{d^2}{\log^3 d}\right).
\]
Finally, since $w(t)\le C(d+1)^2$, then
\[
	\frac{\tau[\cW^{-\frac12}\cP](t)}{w(t)}=\Omega\left(\frac{1}{\log^3 d}\right).
\]
\end{proof}

\subsection{Putting It All Together}
\noindent
We finally obtain \theoremref{chebyshev:ratio} from \lemmaref{mid}, \lemmaref{cap:upper}, and \lemmaref{cap:lower}. 
\begin{reptheorem}{chebyshev:ratio}
There are fixed constants $c_1,c_2,c_3$ such that, letting $w(t)=\min\left(c_1(d+1)^2,\frac{d+1}{\pi\sqrt{1-t^2}}\right)$ be the clipped Chebyshev measure on $[-1,1]$ and letting $\cW$ be the corresponding diagonal operator with $[\cW x](t) = w(t) \cdot x(t)$, for any $t \in [-1,1]$,
\[
	\frac{c_2}{\log^3 d}\le\frac{\tau[\cW^{-\frac12}\cP](t)}{w(t)}\le c_3.
\]
\end{reptheorem}
\begin{proof}
We consider casework on $t\in[-1,1]$. 
Recall that \[
	\cI_{mid} \defeq \{t ~|~ w(t) = v(t)\} = {\textstyle \left[\sqrt{1-\frac{1}{\pi^2(d+1)^2C^2}}, \sqrt{1+\frac{1}{\pi^2(d+1)^2C^2}}\right]}
\]
and \(\cI_{cap} \defeq [-1,1] \ \backslash\  \cI_{mid}\).
We have from \lemmaref{mid} that there exists a constant $C_0\ge 1$ such that $\frac{1}{C_0}\le\frac{\tau[\cW^{-\frac12}\cP](t)}{w(t)}\le C_0$ for all $t\in\cI_{mid}$. 
We have from \lemmaref{cap:upper} and \lemmaref{cap:lower} that there exist constants $C_3,C_4$ such that
\[
\frac{C_3}{\log^3 d}\le\frac{\tau[\cW^{-\frac12}\cP](t)}{w(t)}\le C_4
\]
for all $t\in\cI_{cap}$. 
Thus by setting $C_1=\min\left(C_3,\frac{1}{C_0}\right)$ and $C_2=\max(C_0,C_4)$, we have that
\[
	\frac{C_1}{\log^3 d}\le\frac{\tau[\cW^{-\frac12}\cP](t)}{w(t)}\le C_2
\]
for all $t\in[-1,1]$.
\end{proof}

We now move onto the slightly messier analysis which works for all \(p\in[1,2]\).
The core ideas are all the same, but the mathematical arguments are slightly more nuanced.

\subsection{Middle Region Analysis for \texorpdfstring{\(p\in[\frac23,2]\)}{p in [,2]}}
In this section, we show that $\frac{\tau[\cW^{\frac{1}{2}-\frac{1}{p}}\cP](t)}{w(t)}=\Theta(1)$ for $t\in\cI_{mid}$ defined by 
\[
	\cI_{mid} \defeq \{t ~|~ w(t) = v(t)\} = {\textstyle \left[\sqrt{1-\frac{1}{\pi^2(d+1)^2C^2}}, \sqrt{1+\frac{1}{\pi^2(d+1)^2C^2}}\right]}
\]
and the clipped Chebyshev measure $w(t)$ defined by
\[
	w(t)\defeq\min\{C(d+1)^2,\ v(t)\} = \min\{C(d+1)^2,\ \tsfrac{d+1}{\pi\sqrt{1-t^2}}\}.
\]

\subsubsection{Relating \texorpdfstring{\(\tau[\cW^{\frac{1}{2}-\frac{1}{p}}\cP]\) to \(\tau[\cV^{\frac{1}{2}-\frac{1}{p}}\cP]\)}{tau(W) to tau(V)}}
\label{sec:tauw:tauv:p}
We first show that $\frac{\tau[\cW^{\frac{1}{2}-\frac{1}{p}}\cP](t)}{w(t)}\approx_{\frac{2}{\pi^2C^2}}\frac{\tau[\cV^{\frac{1}{2}-\frac{1}{p}}\cP](t)}{v(t)}$ for $t\in\cI_{mid}$. 
Observe that since $\cI_{mid}$ is defined by
\[
	\cI_{mid} \defeq \{t ~|~ w(t) = v(t)\} = {\textstyle \left[\sqrt{1-\frac{1}{\pi^2(d+1)^2C^2}}, \sqrt{1+\frac{1}{\pi^2(d+1)^2C^2}}\right]},
\]
then we have \(w(t)=v(t)\) for $t\in\cI_{mid}$. 
Thus it suffices to show that \(\normof{\cW^{\frac{1}{2}-\frac{1}{p}}\cP\vx}_2^2 \approx_{\frac{2}{\pi^2C^2}} \normof{\cV^{\frac{1}{2}-\frac{1}{p}}\cP\vx}_2^2\) since
\begin{align*}
	\tau[\cW^{\frac{1}{2}-\frac{1}{p}}\cP](t)
	&= \max_{\vx} \frac{([\cW^{\frac{1}{2}-\frac{1}{p}}\cP\vx](t))^2}{\normof{\cW^{\frac{1}{2}-\frac{1}{p}}\cP\vx}_2^2}
	= \max_{\vx} \frac{([\cV^{\frac{1}{2}-\frac{1}{p}}\cP\vx](t))^2}{\normof{\cW^{\frac{1}{2}-\frac{1}{p}}\cP\vx}_2^2}\\
	&	\approx_{\frac{2}{\pi^2C^2}} \max_{\vx} \frac{([\cV^{\frac{1}{2}-\frac{1}{p}}\cP\vx](t))^2}{\normof{\cV^{\frac{1}{2}-\frac{1}{p}}\cP\vx}_2^2}
	= \tau[\cV^{\frac{1}{2}-\frac{1}{p}}\cP](t).
\end{align*}
Therefore, we first show that \(\normof{\cW^{\frac{1}{2}-\frac{1}{p}}\cP\vx}_2^2 \approx_{\frac{2}{\pi^2C^2}} \normof{\cV^{\frac{1}{2}-\frac{1}{p}}\cP\vx}_2^2\). 
\begin{lemma}
\label{lem:oper:norm:approx:p}
For all $\vx\in\mathbb{R}^{d+1}$, we have
\[
	\normof{\cW^{\frac{1}{2}-\frac{1}{p}}\cP\vx}_2^2 \approx_{\frac{\pi^2C^2}{\pi^2C^2-1}} \normof{\cV^{\frac{1}{2}-\frac{1}{p}}\cP\vx}_2^2
\]
\end{lemma}
\begin{proof}
We first bound the difference between $\normof{\cW^{\frac{1}{2}-\frac{1}{p}}\cP\vx}_2^2$ and $\normof{\cV^{\frac{1}{2}-\frac{1}{p}}\cP\vx}_2^2$. 
\begin{align*}
	\bigg|&\normof{\cW^{\frac{1}{2}-\frac{1}{p}}\cP\vx}_2^2 - \normof{\cV^{\frac{1}{2}-\frac{1}{p}}\cP\vx}_2^2\bigg|
		= \abs{ \int_{-1}^1 ([\cW^{\frac{1}{2}-\frac{1}{p}}\cP\vx](t))^2 - ([\cV^{\frac{1}{2}-\frac{1}{p}}\cP\vx](t))^2 ~ dt } \\
		&= \abs{ \int_{\cI_{cap}} ([\cW^{\frac{1}{2}-\frac{1}{p}}\cP\vx](t))^2 - ([\cV^{\frac{1}{2}-\frac{1}{p}}\cP\vx](t))^2 ~ dt+\int_{\cI_{mid}} ([\cW^{\frac{1}{2}-\frac{1}{p}}\cP\vx](t))^2 - ([\cV^{\frac{1}{2}-\frac{1}{p}}\cP\vx](t))^2 ~ dt}.
		\end{align*}
Because $[\cW^{\frac{1}{2}-\frac{1}{p}}\cP\vx](t)=[\cV^{\frac{1}{2}-\frac{1}{p}}\cP\vx](t)$ for $t\in\cI_{mid}$, then it follows that $\int_{\cI_{mid}} ([\cW^{\frac{1}{2}-\frac{1}{p}}\cP\vx](t))^2 - ([\cV^{\frac{1}{2}-\frac{1}{p}}\cP\vx](t))^2 ~ dt=0$. 
Since $w(t)$ is the clipped Chebyshev measure, we have that $w(t)\le v(t)$ and thus \((w(t))^{\frac{1}{2}-\frac{1}{p}} \geq (v(t))^{\frac{1}{2}-\frac{1}{p}}\). 
Therefore,
\begin{align*}
		\abs{\normof{\cW^{\frac{1}{2}-\frac{1}{p}}\cP\vx}_2^2 - \normof{\cV^{\frac{1}{2}-\frac{1}{p}}\cP\vx}_2^2}
	&= \abs{ \int_{\cI_{cap}} ([\cW^{\frac{1}{2}-\frac{1}{p}}\cP\vx](t))^2 - ([\cV^{\frac{1}{2}-\frac{1}{p}}\cP\vx](t))^2 ~ dt } \\
	&\leq \abs{\int_{\cI_{cap}} ([\cW^{\frac{1}{2}-\frac{1}{p}}\cP\vx](t))^2 ~ dt}=\int_{\cI_{cap}} ([\cW^{\frac{1}{2}-\frac{1}{p}}\cP\vx](t))^2 ~ dt.
	\end{align*}
Since \(w(t) = C(d+1)^2\) on \(\cI_{cap}\), 
\begin{align*}
\abs{\normof{\cW^{\frac{1}{2}-\frac{1}{p}}\cP\vx}_2^2 - \normof{\cV^{\frac{1}{2}-\frac{1}{p}}\cP\vx}_2^2} &\le (C(d+1)^2)^{1-\frac{2}{p}} \int_{\cI_{cap}} ([\cP\vx](t))^2 dt. 
\end{align*}
By \lemmaref{unscaled-leverages}, we have that \(([\cP\vx](t))^2 \leq \frac{(d+1)^2}{2}\normof{\cP\vx}_2^2\). 
Thus, 
\begin{align*}
\abs{\normof{\cW^{\frac{1}{2}-\frac{1}{p}}\cP\vx}_2^2 - \normof{\cV^{\frac{1}{2}-\frac{1}{p}}\cP\vx}_2^2} &\leq (C(d+1)^2)^{1-\frac{2}{p}} \cdot \frac{(d+1)^2}{2} \normof{\cP\vx}_2^2 \int_{\cI_{cap}} dt\\
	&=\tsfrac{C^{1-\frac{2}{p}}(d+1)^{4-\frac{4}{p}}}{2}\normof{\cP\vx}_2^2 \int_{\cI_{cap}} dt.
\end{align*}
We upper bound the length of the interval \(\cI_{cap}\) by observing that \(1-\sqrt{1-\frac{1}{x^2}} \leq \frac1{x^2}\) for $x^2\ge 1$ and thus,
\begin{align*}
	\textstyle\int_{\cI_{cap}} dt &= 2 \cdot \left(1 - \sqrt{1-\frac{1}{\pi^2(d+1)^2C^2}}\right)\leq \frac{2}{\pi^2(d+1)^2C^2}.
\end{align*}
Therefore,
\[
\abs{\normof{\cW^{\frac{1}{2}-\frac{1}{p}}\cP\vx}_2^2 - \normof{\cV^{\frac{1}{2}-\frac{1}{p}}\cP\vx}_2^2}\leq\frac{2}{\pi^2(d+1)^2C^2}\cdot\tsfrac{C^{1-\frac{2}{p}}(d+1)^{4-\frac{4}{p}}}{2}\normof{\cP\vx}_2^2
	= \frac{(d+1)^{2-\frac{4}{p}}}{C^{\frac{2}{p}+1}\pi^2}\cdot\normof{\cP\vx}_2^2.
	\]
We then bound the norm \(\normof{\cP\vx}_2^2\) by noting that \(w(t) \leq C(d+1)^2\). 
Thus, \(1 \leq C^{\frac{1}{p}-\frac{1}{2}}(d+1)^{\frac{2}{p}-1} \cdot (w(t))^{\frac{1}{2}-\frac{1}{p}}\), so that
\begin{align*}
	\normof{\cP\vx}_2^2
	&= \int_{-1}^1 (1 \cdot [\cP\vx](t))^2 dt \\
	&\leq \int_{-1}^1 \left(C^{\frac{1}{p}-\frac{1}{2}}(d+1)^{\frac{2}{p}-1} \cdot (w(t))^{\frac{1}{2}-\frac{1}{p}}\cdot [\cP\vx](t)\right)^2 dt \\
	&= C^{\frac{2}{p}-1}(d+1)^{\frac{4}{p}-2} \normof{\cW^{\frac{1}{2}-\frac{1}{p}}\cP\vx}_2^2.
\end{align*}
Therefore,
\[
	\abs{\normof{\cW^{\frac{1}{2}-\frac{1}{p}}\cP\vx}_2^2 - \normof{\cV^{\frac{1}{2}-\frac{1}{p}}\cP\vx}_2^2}
	\leq \frac{(d+1)^{2-\frac{4}{p}}}{C^{\frac{2}{p}+1}\pi^2}\cdot C^{\frac{2}{p}-1}(d+1)^{\frac{4}{p}-2}\normof{\cW^{\frac{1}{2}-\frac{1}{p}}\cP\vx}_2^2
	= \tsfrac{1}{\pi^2C^2}\normof{\cW^{\frac{1}{2}-\frac{1}{p}}\cP\vx}_2^2.
\]
Rearranging this inequality, we have that
\[
	\abs{1 - \frac{\normof{\cV^{\frac{1}{2}-\frac{1}{p}}\cP\vx}_2^2}{\normof{\cW^{\frac{1}{2}-\frac{1}{p}}\cP\vx}_2^2}} \leq \frac1{\pi^2C^2}
\]
or equivalently,
\[
	1 - \frac{1}{\pi^2C^2} \leq \frac{\normof{\cV^{\frac{1}{2}-\frac{1}{p}}\cP\vx}_2^2}{\normof{\cW^{\frac{1}{2}-\frac{1}{p}}\cP\vx}_2^2} \leq 1 + \frac1{\pi^2C^2}
	\leq \frac{1}{1-\frac{1}{\pi^2C^2}},
\]
for $C>\frac{1}{\pi}$. 
Since \(\frac{1}{1-\frac{1}{\pi^2C^2}} = \frac{\pi^2C^2}{\pi^2C^2-1}\), then we have the multiplicative error guarantee
\[
	\normof{\cW^{\frac{1}{2}-\frac{1}{p}}\cP\vx}_2^2 \approx_{\frac{\pi^2C^2}{\pi^2C^2-1}} \normof{\cV^{\frac{1}{2}-\frac{1}{p}}\cP\vx}_2^2
\]
for \(C > \frac1\pi \approx 0.312\).
\end{proof}

\noindent
We now relate $\frac{\tau[\cW^{\frac{1}{2}-\frac{1}{p}}\cP](t)}{w(t)}$ to $\frac{\tau[\cV^{\frac{1}{2}-\frac{1}{p}}\cP](t)}{v(t)}$ for $t\in\cI_{mid}$
\begin{corollary}
\label{corol:oper:approx:mid:p}
$\frac{\tau[\cW^{\frac{1}{2}-\frac{1}{p}}\cP](t)}{w(t)}\approx_{\frac{\pi^2C^2}{\pi^2C^2-1}}\frac{\tau[\cV^{\frac{1}{2}-\frac{1}{p}}\cP](t)}{v(t)}$ for $t\in\cI_{mid}$.
\end{corollary}
\begin{proof}
By \lemmaref{oper:norm:approx:p}, \(\normof{\cW^{\frac{1}{2}-\frac{1}{p}}\cP\vx}_2^2 \approx_{\frac{\pi^2C^2}{\pi^2C^2-1}} \normof{\cV^{\frac{1}{2}-\frac{1}{p}}\cP\vx}_2^2\) for all $\vx\in\mathbb{R}^{d+1}$. 
Since \(w(t)=v(t)\) for $t\in\cI_{mid}$, it follows from the definition of the leverage functions that
\begin{align*}
\frac{\tau[\cW^{\frac{1}{2}-\frac{1}{p}}\cP](t)}{w(t)}&=\frac{1}{w(t)}\max_{\vx} \frac{([\cW^{\frac{1}{2}-\frac{1}{p}}\cP\vx](t))^2}{\normof{\cW^{\frac{1}{2}-\frac{1}{p}}\cP\vx}_2^2}\\
&=\frac{1}{v(t)}\max_{\vx} \frac{([\cV^{\frac{1}{2}-\frac{1}{p}}\cP\vx](t))^2}{\normof{\cW^{\frac{1}{2}-\frac{1}{p}}\cP\vx}_2^2}\\
&\approx_{\frac{\pi^2C^2}{\pi^2C^2-1}}\frac{1}{v(t)}\max_{\vx} \frac{([\cV^{\frac{1}{2}-\frac{1}{p}}\cP\vx](t))^2}{\normof{\cV^{\frac{1}{2}-\frac{1}{p}}\cP\vx}_2^2}\\
&=\frac{\tau[\cV^{\frac{1}{2}-\frac{1}{p}}\cP](t)}{v(t)},
\end{align*}
as desired.
\end{proof}

\subsubsection{Complete Result in the Middle}
\label{sec:complete:mid:p}
\noindent
We now show that $\tau[\cW^{\frac{1}{2}-\frac{1}{p}}\cP](t)$ and $w(t)$ are within a constant factor for $t\in\cI_{mid}$. 
\begin{lemma}
\label{lem:mid:p}
For $t\in\cI_{mid}$, we have
\[
	\frac{\tau[\cW^{\frac{1}{2}-\frac{1}{p}}\cP](t)}{w(t)}=\Theta(1).
\]
\end{lemma}
\begin{proof}
By \corolref{oper:approx:mid:p} and \corolref{cheby-lp-ratio}, we have that for $t\in\cI_{mid}$,
\[
	\tau[\cW^{\frac{1}{2}-\frac{1}{p}}\cP](t) \approx_\alpha v(t)
\]
for \(\alpha = \frac{\pi^2C^2}{\pi^2C^2-1} \cdot C_0\) for some constants $C_0$ and \(C > \frac1\pi\approx0.312\).
Furthermore, since \(v(t)=w(t)\) in the region $t\in\cI_{mid}$, this further implies \(\tau[\cW^{\frac{1}{2}-\frac{1}{p}}\cP](t) \approx_\alpha w(t)\), as desired.
\end{proof}

\subsection{Endcap Region Analysis for \texorpdfstring{\(p\in[\frac23,2]\)}{p in [2/3,2]}}
\noindent
We now bound the ratio for \(t\in\cI_{cap}\), and we will show that
\[
	\tau[\cW^{\frac{1}{2}-\frac{1}{p}}\cP](t) \approx_{\frac1{2C}} w(t)
\]
for \(t\in\cI_{cap}\). 
Thus it suffices to upper and lower bound the ratio $\frac{\tau[\cW^{\frac{1}{2}-\frac{1}{p}}\cP](t)}{w(t)}$.

\subsubsection{Upper Bounding the Ratio.}
\noindent
In this section, we provide an upper bound on the ratio $\frac{\tau[\cW^{\frac{1}{2}-\frac{1}{p}}\cP](t)}{w(t)}$. 
\begin{lemma}
\label{lem:cap:upper:p}
For $t\in\cI_{cap}$, we have
\[
	\frac{\tau[\cW^{\frac{1}{2}-\frac{1}{p}}\cP](t)}{w(t)}=O(1).
\]
\end{lemma}
\begin{proof}

Since $\tau[\cW^{\frac{1}{2}-\frac{1}{p}}\cP](t) = \max_{\vx}\frac{([\cW^{\frac{1}{2}-\frac{1}{p}}\cP\vx](t))^2}{\normof{\cW^{\frac{1}{2}-\frac{1}{p}}\cP\vx}_2^2}$ and $w(t)\le C(d+1)^2$ for all $t\in[-1,1]$, then for $p\in[\frac23,2]$, we can first lower bound \(\normof{\cW^{\frac{1}{2}-\frac{1}{p}}\cP\vx}_2^2\) by
\begin{align*}
	\normof{\cW^{\frac{1}{2}-\frac{1}{p}}\cP\vx}_2^2
	&= \int_{-1}^1 (w(t))^{1-\frac{2}{p}} ([\cP\vx](t))^2 \ dt \\
	&\geq \int_{-1}^1 (C(d+1)^2)^{1-\frac{2}{p}} ([\cP\vx](t))^2 \ dt = (C(d+1)^2)^{1-\frac{2}{p}} \normof{\cP\vx}_2^2. 
\end{align*}
On the other hand, the leverage function $\tau[\cW^{\frac{1}{2}-\frac{1}{p}}\cP]$ satisfies
\begin{align*}
	\tau[\cW^{\frac{1}{2}-\frac{1}{p}}\cP](t) &= \max_{\vx}\frac{([\cW^{\frac{1}{2}-\frac{1}{p}}\cP\vx](t))^2}{\normof{\cW^{\frac{1}{2}-\frac{1}{p}}\cP\vx}_2^2} = (w(t))^{1-\frac{2}{p}} \max_{\vx}\frac{([\cP\vx](t))^2}{\normof{\cW^{\frac{1}{2}-\frac{1}{p}}\cP\vx}_2^2} \\
	&= (C(d+1)^2)^{1-\frac{2}{p}} \max_{\vx}\frac{([\cP\vx](t))^2}{\normof{\cW^{\frac{1}{2}-\frac{1}{p}}\cP\vx}_2^2}
\end{align*}
because $w(t)=C(d+1)^2$ for $t\in\cI_{cap}$. 
Therefore, from the above inequality, we have
\begin{align*}
	\tau[\cW^{\frac{1}{2}-\frac{1}{p}}\cP](t) &\leq (C(d+1)^2)^{1-\frac{2}{p}}\max_{\vx}\frac{([\cP\vx](t))^2}{(C(d+1)^2)^{1-\frac{2}{p}}\normof{\cP\vx}_2^2} = \tau[\cP](t) \leq \frac{(d+1)^2}{2}.
\end{align*}
Hence,
\[
	\frac{\tau[\cW^{\frac{1}{2}-\frac{1}{p}}\cP](t)}{w(t)}
	\leq \frac{\frac{(d+1)^2}{2}}{C(d+1)^2}
	= \frac{1}{2C},
\]
as desired. 
\end{proof}

\subsubsection{Lower Bounding the Ratio.}
\noindent
We now lower bound the ratio $\frac{\tau[\cW^{\frac{1}{2}-\frac{1}{p}}\cP](t)}{w(t)}$. 
\begin{lemma}
\label{lem:cap:lower:p}
For $t\in\cI_{cap}$, we have
\[
	\frac{\tau[\cW^{\frac{1}{2}-\frac{1}{p}}\cP](t)}{w(t)}=\Omega\left(\frac{1}{\log^3 d}\right).
\]
\end{lemma}
\begin{proof}
To lower bound the ratio $\frac{\tau[\cW^{\frac{1}{2}-\frac{1}{p}}\cP](t)}{w(t)}$, we observe that $w(t)=C(d+1)^2$ for \(t\in\cI_{cap}\). 
Hence, it suffices to lower bound 
\[
	\tau[\cW^{\frac{1}{2}-\frac{1}{p}}\cP](t)=\max_{\vx} \frac{([\cW^{\frac{1}{2}-\frac{1}{p}}\cP\vx](t))^2}{\normof{\cW^{\frac{1}{2}-\frac{1}{p}}\cP\vx}_2^2}
\]
by choosing a specific polynomial represented by $\vx\in\bbR^{d+1}$. 

We choose $q=O\left(\frac{(d+1)^2}{\log d}\right)$ so that by \theoremref{poly:approx}, there exists a degree $d$ polynomial $f$ such that
\[
	\sup_{t\in[-1,1]}|f(t)-t^q|\le d^{-\gamma},
\]
for some fixed constant $\gamma>0$ to be set at a later point in the analysis.  
We choose $\vx\in\bbR^{d+1}$ so that the operator $\cP\vx$ corresponds to $f(t)$.
We then lower bound $\frac{([\cW^{\frac{1}{2}-\frac{1}{p}}\cP\vx](t))^2}{\normof{\cW^{\frac{1}{2}-\frac{1}{p}}\cP\vx}_2^2}$. 

Because $t^q=1$ at $t=1$, then $f(1)\ge 1-d^{-\gamma}$. 
Since $|t|\ge\sqrt{1-\frac{1}{\pi^2(d+1)^2C^2}}\ge1-\frac{1}{2\pi^2(d+1)^2C^2}$ for $t\in\cI_{cap}$, then $t^q\ge\frac{1}{4}$ for $q=O\left(\frac{(d+1)^2}{\log d}\right)$. 
Hence, $f(t)\ge\frac{1}{4}-d^{-\gamma}$ for all $t\in\cI_{cap}$. 
Since $w(t)=C(d+1)^2$ for all $t\in\cI_{cap}$, then
\[
	([\cW^{\frac{1}{2}-\frac{1}{p}}\cP\vx](t))^2\ge\frac{1}{8}(C(d+1)^2)^{1-\frac{2}{p}}=\Omega\left(d^{2-\frac{4}{p}}\right).
\]

We now upper bound $\normof{\cW^{\frac{1}{2}-\frac{1}{p}}\cP\vx}_2^2$ for the operator $\cP\vx$ that corresponds to $f(t)$. 
Since $\sup_{t\in[-1,1]}|f(t)-t^q|\le d^{-\gamma}$, then 
\begin{align*}
\|\cW^{\frac{1}{2}-\frac{1}{p}} f\|_2^2&=\int_{-1}^1(w(t))^{1-\frac{2}{p}}(F(t))^2\,dt\\
&\le2\int_{-1}^1(w(t))^{1-\frac{2}{p}}d^{-2\gamma}\,dt+2\int_{-1}^1(w(t))^{1-\frac{2}{p}}t^{2q}\,dt.
\end{align*}
Since $w(t)=\min\{C(d+1)^2,\ \tsfrac{d+1}{\pi\sqrt{1-t^2}}\}$, then $(w(t))^{1-\frac{2}{p}}=O\left(d^{1-\frac{2}{p}}\right)$ for $p\in[\frac23,2]$. 
Thus, 
\[
	\|\cW^{\frac{1}{2}-\frac{1}{p}} f\|_2^2\le O\left(d^{1-\frac{2}{p}-2\gamma}\right)+4\int_0^1 (w(t))^{1-\frac{2}{p}} t^{2q}\,dt.
\]
Consider a decomposition of the interval $[0,1]$ into intervals $\cI_1=\left[0,\sqrt{1-\frac{C^2\pi^2\log^2 d}{(d+1)^2}}\right)$ and $\cI_2=\left[\sqrt{1-\frac{C^2\pi^2\log^2 d}{(d+1)^2}},1\right]$. 
For $t\in\cI_1$, we have $t\le 1-\frac{C^2\pi^2\log^2 d}{2(d+1)^2}$ so that $t^{2q}\le\exp\left(-O\left(C^2\pi^2\log d\right)\right)$ for $q=O\left(\frac{(d+1)^2}{\log d}\right)$. 
Thus for sufficiently large $C>0$, we have that $t^{2q}=O\left(\frac{1}{d^7}\right)$ for all $t\in\cI_1$. 
Because $(w(t))^{1-\frac{2}{p}}\le 1$ for $p\in[\frac23,2]$, then
\[
	4\int_{\cI_1}(w(t))^{1-\frac{2}{p}} t^{2q}\,dt=O\left(\frac{1}{d^7}\right).
\]
On the other hand, $|\cI_2|\le\frac{C^2\pi^2\log^2 d}{2(d+1)^2}$ and $t^{2q}\le 1$ for $t\in\cI_2$. 
For $t\in\cI_2$, we also have either \(w(t) = C(d+1)^2\) or $w(t)=\frac{d+1}{\pi\sqrt{1-t^2}}\ge\frac{(d+1)^2}{C\pi\log d}$ so that either way \(w(t) \geq \frac{(d+1)^2}{C\pi\log d}\), and so $(w(t))^{1-\frac{2}{p}}=O\left(d^{2-\frac{4}{p}}\log^2 d\right)$.
Hence, 
\begin{align*}
4\int_{\cI_2}(w(t))^{1-\frac{2}{p}} t^{2q}\,dt&4\le\int_{\cI_2}O\left(d^{2-\frac{4}{p}}\log^2 d\right)\,dt\\
&\le O\left(d^{2-\frac{4}{p}}\log^2 d\right)\cdot\frac{C^2\pi^2\log^2 d}{2(d+1)^2}=O\left(d^{-\frac{4}{p}}\log^4 d\right).
\end{align*}
Thus in all,
\begin{align*}
\|\cW^{\frac{1}{2}-\frac{1}{p}} f\|_2^2 &\le O\left(d^{1-\frac{2}{p}-2\gamma}\right)+4\int_0^1 (w(t))^{1-\frac{2}{p}} t^{2q}\,dt\\
&=O\left(d^{1-\frac{2}{p}-2\gamma}\right)+4\int_{\cI_1} (w(t))^{1-\frac{2}{p}} t^{2q}\,dt+4\int_{\cI_2} (w(t))^{1-\frac{2}{p}} t^{2q}\,dt\\
&=O\left(d^{1-\frac{2}{p}-2\gamma}\right)+O\left(\frac{1}{d^7}\right)+O\left(d^{-\frac{4}{p}}\log^4 d\right).
\end{align*}
Hence for $\gamma=5$, for all \(p \in [\frac23,1]\), we have that
\[
	\|\cW^{\frac{1}{2}-\frac{1}{p}} F\|_2^2=O\left(d^{-\frac{4}{p}}\log^4 d\right).
\]
Combined with our previous bound that $([\cW^{\frac{1}{2}-\frac{1}{p}}\cP\vx](t))^2\ge\Omega\left(d^{2-\frac{4}{p}}\right)$ for $p\in[\frac23,2]$ and therefore,
\[
	\frac{([\cW^{\frac{1}{2}-\frac{1}{p}}\cP\vx](t))^2}{\normof{\cW^{\frac{1}{2}-\frac{1}{p}}\cP\vx}_2^2}=\Omega\left(\frac{d^2}{\log^4 d}\right).
\]
Finally, because $w(t)\le C(d+1)^2$, then
\[
	\frac{\tau[\cW^{\frac{1}{2}-\frac{1}{p}}\cP](t)}{w(t)}=\Omega\left(\frac{1}{\log^4 d}\right).
\]
\end{proof}

\subsection{Putting It All Together}
\noindent
We finally obtain \theoremref{chebyshev:ratio} from \lemmaref{mid:p}, \lemmaref{cap:upper:p}, and \lemmaref{cap:lower:p}. 
\begin{reptheorem}{chebyshev:ratio}
There are fixed constants $c_1,c_2,c_3$ such that, letting $w(t)=\min\left(c_1(d+1)^2,\frac{d+1}{\pi\sqrt{1-t^2}}\right)$ be the clipped Chebyshev measure on $[-1,1]$ and letting $\cW$ be the corresponding diagonal operator with $[\cW x](t) = w(t) \cdot x(t)$, for any \(p\in[\frac23,2]\) and $t \in [-1,1]$,
\[
	\frac{c_2}{\log^3 d}\le\frac{\tau[\cW^{-\frac12}\cP](t)}{w(t)}\le c_3.
\]
\end{reptheorem}
\begin{proof}
We consider casework on $t\in[-1,1]$.
Recall that \[
	\cI_{mid} \defeq \{t ~|~ w(t) = v(t)\} = {\textstyle \left[\sqrt{1-\frac{1}{\pi^2(d+1)^2C^2}}, \sqrt{1+\frac{1}{\pi^2(d+1)^2C^2}}\right]}
\]
and \(\cI_{cap} \defeq [-1,1] \ \backslash\  \cI_{mid}\).
By \lemmaref{mid:p}, there exists a constant $C_0\ge 1$ such that $\frac{1}{C_0}\le\frac{\tau[\cW^{\frac{1}{2}-\frac{1}{p}}\cP](t)}{w(t)}\le C_0$ for all $t\in\cI_{mid}$. 
By \lemmaref{cap:upper:p} and \lemmaref{cap:lower:p}, there exists a constant $C_3$ such that
\[
\frac{C_3}{\log^4 d}\le\frac{\tau[\cW^{\frac{1}{2}-\frac{1}{p}}\cP](t)}{w(t)}\le C_3
\]
for all $t\in\cI_{cap}$. 
Thus by setting $C_1=\min\left(C_3,\frac{1}{C_0}\right)$ and $C_2=\max(C_0,C_3)$, we have that
\[
	\frac{C_1}{\log^4 d}\le\frac{\tau[\cW^{\frac{1}{2}-\frac{1}{p}}\cP](t)}{w(t)}\le C_2
\]
for all $t\in[-1,1]$.
\end{proof}

\section*{Acknowledgments.}
Cameron Musco was supported by NSF grants 2046235 and 1763618, and an Adobe Research grant. 
Christopher Musco and Raphael Meyer were supported by NSF grant 2045590 and DOE Award DE-SC0022266. 
David P. Woodruff and Samson Zhou were supported by a Simons Investigator Award and by the National Science Foundation under Grant No. CCF-1815840.
We thank Apoorv Singh and Axel Elaldi for help designing our figures.

\bibliographystyle{alpha}
\bibliography{references}

\appendix

\section{Operator Sensitivity Sampling}
\label{app:sensitivity-sampling}
\noindent
In this section, we show \theoremref{sensitivity-correctness-Lp}, which shows that uniform sampling can achieve a constant factor approximation to the $L_p$ polynomial regression problem. 
\begin{reptheorem}{sensitivity-correctness-Lp}
Let $p\ge 1$ and suppose \(s_1,\ldots,s_{n_0}\) are drawn uniformly from \([-1,1]\). 
Let \(\mA\in\bbR^{n_0 \times (d+1)}\) be the associated Vandermonde matrix, so that \(\mA_{i,j} = s_i^{j-1}\).
Let \(\vb\in\bbR^{n_0}\) be the evaluations of \(f\), so that \(\vb_i = f(s_i)\).
For $n_0 = O\left(d^5 2^p p^2 \log d\right)$, there exists a universal constant \(c\) such that the sketched solution \(\hat\vx = \argmin_{\vx} \normof{\mA\vx-\vb}_p\) satisfies
\[
    \normof{\cP\hat\vx-f}_p \leq c \min_{\vx\in\bbR^{d+1}}\normof{\cP\vx-f}_p
\]
with probability at least $\frac{11}{12}$. 
Further, let \(\eps\in(0,1)\) and suppose \(\normof f_p \leq C \, \min_{\vx}\normof{\cP\vx-f}_p\).
If $n_0 = O\left(\frac{1}{\eps^{O(p^2)}}\,d^5p^{O(p)}\log\frac{d}{\eps}\right)$, then
\[
    \normof{\cP\hat\vx - f}_p^p \leq (1+\eps) \min_{\vx} \normof{\cP\vx-f}_p^p
\]
with probability at least \(\frac {11}{12}\).
\end{reptheorem}

Throughout this paper, we use two formulations of Bernstein's inequality in the analysis for general \(p\).

\begin{importedtheorem}[Bernstein's Inequality, Theorems 3.6 and 3.7 from \cite{chung2006concentration}]
\label{impthm:bernstein-ineq}
Let \(X_1,\ldots,X_n\) be independent zero-mean random variables with \(\abs{X_i} \leq M\) for all \(i\).
Then,
\[
    \Pr\left[
            \abs{\sum_{i=1}^n X_i} \geq \gamma
    \right] \leq 2\exp\left(
        - \frac{\frac12 \gamma^2}{\sum_{i=1}^n \E[X_i^2] + \frac13 Mt}
    \right)
\]
\end{importedtheorem}

\begin{importedtheorem}[Bounded Differences Concentration, Theorem 17 from \cite{chung2006concentration}]
\label{impthm:bded-diffs:conc}
Let $X_1,\ldots,X_n$ be independent random variables such that $|X_i-\mathbb{E}[X_i]|<c_i$ for all $i\in[n]$. 
Let $X=\sum_i X_i$ and $\gamma>0$.
Then
\[
    \Pr[\abs{X-\E[X]} \geq \gamma] \leq \exp\left(\frac{-\gamma^2}{2\sum_i c_i^2}\right)
\]
\end{importedtheorem}

We first show the constant-factor regression guarantee using \(O(d^5 p^2 2^p \log(d)\) samples.
\begin{lemma}
\label{lem:oper:const-factor-err}
Let \mA be the Vandermonde matrix formed by sampling \(n_0=O(d^5 p^2 2^p \log(d))\) points from \([-1,1]\) uniformly at random, and let \(\vb\) be the corresponding observations of \(f\).
Then, with probability at least \(\frac{11}{12}\), the sketched solution \(\hat\vx \defeq \min_{\vx} \normof{\mA\vx-\vb}_p\) has
\[
    \normof{\cP\hat\vx-f}_p \leq C \min_{\vx\in\bbR^{d+1}}\normof{\cP\vx-f}_p
\]
for some universal constant \(C > 1\).
\end{lemma}
\begin{proof}
This proof is completed in two standard arguments.
First, we show that uniformly sampling enough points yields a \(\ell_p\) subspace embedding, via an \(\eps-\)Net argument.
Second, we use a standard argument that triangle inequality and subspace embedding suffice for constant factor regression \cite{eldar2020sample,meyer2020statistical,ParulekarPP21}.

Let \(s_1,...,s_{n_0}\) denotes the uniformly sampled times.
First, fix any vector \(\vx\in\bbR^{d+1}\).
Then let \(Y_i \defeq \frac2{n_0} \abs{\cP\vx(s_i)}^p\), so that \(\E[Y_i] = \frac1{n_0} \normof{\cP\vx}_p^p\).
Note that \(\abs{\cP\vx(s_i)}^p \leq d^2(p+1)\) by \lemmaref{uniform-sensitivity-bound}, so that \(Y_i \leq \frac{2d^2(p+1)}{n_0} \normof{\cP\vx}_p^p\) and therefore \(\abs{Y_i - \E[Y_i]} \leq \frac{3d^2(p+1)}{n_0} \normof{\cP\vx}_p^p\).
Then by letting \(Y = \sum_{i=1}^{n_0} Y_i = \normof{r\mA\vx}_p^p\), where \(r = (\frac1{n_0})^{1/p}\) is a rescaling factor, and applying the Bounded Differences Inequality (\importedtheoremref{bded-diffs:conc}) for \(\gamma=2^{-p}\normof{\cP\vx}_p^p\), yields
\begin{align*}
    \PPr{\abs{\normof{r\mA\vx}_p^p - \normof{\cP\vx}_p^p} \geq 2^{-p} \normof{\cP\vx}_p^p}
    &\leq 2\exp\left(\frac{-2^{-p} \normof{\cP\vx}_p^{2p}}{2n_0 \frac{9d^4(p+1)^2}{n_0^2} \normof{\cP\vx}_p^{2p}}\right) \\
    &= 2\exp\left(\frac{-n_0}{9\cdot 2^pd^4(p+1)^2}\right) \\
    &\leq \frac{1}{\exp(O(d \log d))}
\end{align*}
Where the last line uses the fact that \(n_0 = O(d^5 p^2 2^p \log(d))\).
Note that \(\abs{a-b} \leq \abs{a^p - b^p}^{1/p}\) for all \(a,b>0\) and \(p\geq 1\).
So, we get \(\abs{\normof{r\mA\vx}_p - \normof{\cP\vx}_p} \leq \abs{\normof{r\mA\vx}_p^p - \normof{\cP\vx}_p^p}^{1/p} \leq \frac12 \normof{\cP\vx}_p\) with high probability.
That is,
\begin{align}
\label{eq:oper:const-factor-per-vector}
    \PPr{\abs{\normof{r\mA\vx}_p - \normof{\cP\vx}_p} \geq \frac12 \normof{\cP\vx}_p} \leq \frac{1}{\exp(O(d \log d))}
\end{align}

We now union bound this guarantee over a net.
We first define the ball $ \cB=\left\{\vx ~\big|~ \normof{\cP\vx}_p=1 \right\}$.
The let \cN denote a net over \cB such that, for any \(\vx\in\cB\), there exists some \(\vy\in\cN\) such that \(\normof{\cP\vx-\cP\vy}_p \leq 0.1\).
By Lemma 2.4 of \cite{bourgain1989approximation}, \cN has at most \(10^{O(d)}\) elements.

Next, note that any \(\vx\in\cB\) can be written as \(\vx=\sum_{i=0}^\infty \alpha_i \vy_i\) where \(\alpha_0=1\) and \(\abs{\alpha_i}\leq0.1^{i}\) and \(\vy_i\in\cN\).
So, we can union bound \equationref{oper:const-factor-per-vector} over all \(\vy\in\cN\) to upper bound
\begin{align*}
    \normof{r\mA\vx}_p
    &\leq \sum_{i=0}^\infty \alpha_i \normof{r\mA\vy_i}_p \\
    &\leq 1.5 \sum_{i=0}^\infty \alpha_i \normof{\cP\vy_i}_p \\
    &\leq 1.5 \sum_{i=0}^\infty 0.1^i \\
    &= \frac{1.5}{1-0.1} \\
    &\leq 1.825
\end{align*}
And lower bound
\begin{align*}
    \normof{r\mA\vx}_p
    &\geq \alpha_0 \normof{r\mA\vy_0}_p - \sum_{i=1}^\infty \alpha_i \normof{r\mA\vy_i}_p \\
    &\geq (1-0.5) \alpha_0 \normof{\cP\vy_0}_p - (1+0.5)\sum_{i=1}^\infty \alpha_i \normof{\cP\vy_i}_p \\
    &\geq 0.5 - 1.5\sum_{i=1}^\infty 0.1^i \\
    &= 0.3
\end{align*}
That is, \(\normof{r\mA\vx}_p = 1 \pm 0.9\) for any \vx such that \(\normof{\cP\vx}_p = 1\).
So, just by scaling this guarantee, we have shown that for all \(\vx\in\bbR^{d+1}\) we have
\[
    \bigg|\normof{r\mA\vx}_p - \normof{\cP\vx}_p\bigg| \leq 0.9 \normof{\cP\vx}_p
\]
This is the complete subspace guarantee.
We now bound the error of the sketched solution \(\hat\vx\).

Let \(\vx^* \defeq \argmin_{\vx} \normof{\cP\vx-f}_p\) attain the best optimal loss.
Then, by repeated use of the triangle inequality and our subspace embedding,
\begin{align*}
    \normof{\cP\hat\vx - f}_p
    &\leq \normof{\cP\hat\vx - \cP\vx^*}_p + \normof{\cP\vx^* - f}_p \\
    &\leq 2r\normof{\mA\hat\vx - \mA\vx^*}_p + \normof{\cP\vx^* - f}_p \\
    &\leq 2r(\normof{\mA\hat\vx - \vb}_p + \normof{\mA\vx^* - \vb}_p) + \normof{\cP\vx^* - f}_p \\
    &\leq 4r\normof{\mA\vx^*-\vb}_p + \normof{\cP\vx^* - f}_p \tag{Optimality of \(\tilde\vx\)} \\
\end{align*}
Then, noting that \(\E[\normof{r(\mA\vx^*-\vb)}_p^p] = \normof{\cP\vx^*-f}_p^p\) so that by Markov's inequality we have \(\normof{r(\mA\vx^*-\vb)}_p \leq 10 \normof{\cP\vx^*-f}_p^p\) with probability 0.9, we conclude that
\[
    \normof{\cP\hat\vx - f}_p \leq 41\normof{\cP\vx^* - f}_p
\]
which completes the proof.
\end{proof}

We next show that any near-optimal solution to the $L_p$ matrix regression problem formed from subsampling a large number of points in $[-1,1]$ also corresponds to a near-optimal solution to the $L_p$ polynomial regression problem. 
\begin{lemma}
\label{lem:oper:const}
Let \mA be the Vandermonde matrix formed by sampling $n_0=O(d^5 p^2 2^p \log(d))$ points on $[-1,1]$, and let $\vb$ be the corresponding observations of $f$. 
Let $OPT=\min_{\vx\in\mathbb{R}^{d+1}}\normof{\cP\vx-f}_p$.
Then with probability at least \(0.9\), all $\hat\vx\in\mathbb{R}^{d+1}$ with $\normof{\mA\hat\vx-\vb}_p\le 11OPT$ have $\normof{\cP\hat\vx-f}_p \le 24OPT$.
\end{lemma}
\begin{proof}
Let $\vx^*$ be a minimizer of $\normof{\cP\vx-f}_p$ so that $OPT=\normof{\cP\vx^*-f}_p$.
We now suppose by contradiction that $\normof{\cP\hat\vx-f}_p>24OPT$. 
By triangle inequality,
\[
    \normof{\mA\hat\vx-\vb}_p \geq \normof{\mA(\hat\vx-\vx^*)}_p - \normof{\mA\vx^*-\vb}_p.
\]
Since \mA is formed by uniform sampling with $n_0=\poly(dp/\eps)$ points from $[-1,1]$, then with high probability,
\[
    \frac{23}{24} \normof{\mA\vx}_p \leq \normof{\cP\vx}_p \leq \frac{25}{24} \normof{\mA\vx}_p
\]
for all $\vx\in\mathbb{R}^{d+1}$. Formally, we prove such a bound in the proof of \lemmaref{oper:const-factor-err}.
Moreover, note that since $\vx^*=\argmin_{\vx\in\mathbb{R}^{d+1}}\normof{\cP\vx-f}_p$ has $OPT=\normof{\cP\vx^*-f}_p$, then we have $\mathbb{E}[\normof{\mA\vx^*-\vb}^p_p]=OPT^p$. 
Thus by Jensen's inequality for $p\ge 1$, we have $\mathbb{E}[\normof{\mA\vx^*-\vb}_p]\le OPT$ and by Markov's inequality, 
\[
    \PPr{\|\mA\bx^*-\vb\|_p\ge 11OPT}\le\frac{1}{11}.
\]
Thus with probability at least $0.9$,
\[
    \normof{\mA\hat\vx-\vb}_p \geq \frac{23}{24} \normof{\cP(\hat\vx-\vx^*)}_p - 11OPT.
\]
By triangle inequality,
\[
    \normof{\mA\hat\vx-\vb}_p \geq \frac{23}{24} \normof{\cP\hat\vx-f}_p - \normof{\cP\vx^*-f}_p - 11OPT.
\]
Thus if $\normof{\cP\hat\vx-f}_p>24OPT$, then
\[\normof{\mA\hat\vx-\vb}_p>23OPT-OPT-11OPT=11OPT,\]
which contradicts the given fact that $\normof{\mA\hat\vx-\vb}_p \le 11OPT$.
Hence we must have 
\[\normof{\cP\hat\vx-f}_p \le 24OPT.\]
\end{proof}

\begin{lemma}
\label{lem:uni:oper:rel:lp}
Let $OPT=\min_{\vx\in\mathbb{R}^{d+1}}\normof{\cP\vx-f}_p$ and suppose that $\normof{f}_p \le C\cdot OPT$ for some fixed constant $C\ge 1$. 
Let $\mA$ be the Vandermonde matrix formed by sampling $n_0 = O\left(\frac{1}{\eps^{O(p^2)}}\,d^5p^{O(p^2)}\log\frac{d}{\eps}\right)$ random points uniformly from $[-1,1]$.
Let $\vb$ be the corresponding evaluations of $f$.
Then with probability at least $0.9$, the minimizer $\hat\vx$ to $\min_{\vx\in\bbR^{d+1}}\normof{\mA\vx-\vb}_p$ satisfies
\[
    \normof{\cP\hat\vx-f}_p \leq (1+\eps) OPT.
\]
\end{lemma}
\begin{proof}
We first note that by \lemmaref{uniform-sensitivity-bound}, all sensitivities of $\cP$ are at most \(M \defeq d^2(p+1)\).

Note that if $\vx^*=\argmin_{\vx\in\mathbb{R}^{d+1}}\normof{\cP\vx-f}_p$, so that $OPT=\normof{\cP\vx^*-f}_p$, then we have $\mathbb{E}[\normof{\mA\vx^*-f}^p_p]=OPT^p$.
By Jensen's inequality for $p\ge 1$, we have $\mathbb{E}[\normof{\mA\vx^*-f}_p]\le OPT$. 
Thus by Markov's inequality,
\[
    \PPr{\normof{\mA\vx^*-\vb}_p \geq 11OPT} \le \frac{1}{11}.
\]
We condition against this event.
Then, we have \(\normof{\mA\hat\vx-\vb}_p \leq \normof{\mA\vx^*-\vb}_p \leq 11OPT\), so by \lemmaref{oper:const} we also have \(\normof{\cP\hat\vx-f}_p\leq24OPT\).

For the rest of this proof, let \(\vz\in\bbR^{d+1}\) be any vector such that \(\normof{\cP\vz-f}_p\leq24OPT\).
By triangle inequality, we have
\[
    \normof{\cP\vz}_p
    \leq \normof{\cP\vz-f}_p + \normof{f}_p
    \leq 25C \cdot OPT
\]
Since the sensitivities of $\cP$ are all at most $M$, then by definition of sensitivities, we have that for all $u\in[-1,1]$, 
\[
    \frac{|[\cP\vz](u)|^p}{\normof{\cP\bz}_p^p}\le M
\]
In particular, for all $u\in[-1,1]$,
\[
    |[\cP\vz](u)|^p\le \tau \defeq M 25^p C^p OPT^p.
\]
We partition the points of interval $[-1,1]$ into two groups. 
We define $G=\{t\,:\,|f(t)|^p\le\tau p^{2p}/\eps^{p^2}\}$ and $B=\{t\,:\,|f(t)|^p>\tau p^{2p}/\eps^{p^2}\}$.
Intuitively, \(B\) is a set of ``bad times'' where \(f\) is so large that polynomials cannot fit it, and \(G\) is the remaining set of ``good times''.
So for any $\vz$ with $\normof{\cP\vz-f}_p \le 24OPT$, we have $|[\cP\vz](t)| \le \tau^{1/p}$ as before, and also for any $u\in B$ we have $|f(u)|> \frac{p^2}{\eps^p} \tau^{1/p}$. 
Thus, for any $u\in B$, we have
\[
    \left(1-\frac{\eps^p}{p^2}\right)|f(u)|\le|[\cP\vz](u)-f(u)|\le\left(1+\frac{\eps^p}{p^2}\right)|f(u)|
\]
Therefore,
\begin{align}
\label{eqn:se:bad}
    (1-\eps)|f(u)|^p\le|[\cP\vz](u)-f(u)|^p\le(1+\eps)|f(u)|^p.
\end{align}
This formalizes the idea that \(f\) cannot be fit by a polynomial on \(B\).
On the other hand, for \(G\), we have 
\begin{align}
\label{eqn:se:add}
    \normof{\mA\vz-\vb}_p^p = \normof{\mA_G\vz-\vb_G}_p^p + \normof{\mA_B\vz-\vb_B}_p^p,
\end{align}
where $\mA_G$ and \(\vb_G\) are the rows of \mA and \vb associated with points sampled in $G$, and where \(\mA_B\) and \(\vb_B\) are similarly the rows associated with points sampled in $B$.
We will next show via an \(\eps\)-Net argument that the residual \(\normof{\mA_G\vz-\vb_G}_p\) is preserved for all valid \vz vectors.

\paragraph{Accuracy of $\mA$ on a single coefficient vector $\vz$ at points in $G$.}
For each sample $s_i$ with $i\in[n_0]$, if \(s_i\in G\), we define $X_i=\frac{2}{n_0}|\cP\vz(s_i)-f(s_i)|^p$ to be the corresponding contribution to the empirical residue by the sample.
Otherwise, if $s_i\notin G$, we define $X_i=0$. 
Since we sample uniformly, i.e., the probability density function for $s_i$ satisfies $p(t)=\frac{1}{2}$ for all $t\in[-1,1]$, then 
\[
    \Ex{X_i}=\frac{1}{n_0} \int_{t\in G}|\cP\vz(t)-f(t)|^p\,dt = \frac1{n_0}\normof{\cP\vz-f}_G^p,
\]
where \(\normof{f}_G^p \defeq \int_{t\in G} \abs{f}^p dt\) is the integral only over the set \(G\).
Because $|\cP\vz(u)|^p\le \tau$ and $|f(u)|^p\le\tau p^{2p}/\eps^{p^2}$ for all $u\in G$, we have
\[
    |\cP\vz(u)-f(u)|\le \frac{2p^2}{\eps^p}\tau^{1/p}
\]
so that
\[
    \abs{X_i - \Ex{X_i}}
    \leq \frac{2}{n_0}\abs{\cP\vz(s_i)-f(s_i)}^p + \frac1{n_0}\normof{\cP\vz-f}_G^p
    \leq \frac{4}{n_0}\left(\frac{2p^2}{\eps^p}\right)^p \tau
\]
Then let $X=\sum_{i\in[n_0]}X_i$ so that, letting \(r \defeq (\frac2{n_0})^{1/p}\) be a rescaling factor,
\[
    \E[\normof{r(\mA_G\vz-\vb_G)}_p^p] =\Ex{X} = \normof{\cP\vz-f}_G^p
\]
Setting $\gamma=\frac{\eps^p}{2^{O(p^2)}} \normof f_p^p$ in the formulation of Bernstein's concentration inequality in \importedtheoremref{bded-diffs:conc}, we have 
\[
    \PPr{\abs{X - \E[X]} \geq \gamma}
    \le \exp\left(\frac{-\frac{\eps^p}{2^{O(p^2)}}\normof{f}_p^{2p}}{2\sum_{i\in[n_0]} (\frac{4}{n_0}(\frac{2p^2}{\eps^p})^p \tau)^2 }\right)
    \le \exp\left(\frac{-\eps^{2+2p^2}}{32(50Cp^2)^{2p}M^22^{O(p^2)}} \cdot n_0\right).
\]
Thus for $n_0 = O\left(\frac{1}{\eps^{2+2p^2}}\,d\cdot M^2p^{O(p^2)}\log\frac{d}{\eps}\right)$, we have
\[
    \PPr{\bigg| \normof{r(\mA_G\vz-\vb_G)}_p^p - \normof{\cP\vz-f}_G^p \bigg| \ge \frac{\eps^p}{2^{O(p^2)}} \normof{f}_p^p} \le \frac{1}{\exp(O(dp \log d/\eps))},
\]
which implies by concavity (and therefore subadditivity) of \(t\mapsto t^{1/p}\) for $p\ge 1$,
\begin{align}
\label{eq:oper:se:residual}
    \PPr{\bigg| \normof{r(\mA_G\vz-\vb_G)}_p - \normof{\cP\vz-f}_G \bigg| \ge \frac{\eps}{2^{O(p)}} \normof{f}_p} \le \frac{1}{\exp(O(dp \log d/\eps))}.
\end{align}
By a similar argument, let \(Y_i \defeq \frac2{n_0} \abs{\cP\vz(s_i)}^p\) for all \(s_i\), so that \(\E[Y_i] = \frac1{n_0} \normof{\cP\vz}_p^p\) and \(\abs{Y_i-\E[Y_i]} \leq \frac{3M}{n_0} \normof{\cP\vz}_p^p\).
Then \(Y \defeq \sum_{i\in[n_0]} Y_i\), by \importedtheoremref{bded-diffs:conc} for \(\gamma=\eps^{p+1}\normof{\cP\vz}_p^p\), yields
\begin{align}
\label{eq:oper:se}
    \PPr{\bigg|\normof{r\mA\vz}_p-\normof{\cP\vz}_p\bigg| \le \eps\normof{\cP\vz}_p } \le\frac{1}{\exp(O(d\log d/\eps))}.
\end{align}
Since $M=(p+1)(d+1)^2$, the total number of samples is $n_0 = O\left(\frac{1}{\eps^{2+2p^2}}\,d^5p^{O(p^2)}\log\frac{d}{\eps}\right)$.

The arguments so far, when combined carefully (see the last part of this proof), imply that the error from uniform sampling does not matter on \(B\), and that for any fixed \vz such that \(\normof{\cP\vz-f}_p\leq24OPT\), the error on \(G\) is preserved.
So, for any such \vz, we can say with high probability that
\begin{align*}
    (1-\epsilon) \normof{\cP\vz-f}_p^p \leq \normof{\mA\vz-\vb}_p^p \leq (1+\epsilon) \normof{\cP\vz-f}_p^p
\end{align*}
However, the epsilon-net argument needs to be applied to just \(G\) on its own, so we now construct a net under the \(\normof\cdot_G\) norm.

\paragraph{$\eps$-net argument for subspace embedding.}
We now union bound over a net by first defining the ball $ \cB=\left\{\vx ~\big|~ \normof{\cP\vx}_p=1 \right\}$.
The let \cN denote a net over \cB such that, for any \(\vx\in\cB\), there exists some \(\vy\in\cN\) such that \(\normof{\cP\vx-\cP\vy}_p \leq \eps\).
By Lemma 2.4 of \cite{bourgain1989approximation}, \cN has at most \((\frac 1\eps)^{O(d)}\) elements.

Next, note that any \(\vx\in\cB\) can be written as \(\vx=\sum_{i=0}^\infty \alpha_i \vy_i\) where \(\alpha_0=1\) and \(\abs{\alpha_i}\leq\eps^{i}\) and \(\vy_i\in\cN\).
So, we can union bound \equationref{oper:se} over all \(\vy\in\cN\) to upper bound
\begin{align*}
    \normof{r\mA\vx}_p
    &\leq \sum_{i=0}^\infty \alpha_i \normof{r\mA\vy_i}_p \\
    &\leq (1+\eps) \sum_{i=0}^\infty \alpha_i \normof{\cP\vy_i}_p \\
    &\leq (1+\eps) \sum_{i=0}^\infty \eps^i \\
    &= \frac{1+\eps}{1-\eps} \\
    &\leq 1+4\eps
\end{align*}
And lower bound
\begin{align*}
    \normof{r\mA\vx}_p
    &\geq \alpha_0 \normof{r\mA\vy_0}_p - \sum_{i=1}^\infty \alpha_i \normof{r\mA\vy_i}_p \\
    &\geq (1-\eps) \alpha_0 \normof{\cP\vy_0}_p - (1+\eps)\sum_{i=1}^\infty \alpha_i \normof{\cP\vy_i}_p \\
    &\geq (1-\eps) - (1+\eps)\sum_{i=1}^\infty \eps^i \\
    &\geq 1-6\eps
\end{align*}
That is, \(\normof{r\mA\vx}_p = 1 \pm 6\eps\) for any \vx such that \(\normof{\cP\vx}_p = 1\).
So, just by scaling this guarantee, we have shown that for all \(\vx\in\bbR^{d+1}\) we have
\[
    \bigg|\normof{r\mA\vx}_p - \normof{\cP\vx}_p\bigg| \leq 6\eps \normof{\cP\vx}_p
\]

\paragraph{$\eps$-net argument over all coefficient vectors.}
Now again consider any \(\vz\) such that \(\normof{\cP\vz-f}_p \leq 24OPT\).
Then \(\normof{\cP\vz}_p \leq 25\normof f_p\):
\[
    \normof{\cP\vz}_p \leq \normof{\cP\vz-f}_p + \normof{f}_p \leq 24\min_{\vx}\normof{\cP\vx-f}_p + \normof{f}_p \leq 25\normof{f}_p
\]
Then let \(\vy=\alpha\vy_0\) be the corresponding net vector as in the previous paragraph.
Then we have \(\normof{\cP\vy-f}_p\leq26OPT\) for \(\eps \leq O(1)\):
\[
    \normof{\cP\vy-f}_p
    \leq \normof{\cP\vz-f}_p + \normof{\cP\vz-\cP\vy}_p
    \leq 24OPT + 6\eps\normof{\cP\vx'}_p
    \leq 24OPT + 6\eps 21 \normof{f}_p
    \leq (20+21\cdot6\eps C) OPT
\]

Let \(\cN'\) be a net over \(\cB' \defeq \{\vz ~|~ \normof{\cP\vz}_p \leq 26OPT\}\) such that any \(\vz\in\cB'\) has some \(\vy\in\cB'\) such that \(\normof{\cP\vz-\cP\vy}_p \leq \tsfrac{\eps}{2^{O(p)}} \cdot 26OPT\).
Since \(\frac1{26OPT}\cN'\) is a \(\eps\)-Net for the unit ball in the range of \cP, Lemma 2.4 from \cite{bourgain1989approximation} tells us that this net has size \((\frac{2^{O(p)}}\eps)^{O(d)}\).
We union bound \equationref{oper:se:residual} over all \(\vy\in\cN'\).

Then, we can write \(\vz=\sum_{i=0}^\infty \alpha_i \vy_i\) with \(\alpha_0=1\), \(\alpha_i \leq \eps\), and \(\vy_i\in\cN'\).
We will then write \(\vz = \vy_0 + \vDelta\) where \(\vDelta \defeq \sum_{i=1}^\infty \alpha_i\vy_i\) and apply the triangle inequality:

\begin{align*}
    \bigg| \normof{r(\mA_G\vz-\vb_G)}_p - \normof{\cP\vz-f}_G \bigg|
    &\leq \bigg| \normof{r(\mA_G\vy_0-\vb_G)}_p - \normof{\cP\vy_0-f}_G \bigg| + \normof{r\mA_G\vDelta_0}_p + \normof{\cP\vDelta_0}_G \\
    &\leq \bigg| \normof{r(\mA_G\vy_0-\vb_G)}_p - \normof{\cP\vy_0-f}_G \bigg| + \normof{r\mA\vDelta_0}_p + \normof{\cP\vDelta_0}_p \\
    &\leq \bigg| \normof{r(\mA_G\vy_0-\vb_G)}_p - \normof{\cP\vy_0-f}_G \bigg| + (1+6\eps)\normof{\cP\vDelta_0}_p \\
    &\leq \tsfrac{\eps}{2^{O(p)}}\normof{f}_p + (1+6\eps)\sum_{i=1}^\infty(\tsfrac{\eps}{2^{O(p)}})^i \cdot 26OPT \\
    &\leq O(\tsfrac{\eps}{2^{O(p)}})\normof{f}_p
\end{align*}
In other words, we have that with high probability for all $\vz$ with $\normof{\cP\vz-f}_p\le 24OPT$,
\begin{align}
    \label{eq:affine-subspace-embed-norm}
    \normof{\cP\vz-f}_G-O(\tsfrac{\eps}{2^{O(p)}})\normof{f}_p
    \le \normof{r(\mA_G\bz-\vb_G)}_p
    \le \normof{\cP\vz-f}_G +O(\tsfrac{\eps}{2^{O(p)}})\normof{f}_p
\end{align}
Now we extend \equationref{affine-subspace-embed-norm} to holds for norms with the exponent of \(p\) on them.
We do this by case analysis, with either \(\normof{\cP\vz-f}_G \leq \frac12\normof{f}_p\) or \(\normof{\cP\vz-f}_G \geq \frac12\normof{f}_p\).
If \(\normof{\cP\vz-f}_G \leq \frac12\normof{f}_p\), then we use the bound \((u+\eps)^p \leq u^p + 2\eps p\) for \(u + \eps \leq 1\), as proven later in \lemmaref{uniform-sampling-misc-bounds}:
\begin{align*}
    \normof{\mA_G\vz-\vb_G}_p^p
    &\leq (\normof{\cP\vz-f}_G + O(\tsfrac{\eps}{2^{O(p)}})\normof{f}_p)^p \\
    &= \normof{f}_p^p (\tsfrac{\normof{\cP\vz-f}_G}{\normof{f}_p} + O(\tsfrac{\eps}{2^{O(p)}}))^p \\
    &\leq \normof{f}_p^p (\tsfrac{\normof{\cP\vz-f}_G^p}{\normof{f}_p^p} + O(\tsfrac{\eps}{2^{O(p)}} p)) \tag{\(\tsfrac{\normof{\cP\vz-f}_G}{\normof{f}_p} + O(\tsfrac{\eps}{2^{O(p)}}) \leq 1\)} \\
    &= \normof{\cP\vz-f}_G^p + O(\tsfrac{\eps}{2^{O(p)}} p) \normof{f}_p^p \\
    &\leq \normof{\cP\vz-f}_G^p + O(\eps) \normof{f}_p^p
\end{align*}
and similarly the lower bound is
\begin{align*}
    \normof{\mA_G\vz-\vb_G}_p^p
    &\geq (\normof{\cP\vz-f}_G - O(\tsfrac{\eps}{2^{O(p)}})\normof{f}_p)^p \\
    &= \normof{f}_p^p (\tsfrac{\normof{\cP\vz-f}_G}{\normof{f}_p} - O(\tsfrac{\eps}{2^{O(p)}}))^p \\
    &\geq \normof{f}_p^p (\tsfrac{\normof{\cP\vz-f}_G^p}{\normof{f}_p^p} - O(\tsfrac{\eps}{2^{O(p)}} p)) \tag{\(\tsfrac{\normof{\cP\vz-f}_G}{\normof{f}_p} + O(\tsfrac{\eps}{2^{O(p)}}) \leq 1\)} \\
    &= \normof{\cP\vz-f}_G^p - O(\tsfrac{\eps}{2^{O(p)}} p) \normof{f}_p^p \\
    &\geq \normof{\cP\vz-f}_G^p - O(\eps) \normof{f}_p^p
\end{align*}
Which completes the first case.
For the second case, where \(\normof{\cP\vz-f}_G\geq \frac12\normof{f}_p\) so that \(\frac{\normof{f}_p}{\normof{\cP\vz-f}_G} \leq 2\), we use the bound \((1\pm u)^p \in 1 \pm p(2e)^{p/2}u\) for \(u\in[0,1]\), as proven later in \lemmaref{uniform-sampling-misc-bounds}:
\begin{align*}
    \normof{\mA_G\vz-\vb_G}_p^p
    &\in (\normof{\cP\vz-f}_G \pm O(\tsfrac{\eps}{2^{O(p)}})\normof{f}_p)^p \\
    &= \normof{\cP\vz-f}_G^p (1 + O(\tsfrac{\eps}{2^{O(p)}}) \tsfrac{\normof{f}_p}{\normof{\cP\vz-f}_G})^p \\
    &\in \normof{\cP\vz-f}_G^p (1 \pm O(\tsfrac{\eps}{2^{O(p)}} 2^{O(p)}) \tsfrac{\normof{f}_p}{\normof{\cP\vz-f}_G}) \\
    &= \normof{\cP\vz-f}_p \pm O(\eps)\normof{f}_p \normof{\cP\vz-f}_G^{p-1} \\
    &\in \normof{\cP\vz-f}_p \pm O(\eps)\normof{f}_p^p \tag{\(\normof{\cP\vz-f}_G \leq C\normof{f}_p\)}
\end{align*}
Which concludes the case analysis, and we find that all \vz with \(\normof{\cP\vz-f}_p \leq 24 OPT\) have
\[
    \normof{\cP\vz-f}_G^p-O(\eps)\normof{f}_p^p
    \le \normof{r(\mA_G\bz-\vb_G)}_p^p
    \le \normof{\cP\vz-f}_G^p +O(\eps)\normof{f}_p^p
\]

\paragraph{Finishing the argument.}
Recall that the interval $[-1,1]$ is partitioned into two groups $G$ and $B$ and that we analyze the samples $s_i$ with $i\in[n_0]$ depending on whether $s_i\in G$ or $s_i\in B$.
Moreover, recall that by Equation~\ref{eqn:se:add}, we have
\[
    \normof{\mA\vz-\vb}_p^p = \normof{\mA_G\vz-\vb_G}_p^p + \normof{\mA_B\vz-\vb_B}_p^p
\]
where $\mA_G$ and $\vb_G$ contain the points in $G$ while $\mA_B$ and $\vb_B$ contain the points in $B$. 
For any $\vz$ with $\normof{\cP\vz-f}_p \le 24OPT$ and $u\in B$, we have by Equation~\ref{eqn:se:bad},
\[
    (1-\eps)|f(u)|^p \le |\cP\vz(u)-f(u)|^p \le (1+\eps)|f(u)|^p.
\]
Since this loss is independent of the value of \vz, we can view \(\sum_{i:s_i\in B} \abs{f(s_i)}^p\) effectively as the sample error of any \vz on the bad set.
Since \(\E[\sum_{i:s_i\in B} \abs{f(s_i)}^p] = \sum_{i=1}^{n_0} \E[\mathbbm{1}_{[s_i\in B]} \abs{f(s_i)}^p] = \frac{n_0}{2} \normof{f}_B^p\), we get \(\sum_{i:s_i\in B} \abs{f(s_i)}^p \leq 50n_0 \normof{f}_B^p\)with probability \(\frac{99}{100}\) by Markov's Inequality.
Recalling that \(r^p = \frac1{n_0}\), we get
\begin{align*}
    \normof{r(\mA_B\vz-\vb_B)}_p^p
    &= \sum_{i: s_i \in B} \frac1{n_0} \abs{\cP\vz(s_i)-f(s_i)}^p \\
    &\in \sum_{i: s_i \in B} \frac1{n_0} \left(\abs{f(s_i)}^p \pm \eps\abs{f(s_i)}^p\right) \\
    &= \sum_{i: s_i \in B} \frac1{n_0} \abs{f(s_i)}^p \pm \eps \left(\frac{1}{n_0}\sum_{i: s_i \in B}\abs{f(s_i)}^p\right) \\
    &\subseteq \normof{r\vb_B}_p^p \pm O(\eps)\normof{f}_B^p \\
    &\subseteq \normof{r\vb_B}_p^p \pm O(\eps)\normof{f}_p^p
\end{align*}
Next recall that for any $\vz$ with $\normof{\cP\vz-f}_p\le 24OPT$, we have
\begin{align*}
    \normof{\mA_G\vz-\vb_G}_p^p &\geq \normof{\cP\vz-f}_G^p - O(\eps)\normof{f}_p^p \\
    \normof{\mA_G\vz-\vb_G}_p^p &\leq \normof{\cP\vz-f}_G^p + O(\eps)\normof{f}_p^p
\end{align*}
Thus, the minimizer $\hat\vx$ to $\min_{\vx\in\bbR^{d+1}}\normof{\mA\vx-\vb}_p$ and $\vx^*=\argmin_{\vx\in\mathbb{R}^{d+1}}\normof{\cP\vx-f}_p$ must satisfy
\begin{align*}
    \normof{\cP\hat\vx-f}_G^p
    &\leq \normof{r(\mA_G\hat\vx - \vb_G)}_p^p + O(\eps) \normof{f}_p^p \\
    &= \normof{r(\mA\hat\vx - \vb)}_p^p - \normof{r(\mA_B\hat\vx - \vb_B)}_p^p + O(\eps) \normof{f}_p^p \\
    &\leq \normof{r(\mA\hat\vx - \vb)}_p^p - \normof{r\vb_B}_p^p + O(\eps) \normof{f}_p^p \\
    &\leq \normof{r(\mA\vx^* - \vb)}_p^p - \normof{r\vb_B}_p^p + O(\eps) \normof{f}_p^p \\
    &\leq \normof{r(\mA\vx^* - \vb)}_p^p - \normof{r(\mA_B\vx^* - \vb_B)}_p^p + O(\eps) \normof{f}_p^p \\
    &= \normof{r(\mA_G\vx^* - \vb_G)}_p^p + O(\eps) \normof{f}_p^p \\
    &\leq \normof{\cP\vx^*-f}_G^p + O(\eps)\normof{f}_p^p
\end{align*}
Since $\normof{f}_p=O(OPT)$, it follows that
\begin{align*}
    \normof{\cP\hat\vx-f}_G^p \le \normof{\cP\vx^*-f}_G^p + O(\eps C^p) OPT^p
\end{align*}
Finally, since for \(u\in B\) we have both
\begin{align*}
(1-\eps)|f(u)|^p &\le |\cP\hat\bx(u)-f(u)|^p \le (1+\eps)|f(u)|^p\\
(1-\eps)|f(u)|^p &\le |\cP\bx^*(u)-f(u)|^p \le (1+\eps)|f(u)|^p
\end{align*}
by Equation~\ref{eqn:se:bad}, it then follows that
\[
    \int_{t\in B}|\cP\hat\vx(t)-f(t)|^p\,dt\le\int_{t\in B}|\cP\vx^*(t)-f(t)|^p\,dt+O(\eps C^p) OPT^p
\]
Therefore, we have
\[
    \int_{-1}^1 |\cP\hat\vx(t)-f(t)|^p \, dt
    \le O(\eps C^p) OPT^p
\]
The claim then follows from rescaling $\eps$ to $\tsfrac{\eps}{C^p}$.  
\end{proof}

\begin{lemma}
\label{lem:uniform-sampling-misc-bounds}
Fix \(u\geq 0\), \(\eps\geq 0\), and even integer \(p \geq 1\).
If \(u+\eps \leq 1\), then \((u+\eps)^p \leq u^p + 2\eps p\).
If \(u\in[0,1]\), then \((1\pm u)\in 1\pm p(2e)^{p/2} u\).
\end{lemma}
\begin{proof}
Since $u+\eps\le 1$ and $u^p+2\eps p\ge 1$ for $p\ge\frac{1}{\eps}$, then we have that 
\[(u+\eps)^p \le 1\le u^p+2\eps p,\]
for all $p\ge\frac{1}{\eps}$ and thus it remains to consider the case where $p<\frac{1}{\eps}$. 

To that end, note that by the binomial expansion, we have
\begin{align*}
(u+\eps)^p&=u^p\left(1+\frac{\eps}{u}\right)^p\\
&=u^p\left(1+\binom{p}{1}\frac{\eps}{u}+\binom{p}{2}\left(\frac{\eps}{u}\right)^2+\binom{p}{3}\left(\frac{\eps}{u}\right)^3+\ldots+\frac{\eps^p}{u^p}\right)\\
&\le u^p\left(1+\frac{\eps p}{u}+\frac{\eps^2p^2}{2!u^2}+\frac{\eps^3p^3}{3!u^3}+\ldots+\frac{\eps^pp^p}{p!u^p}\right).
\end{align*}
For $p<\frac{1}{\eps}$, we thus have
\begin{align*}
(u+\eps)^p&\le u^p\left(1+\frac{\eps p}{u}+\frac{\eps^2p^2}{2!u^2}+\frac{\eps^3p^3}{3!u^3}+\ldots+\frac{\eps^pp^p}{p!u^p}\right)\\
&<u^p\left(1+\frac{\eps p}{u}+\frac{\eps p}{2!u^2}+\frac{\eps p}{3!u^3}+\ldots+\frac{\eps p}{p!u^p}\right)\\
&<u^p+u^p\left(\frac{\eps p}{u^p}+\frac{\eps p}{2!u^p}+\frac{\eps p}{3!u^p}+\ldots+\frac{\eps p}{p!u^p}\right)\\
&\le u^p+\eps p\left(1+\frac{1}{2!}+\frac{1}{3!}+\ldots+\frac{1}{p!}\right)\\
&<u^p+\eps p\left(1+\frac{1}{2!}+\frac{1}{3!}+\ldots\right)=u^p+\eps p(e-1)<u^p+2\eps p,
\end{align*}
as desired.

For the second claim, consider any \(u\in[0,1]\).
\begin{align*}
    (1+u)^p
    &= 1 + \sum_{k=1}^p \binom{p}{k} u^k \\
    &\leq 1 + \sum_{k-1}^p \binom{p}{p/2} u^k \tag{\(\binom{p}{k} \leq \binom{p}{p/2}\)} \\
    &\leq 1 + p (\tsfrac{pe}{p/2})^{p/2} u \tag{\(\binom{p}{k} \leq (\frac{pe}{k})^k\)} \\
    &= 1 + p(pe)^{p/2} u \\
    (1-u)^p
    &= 1 + \sum_{k=1}^p \binom{p}{k} (-u)^k \\
    &\geq 1 - \sum_{k=1}^p \binom{p}{p/2} u^k \\
    &\geq 1 - p(pe){p/2} u
\end{align*}
\end{proof}

\subsection{Tighter Bounds for \texorpdfstring{$L_p$}{Lp} Sensitivities}
\label{app:lp:sens:bounds}
\noindent

In this section, we show tighter bounds for the $L_p$ sensitivities.
While we do not use this result in the paper, we find that it may be useful for future research on polynomial regression.
We first require the following results from polynomial approximation theory.

When $t$ is not near the boundaries of the interval $[-1,1]$, we have a sharper upper bound on the magnitude of the derivative when compared to the Markov brothers' inequality. 
\begin{theorem}[Bernstein's inequality, e.g., Theorem 2.8 in \cite{GovilM99}]
\label{thm:bernstein:poly}
Suppose $q(t)$ is a polynomial of degree at most $d$ such that $|q(t)|\le 1$ for $t\in[-1,1]$. 
Then for all $t\in[-1,1]$, $|q'(t)|\le\frac{d}{\sqrt{1-t^2}}$. 
\end{theorem}

\begin{theorem}[Polynomial approximation of the inverse exponential function, e.g., Theorem 4.1 in \cite{SachdevaV14}]
\label{thm:poly:approx:gauss}
For every $c>0$ and $\eps\in(0,1]$, there exists a polynomial $q_{c,\eps}$ with degree $O\left(\sqrt{\max\left(c,\log\frac{1}{\eps}\right)\cdot\log\frac{1}{\eps}}\right)$ such that
\[
    \max_{x\in[0,c]} |e^{-x}-q_{c,\eps}(x)| \le \eps
\] 
\end{theorem}

\begin{corollary}[Polynomial approximation of the Gaussian kernel]
\label{corol:poly:approx:gauss}
There exists a polynomial $q$ with degree $O(d \log (pd))$ such that \(\abs{q(x)}\leq1\) for all \(x\in[-2,2]\) and
\[
    \max_{x\in[-2,2]}|e^{-20d^2\log(d)x^2}-q(x)|\le\frac{1}{pd^4}
\]
\end{corollary}

\begin{proof}
By \theoremref{poly:approx:gauss}, there exists a polynomial $\hat q$ of degree $O(d\log (pd))$ such that
\[
    \max_{x\in[0, 80d^2\log(d)]} \abs{e^{-x}-\hat q(x)} \leq \frac{1}{2pd^4}
\]
By taking the polynomial $\tilde q(x)=\hat q(20 d^2 \log(d) x^2)$, we find a polynomial $q$ with degree $O(d \log d)$ such that
\[
    \max_{x\in[-2,2]} \abs{e^{-20d^2\log(d)x^2}-q(x)} \leq \frac{1}{2pd^4}
\]
Then, since \(0 \leq e^{-4d^2\log(d)x^2} \leq 1\), it suffices to take \(q(t) = \tilde q(t) - \frac{1}{2pd^4}\) to ensure \(\abs{q(x)} \leq 1\) for all \(x\in[-2,2]\).
\end{proof}
We now prove an upper bound on the $L_p$ sensitivities that is linear in $d$ in the interior of the interval $[-1,1]$, which crucially improves upon known quadratic bounds, e.g., \lemmaref{uniform-sensitivity-bound}. 
\begin{theorem}[Upper bound on sensitivity]
\label{thm:sens:bound}
Let $p\ge 1$ be a fixed constant and let $q$ be a polynomial of degree at most $d\ge 12$. 
Then for $t\in[-1,1]$, the $L_p$ sensitivity of $t$ satisfies
\[ \psi_p[\cP](t) \defeq \max_{\deg(q)\leq d} \frac{|q(t)|^p}{\int_{-1}^1|q(x)|^p\,dx}\le d^2(p+1).\]
Moreover for $\abs{t} \leq 1-\frac1d$, the $L_p$ sensitivity of $t$ satisfies
\[
    \psi_p[\cP](t)
    \defeq \max_{\deg(q)\leq d} \frac{|q(t)|^p}{\int_{-1}^1|q(x)|^p\,dx}
    =O\left(\frac{dp\log (dp)}{\sqrt{1-t^2}}\right)
\]
\end{theorem}
\begin{proof}
The first bound is just \lemmaref{uniform-sensitivity-bound} restated, which directly relied on the Markov brothers' bound.
So here we show that for $|t|\le 1-\frac1d$, we have
\[
    \frac{|q(t)|^p}{\int_{-1}^1|q(x)|^p\,dx}
    = O\left(\frac{dp\log (dp)}{\sqrt{1-t^2}}\right)
\]

For this sharper bound on the sensitivity in the middle of the interval $[-1,1]$, we need a more sophisticated argument.

Let $q\defeq \argmax_{\normof{q}_\infty=1} \frac{\abs{q(t)}^p}{\normof{q}_p^p}$ maximize the sensitivity at $t$.

We would ideally like to use Bernstein's Inequality (\theoremref{bernstein:poly}) to lower bound the mass of \(q\) around \(t\), much like the proof of \lemmaref{uniform-sensitivity-bound}.
Indeed, if \(q\) were maximized at \(t\), then such a proof would be as simple as \lemmaref{uniform-sensitivity-bound}.
That proof picks any \(t^*\) such that \(\abs{q(t^*)}=1\) and lower bounds \(\normof{q}_p^p\) by integrating over an interval of length \(\frac{1}{d^2}\) around \(t^*\).
Crucially, this is tight because even if \(t^*\) is far from \(t\), the Markov Brothers' bound on \(\abs{q'(x)}\) is independent of \(x\).
However, Bernstein's bound \(\abs{q'(x)} \leq \frac{d}{\sqrt{1-x^2}}\) would give very different results depending on how far \(t^*\) is from \(t\).
So, the weight of this proof is in showing that \(q\) must be maximized near \(t\).

We first show that \(q(t)\) is not terribly small.
Since \(\normof{q}_\infty=1\), by the proof of \lemmaref{uniform-sensitivity-bound}, we know that \(\normof{q}_p^p \geq \frac{1}{d^2(p+1)}\).
Let \(c(t) \defeq 1\) be the constant function.
Since \(q\) maximizes the sensitivity function, we get have \(\frac{\abs{q(t)}^p}{\normof{q}_p^p} \geq \frac{\abs{c(t)}^p}{\normof{c}_p^p} = \frac12\).
So, \(\abs{q(t)}^p \geq \frac{\normof{q}_p^p}{2} \geq \frac{1}{2 d^2 (p+1)}\).
Without loss of generality \(q(t)>0\), so we just write \(q(t) \geq \frac{1}{d^{2/p}(2p+2)^{1/p}} \geq \frac{1}{4d^2}\) (since \((2p+2)^{1/p}\leq 4\) for \(p\geq1\)).

Next, we multiply \(q\) with a degree \(\tilde O(d)\) polynomial approximation of a Gaussian pdf centered at \(t\), which effectively erases \(q\) outside of a small interval of \(t\).
Intuitively, this negligibly changes the degree of \(q\) but ensures that the maximum is achieved near \(t\).

We first argue that multiplying by an exact Gaussian bump maximizes \(q\) near \(t\).
Let \(a(x) \defeq e^{-4(x-t)^2 d^2 \log(d)}\) be this Gaussian bump.
Let \(C \defeq q(t)\).
By Markov Brothers', \(\abs{q'(x)}\leq d^2\).
So, we can bound the growth of \(q\) around \(t\):
\[
    \abs{q(t+x)} \leq C + d^2 \abs{x} \leq 1 + d^2\abs{x}
\]
Scaling by the Gaussian,
\[
    \abs{a(t+x) q(t+x)} \leq (1 + d^2 \abs{x}) e^{-5x^2 d^2 \log(d)}
\]
For \(\abs{x}>\frac1{2d}\), we have $e^{-20 x^2 d^2 \log(d)} < e^{-5 \log d} = \frac{1}{d^5}$.
So, for \(\abs{x} > \frac1{2d}\),
\[
    \abs{a(t+x) q(t+x)} \leq \frac{1}{d^5} (1 + d^2\abs{x}) \leq \frac{3}{d^3}
\]
And since \(a(t)q(t) = q(t) = C > \frac{1}{4 d^2}\), we guarantee that \(\argmax_{[-1,1]} a(x)q(x) \in [t-\frac1{2d}, t+\frac1{2d}]\) for \(d \geq 12\).

Next, we substitute \(a(x)\) with the polynomial approximation \(b(x)\) guaranteed by \corolref{poly:approx:gauss}.
Namely, we know that \(b(x)\) has degree at most \(\xi d \log(pd)\) for some constant \(\xi>1\), and that \(\abs{b(x)-a(x)} \leq \frac1{d^4}\) for all \(x\in[-1,1]\).
Then, we get that \(f(x) \defeq b(x) q(x)\) is a degree \(d + \xi d \log(pd) \leq 2\xi d \log(pd)\) polynomial with \(f(t) \geq \frac1{4d^2} - \frac1{d^4}\) and \(\abs{f(t+x)} \leq \frac3{d^3}+\frac{1}{d^4}\), so that \(f\) is still maximized in \([t-\frac1{2d},t+\frac1{2d}]\).

Now, we can appeal to Bernstein's bound to control the sensitivity of \(f\).
Let \(r(x) \defeq \frac1{f(t)} f(x)\) be a rescaling of \(f\) so that \(\normof{r}_\infty = \frac1{f(t)}\) and \(r(t)=1\). 
By Bernstein's bound, we have \(\abs{r'(x)} \leq \frac{2\xi d \log (pd)}{f(t)\sqrt{1-x^2}}\), and therefore that \(\abs{r'(t+x)} \leq \frac{4\xi d \log(pd)}{f(t)\sqrt{1-t^2}}\) for \(x\in[0,\frac1{2d}]\) (via the smoothness of the Chebyshev measure -- see \lemmaref{cheby-smooth} in the appendix).
Let \(m_t \defeq \frac{4\xi d \log(pd)}{\sqrt{1-t^2}}\) be this locally accurate bound on the derivative (but without \(f(t)\)), and let \(t^* \defeq \argmax_{[-1,1]} r(x) \in [t-\frac1{2d},t+\frac1{2d}]\), so that we have:
\[
    r(t^*+x) \geq \frac1{f(t)} - \frac{m_t}{f(t)} x \geq 1 - m_t x \geq 0
    \hspace{1cm}
    \forall x\in[0,\tsfrac1{m_t}]
\]
which follows since \(f(t) \leq q(t) \leq 1\).
Then, we get
\[
    \normof{r}_p^p
    \geq \int_0^{1/m_t} (1-m_tx)^p dx
    = \frac1{m_t(p+1)}
\]
So that
\[
    \frac{\abs{f(t)}^p}{\normof{f}_p^p}
    = \frac{\abs{r(t)}^p}{\normof{r}_p^p}
    \leq \frac{1}{\frac1{m_t(p+1)}}
    = \frac{4\xi(p+1) d\log(pd)}{\sqrt{1-t^2}}
\]
Then, since \(\abs{b(x)}\leq 1\), we have \(\abs{q(x)} \geq \abs{f(x)}\), so that \(\normof{q}_p^p \geq \normof{f}_p^p\).
Further, since \(q(t) = \frac{f(t)}{b(t)} \leq \frac{f(t)}{1-\frac{1}{pd^4}}\), we get \(\abs{q(t)}^p \leq (1-\frac{1}{pd^4})^{-p} \abs{f(t)}^p \leq 2\abs{f(t)}^p\).
We conclude:
\[
    \frac{\abs{q(t)}^p}{\normof{q}_p^p}
    \leq 2\frac{\abs{f(t)}^p}{\normof{f}_p^p}
    \leq \frac{8\xi(p+1) d\log(pd)}{\sqrt{1-t^2}}
\]
\end{proof}

\noindent
We also offer the following lower bound on the $L_p$ sensitivities.
\begin{lemma}[Lower bound on sensitivity]
For any $t\in[-1,1]$, $p\ge 1$, and $d\geq \Omega(p)$ we have
\[
    \psi_p[\cP](t)
    = \max_{q:\deg(q)\le d}\frac{|q(t)|^p}{\int_{-1}^1|q(x)|^p\,dx}
    = \Omega\left(\frac{d\sqrt{p}}{\sqrt{\log d}}\right)
\] 
\end{lemma}
\begin{proof}
Let $t\in[-1,1]$. 
By \corolref{poly:approx:gauss}, there exists a polynomial $q$ with degree $d$ such that
\[\max_{x\in[-1,1]}|e^{-(Cd/\sqrt{\log d})^2(x-t)^2}-q(x)|\le\frac{1}{d},\] 
for a fixed constant $C>0$. 
Let $f(x)=e^{-(Cd/\sqrt{\log d})^2(x-t)^2}$ so that $f(t)=1$ and for $p\ge 1$, 
\[
    \int_{-1}^1|f(x)|^p\,dx = \int_{-1}^1 (f(x))^p \,dx<\int_{-\infty}^\infty e^{-p(Cd/\sqrt{\log d})^2x^2}\,dx=\frac{\sqrt{\pi\log d}}{Cd\sqrt{p}}
\]
Hence, $\frac{|f(t)|^p}{\int_{-1}^1|f(x)|^p\,dx}\ge\frac{Cd\sqrt p}{\sqrt{\pi\log d}}$. 
Since $\max_{x\in[-1,1]}|e^{-(Cd/\sqrt{\log d})^2(x-t)^2}-q(x)|\le\frac{1}{d}$, then it follows that
\[|q(t)|\ge\frac{1}{2}\]
and
\[
    \normof{q-f}_p^p \leq \int_{-1}^1 \frac{1}{d^p} dt = \frac 2{d^p}
\]
Therefore,
\begin{align*}
    \normof{q}_p^p
    &\leq \left(\normof{f}_p + \normof{q-f}_p\right)^p \\
    &\leq \left(\left(\frac{\sqrt{\pi \log d}}{Cd \sqrt p}\right)^{1/p} + \frac{2^{1/p}}{d}\right)^p \\
    &\leq \left(\left(1+\frac1{p^{\frac{p-1}{p}}}\right)\left(\frac{\sqrt{\pi \log d}}{Cd \sqrt p}\right)^{1/p}\right)^p \\
    &= \left(1+\frac1{p^{\frac{p-1}{p}}}\right)^p \cdot \frac{\sqrt{\pi \log d}}{Cd \sqrt p} \\
    &\leq 2 \frac{\sqrt{\pi \log d}}{Cd \sqrt p}
\end{align*}
where the third inequality comes showing that \(\frac{2^{1/p}}{d} \leq \frac{1}{p^{^{\frac{p-1}{p}}}} (\frac{\sqrt{\pi \log d}}{Cd \sqrt p})^{1/p}\), which holds when \(d \geq p (2C\sqrt p)^{\frac1{p-1}} = \Omega(p)\).
We therefore conclude that
\[
    \frac{|q(t)|^p}{\int_{-1}^1|q(x)|^p\,dx}
    \geq \frac{\tsfrac12}{2 \frac{\sqrt{\pi \log d}}{Cd \sqrt p}}
    = \frac{d\sqrt p}{4\sqrt{\pi \log(d)}}
\]
\end{proof}

\section{Reweighted Operator \texorpdfstring{\(L_2\)}{L2} Subspace Embedding}
\label{app:reweighted-subspace-embedding}

\begin{theorem}
Suppose \(s_1,\ldots,s_{n_0}\) are drawn uniformly from \([-1,1]\).
Let \(\mA\in\bbR^{n_0 \times (d+1)}\) be the associated Vandermonde matrix, so that \(\mA_{i,j} = s_i^{j-1}\).
Let \(\gamma \defeq \frac2{n_0}\).
Let \(\mW\in\bbR^{n_0 \times n_0}\) be diagonal with \(\mW_{ii} = \gamma w(s_i)\).
Then for \(n_0 = \Omega(d^4\log d)\), we have that with probability at least \(\frac{11}{12}\),
\[
    \frac12 \cP^\top\cW^{1-\frac2p}\cP \preceq \gamma^{-\frac2p}\mA^\top\mW^{1-\frac2p}\mA \preceq 2\cP^\top\cW^{1-\frac2p}\cP,
\]
where $\cW$ is the operator that rescales by the truncated Chebyshev density $w(t)$ and $p\in[1,2]$. 

\end{theorem}

To prove this claim, it will be more convenient to shift where the \(\gamma\) term is located, into the matrix \mA:

\begin{theorem}
\label{thm:lp-subspace-embedding}
Suppose \(s_1,\ldots,s_{n_0}\) are drawn uniformly from \([-1,1]\), and we construct the associated scaled Vandermonde matrix \(\mA\in\bbR^{n_0 \times (d+1)}\), so that \(\mA_{i,j} = (\frac{2}{n_0})^{\frac12} s_i^{j-1}\).
Let \(\mW\in\bbR^{n_0 \times n_0}\) be the diagonal matrix with \(\mW_{i,i} = w(t_i)\).
If \(n_0 = \Omega(d^5\log d)\), then with probability at least \(\frac{11}{12}\), we have
\[
	\frac12 \cP^\top\cW^{1 - \frac2p}\cP
	\preceq
	\mA^\top\mW^{1 - \frac2p}\mA
	\preceq
	2 \cP^\top\cW^{1 - \frac2p}\cP,
\]
where $\cW$ is the operator that rescales by the truncated Chebyshev density $w(t)$ and $p\in[\frac23,2]$.
\end{theorem}
\begin{proof}
We first prove a more general statement by considering a general operator $\cW$ , which we eventually set to be the Lewis weight operator. 
Let $\cW:L_2([-1,1])\to\mathbb{R}$ be any operator so that $\max_{t\in[-1,1], \vx\in\mathbb{R}^{d+1}}\frac{|\cW\cP\vx(t)|^2}{\normof{|\cW\cP\vx|}_2^2}\le\frakS$ for some \(\frakS<\infty\). 
Consider a fixed $\vx\in\mathbb{R}^{d+1}$ and suppose we uniformly sample $n_0=O\left(\frac{\frakS^2}{\eps^2}\log \frac{d}{\eps}\right)$ points from $[-1,1]$. 
For each $i\in[n_0]$, let $X_i$ be the random variable with value $\abs{[\mW\mA\vx](i)} = \frac2{n_0} |\cW\cP\vx(s_i)|^2$. 
Then $\Ex{X_i}=\frac1{n_0}\|\cW\cP\vx\|_2^2$. 
Moreover, since $|\cW\cP\vx(t)|^2\le\frakS\cdot\|\cW\cP\vx\|_2^2$, we get $\abs{X_i - \E[X_i]} \leq \frac{1}{n_0}(2\frakS + 1) \normof{\cW\cP\vx}_2^2$. 
Let $X=\sum_{i\in[n_0]} X_i = \normof{\mW\mA\vx}_2^2$ so that by linearity of expectation, $\Ex{X}=\|\cW\cP\vx\|_2^2$. 
Setting $\gamma=\eps\normof{\cW\cP\vx}^2_2$ in the formulation of Bernstein's concentration inequality in \importedtheoremref{bded-diffs:conc}, we thus have
\[
    \PPr{ |\normof{\mW\mA\vx}_2^2-\normof{\cW\cP\vx}_2^2| \le \eps \normof{\cW\cP\vx}_2^2 } \le \exp(-O(d\log d/\eps))
\]
for any fixed $\vx\in\mathbb{R}^{d+1}$.
That is, we have
\[(1-\eps)\|\cW\cP\vx\|_2^2\le\frac{1}{n_0}\|\bW\bA\vx\|_2^2\le(1+\eps)\|\cW\cP\vx\|_2^2,\]
with probability at least $1-\exp(-O(d\log d/\eps))$. 

\paragraph{$\eps$-net argument over all coefficient vectors.}
We now union bound over an $\eps$-net over all coefficient vectors $\vx\in\mathbb{R}^{d+1}$. 
Let $\cB=\left\{\cW\cP\vx ~\big|~ \normof{\cW\cP\vx}_2^2 \le 1 \right\}$.
By Lemma 2.4 of \cite{bourgain1989approximation}, we construct a net $\cN$ over $\cB$ by greedily adding in points that are within $L_2$ distance $\left(\frac{\eps}{d}\right)$, so that $|\cN|\le\left(\frac{d}{\eps}\right)^{O(d)} = e^{O(d \log d/\eps)}$.
Note that since we have $(1+\eps)$-accuracy for any $\cW\cP\vx\in\cN$ with probability $1-\exp(-O(d\log d/\eps))$ by sampling $n_0=O\left(\frac{\frakS^2}{\eps^2}\log\frac{d}{\eps}\right)$ points from $[-1,1]$, then by a union bound, we have $(1+\eps)$-accuracy for all points in $\cN$ with high probability. 

For any $\cW\cP\bx$ with $\|\cW\cP\bz\|_2=1$, we construct a sequence $\cW\cP\vy_1,\cW\cP\vy_2,\ldots$ such that $\|\cW\cP\bz-\sum_{j=1}^i\cW\cP\vy_j\|_2\le\eps^i$ and such that there exists constants $\gamma_i\le\eps^{i-1}$ with $\frac{1}{\gamma_i}\cW\cP\by_i\in\mathcal{N}$ for all $i$. 
Formally, we let $\cW\cP\by_1$ be the point in the net $\mathcal{N}$ that is closest to $\cW\cP\bx$, so that $\|\cW\cP\bz-\cW\cP\by_1\|_2\le\eps$. 
We can then define the remaining points $\cW\cP\by_i$ inductively:
For a sequence $\cW\cP\by_1,\ldots,\cW\cP\by_{i-1}$ such that $\gamma_i:=\|\cW\cP\bs-\sum_{j=1}^{i-1}\cW\cP\by_j\|_2\le\eps^{i-1}$, observe that $\frac{1}{\gamma_i}\|\cW\cP\bs-\sum_{j=1}^{i-1}\cW\cP\by_j\|_2=1$. 
Thus, there exists a function $\cW\cP\by_i\in\mathcal{N}$ that is within distance $\eps$ of $\cW\cP\bs-\sum_{j=1}^{i-1}\cW\cP\by_j$, which completes the induction. 

Therefore, for the matrix $\bA$ sampled by the algorithm, by triangle inequality,
\[
    |\|\cW\cP\bx\|_2-\|\bW\bA\bx\|_2|\le\sum_{i=1}^\infty|\|\cW\cP\by_i\|_2-\|\bW\bA\by_i\|_2|\le\sum_{i=1}^\infty\eps^i\|\cW\cP\by_i\|_2=O(\eps)\|\cW\cP\bx\|_2
\]
The correctness over all $\bx\in\mathbb{R}^{d+1}$ then follows from a rescaling of $\eps$. 
Hence, we have that with high probability all \(\vx\in\bbR^{d+1}\) enjoy
\[
    (1-\eps)\|\cW\cP\bx\|_2^2\le\|\bW\bA\bx\|_2^2\le(1+\eps)\|\cW\cP\bx\|_2^2
\]
or equivalently
\[
    (1-\eps)\cP^\top\cW\cP \preceq \mA^\top \mW^{1-\frac2p}\mA \preceq (1+\eps)\cP^\top\cW\cP
\]

\paragraph{Finishing the argument for the Chebyshev density.}
Observe that the Chebyshev density satisfies $w(t)\in[d,(d+1)^2]$ for each $t\in[-1,1]$. 
Since $\max_{t\in[-1,1], q:\deg(q)\le d}\frac{|q(t)|^2}{\normof{q}_2^2}\le O(d^2)$, then by substituting the Lewis weight operator $\cW^{\frac{1}{2}-\frac{1}{p}}$ in place of the general operator $\cW$ in the above analysis, we have that 
\[
    \frakS = \max_{t\in[-1,1], \vx\in\mathbb{R}^{d+1}}\frac{|\cW^{\frac{1}{2}-\frac{1}{p}}\cP\vx(t)|^p}{\normof{\cW^{\frac{1}{2}-\frac{1}{p}}\cP\vx}_p^p}
    \leq O(d^3) \max_{t\in[-1,1], \vx\in\bbR^{d+1}} \frac{\abs{\cP\vx(t)}}{\normof{\cP\vx}_2^2}
    \leq O(d^5p)
\]
for the operator $\cW$ that corresponds to the Chebyshev weights.
Hence the claim follows by taking $\eps=O(1)$. 
\end{proof}
\section{From Two-Stage Sampling to One-Stage Sampling}
\label{app:stage-squashing}
\begin{lemma}
Fix parameter \(n_0\) and function \(f:[-1,1]\rightarrow[0,1]\).
Suppose \(s_1,\ldots,s_{n_0}\) are drawn iid uniformly from \([-1,1]\), and we sample biased coins \(c_i \sim B(1,f(s_i))\) for \(i=1,\ldots,n_0\).
Then, the marginal distribution of \(\{s_i | c_i=1\}\) is a distribution with \(B(n_0,\frac12\int_{-1}^1 f(\tau) d\tau)\) i.i.d.~samples with PDF proportional to \(f\).
\end{lemma}
\begin{proof}
For intuition, we can think of the event \(c_i=1\) as indicating that sample \(i\) is accepted.
Then \(\{s_i | c_i=1\}\) is the set of time samples returned by this rejection sampling scheme.
We now formalize this intuition. 

Let \(p_t\) denote the PDF of the \(t\) variables.
We first write simplify two probabilities for a fixed \(i\in[n_0]\).
First we expand the marginal distribution of the coins:
\begin{align*}
	\Pr[c_i=1]
	&= \int_{-1}^1 \Pr[c_i = 1 ~|~ s_i=\tau] p_t(\tau) d\tau \\
	&= \frac12 \int_{-1}^1 f(\tau) d\tau \\
	&= \frac12 \normof f_1
\end{align*}
Since each coin is marginally distributed as a \(B(1,\frac12\normof f_1)\) random variable, and the number of items in the set \(\{t_i | c_i=1\}\) is the sum of the coins, we conclude that \(\abs{\{t_i | c_i=1\}} \sim B(n_0 ,\frac12\normof f_1)\).

Let \(p_{t_i|c_i=1}\) denote the PDF for \(t_i\) when conditioned on \(c_i=1\).
Using Bayes' Theorem for continuous and discrete random variables,
\begin{align*}
	p_{t|c_i=1}(\tau)
	&= \frac{\Pr[c_i=1|t_i=\tau] \,\cdot\, p_t(\tau)}{\Pr[c_i=1]} \\
	&= \frac{f(\tau) \,\cdot\, \frac12}{\frac12\normof f_1}\\
	&= \frac{f(\tau)}{\int_{-1}^1 f(s) ds}
\end{align*}
Thus, each item in \(\{t_i | c_i=1\}\) (which are trivially independent of each-other), is distributed with PDF proportional to \(f\).
\end{proof}

\begin{replemma}{stage-squashing}
Suppose \(n_0\) time samples are drawn uniformly from \([-1,1]\), and each sample is thrown away with probability \(1-\min\{\frac{m}{n_0} \frac{1}{1-s_i^2}, 1\}\).
Let \(n\) denote the number of remaining samples.
Then \(n\) is distributed as \(B(n_0, O(\frac{m}{n_0}))\), and with probability \(\frac{99}{100}\) the resulting samples cannot be distinguished from iid samples from the Chebyshev measure.
\end{replemma}
\begin{proof}
Below this paragraph, we isolate a simple probability theory claim.
Specifically, we apply the below lemma with \(f(t) = \min\{1, \frac{m}{n_0} \frac{1}{\sqrt{1-t^2}}\}\), so that the remaining number of samples are distributed as \(B(n_0, \frac{m}{2n_0})\).
Since the total variation distance between sampling 1 time with respect to \(f\) and with respect to \(v\) is \(O( (\frac{m}{n_0})^2)\), the distance between \(m\) \iid samples from the two distributions is \(O(\frac{m^3}{n_0^2}) = \frac1{\tilde O(d^{7})}\).
So, the difference in success probability of our algorithm and one that samples by \(f\) is at most \(\frac{1}{\tilde O(d^7)} \leq \frac{1}{100}\).
\end{proof}

\section{Compact Rounding}
\label{app:compact-rounding}

\begin{lemma}
\label{lem:compact-rounding-full}
Let \(\mB\in\bbR^{n_0 \times d_B}\) and \(q\in[0,2]\).
Let \(\vv\in\bbR^{n_0}\) with \(\abs{\vv(i)} \leq \frac1\eps (w_q[\mB](i))^{1/q}\).
Let \(\cN_\eps\) be an \(\eps\)-Net over \(\normof{\mB\vy}_q=1\) with \(\log\abs{\cN_\eps} = O(d \log(\frac1\eps))\).
Let \(\ell = \log_{1+\eps}((2d_B)^{1/q})\).
Then, there exists sets of vectors \(\cD_0,\ldots,\cD_\ell \subseteq \bbR^{n_0}\), such that:
For all \(\vu\in\cN_\eps\) we can define \(\vr \defeq \vu-\vv\) and pick \(\vd_0\in\cD_0,\ldots,\vd_\ell\in\cD_\ell\) to create a ``compact rounding'' \(\vr' = \sum_{k=0}^\ell \vd_{k}\) where:
\begin{enumerate}
	\item \(\abs{\vr(i) - \vr'(i)} \leq \eps\max\{\abs{\vu(i)}, \abs{\vv(i)}\}\) for all \(i\in[n_0]\)
	\item \(\abs{\vd_{k}(i)} \leq \frac{2}{\eps} (\frac12(\frac{w_q[\mB](i)}{d_B} + \frac1{n_0}))^{1/q} (1+\eps)^{k+2}\) for all \(i\in[n_0], k\in\{0,\ldots,\ell\}\)
	\item \(\vd_0,\ldots,\vd_\ell\) all have disjoints supports
\end{enumerate}
Further, we have that the sets \(\cD_0,\ldots,\cD_\ell\) are not too large:
\[
	\log\abs{\cD_k} \leq C_r \frac{d_B \log(n_0)}{\eps^{1+q}(1+\eps)^{qk}}
\]
\end{lemma}

To prove \lemmaref{compact-rounding-full}, we need the following structural statement from \cite{MuscoMWY21}, attributed to Corollary 4.7 and Proposition in \cite{bourgain1989approximation} as well as Proposition 3.1 and Remark 3.2 in \cite{ScechtmanZ01}.
\begin{lemma}[Entropy estimates,~\cite{bourgain1989approximation,MuscoMWY21,ScechtmanZ01}]
\label{lem:entropy:est}
Let \(\mB\in\bbR^{n_0 \times d_B}\) with \(n_0 \geq d_B\) and let \(q\in(0,2)\) be a fixed constant.
Let \(\mW\in\bbR^{n_0 \times n_0}\) be the diagonal matrix with \(\mW_{i,i}=\frac{1}{2}\left(\frac{w_q[\mB](i)}{d_B}+\frac{1}{n_0}\right)\).
Let \(\cB_q=\{\mB\vy \, : \, \normof{\mB\vy}_q \leq 1 \}\).
Then for any \(\gamma\in[1,d^{1/q}]\), there exists a net \(\cN_\infty\subset\bbR^{n_0}\) such that for any \(\vu\in\cB_q\), there exists \(\vf\in\cN_\infty\) with \(\normof{\mW^{-1/q}(\vu-\vf)}_\infty \leq \gamma\) and
\[
	\log \abs{\cN_\infty} \le c_q\cdot\frac{d_B \log n_0}{\gamma^q}
\]
where \(c_q\) is a constant that only depends on \(q\).
\end{lemma}
We next define the index sets and state a structural property on the index sets. 
The following proof is almost identical to the structural property on the index sets by \cite{MuscoMWY21}.
\begin{lemma}[Index sets,~\cite{bourgain1989approximation,MuscoMWY21}]
\label{lem:index:sets}
For each \(k\in\{0,\ldots,\ell\}\), let \(\cN_k\) be the net defined by \lemmaref{entropy:est} for \(\gamma=\frac13\eps(1+\eps)^k>1\).
Otherwise if \(\gamma \le 1\), let \(\cN_k = \cN_\eps\).
For each \(\vu\in\cN_\eps\) and \(k\in\{0,\ldots,\ell\}\), let \(\vf_{k,\vu}\in\cN_k\) satisfy
\[
	\normof{\mW^{-1/q}(\vf_{k,\vu}-\vu)}_\infty \leq \tsfrac13\eps(1+\eps)^k
\]
as defined in \lemmaref{entropy:est}. 
Define the index sets
\begin{align*}
	C_{k,\vu} &\defeq \left\{ i\in[n] ~\Big|~ W_{ii}^{-1/q}\abs{\vf_{k,\vu}(i)} \geq (1+\eps)^{k-1} \right\} \\
	D_{k,\vu} &\defeq C_{k,\vu} \setminus \bigcup_{k' > k} C_{k',\vu} \\
	D_{0,\vu} &\defeq [n_0] \setminus \bigcup_{k \geq 1} C_{k,\vu}
\end{align*}
Then for each \(k\), we have \(\log \abs{\cN_k} = O\left(\frac{d\log n_0}{(\eps(1+\eps)^k)^{q}}\right)\) and for every \(i \in D_{k,\vu}\), we have
\[
	W_{ii}^{1/q}(1+\eps)^{k-2}
	\leq
	\abs{\vu(i)}
	\leq
	W_{ii}^{1/q}(1+\eps)^{k+2}
\]
\end{lemma}
\begin{proof}
First note that the largest \(\gamma\) value we use is \(\frac13\eps(1+\eps)^\ell = \frac\eps3 (2d)^{1/q} \leq d^{1/q}\), so we can safely create all of these \(\cN_k\) sets.
Because \(\normof{\mW^{-1/q}(\vf_{k,\vu}-\vu)}_\infty \leq \frac13 \eps(1+\eps)^k\), we get that all \(i\in C_{k,\vu}\) have
\begin{align*}
	\abs{\vu(i)}
	&\geq \abs{\vf_{k,\vu}(i)} - \tsfrac13 W_{ii}^{1/q}\eps(1+\eps)^k \\
	&\geq W_{ii}^{1/q} \left((1+\eps)^{k-1}-\tsfrac13\eps(1+\eps)^k \right) \\
	&\geq W_{ii}^{1/q} (1+\eps)^{k-2}
\end{align*}
where the second inequality follows from the definition that \(\abs{f_{k,\vu}(i)} \geq W_{ii}^{1/q} (1+\eps)^{k-1}\) for \(i\in C_{k,\vu}\).
We similarly have \(\abs{\vu(i)} \leq W_{ii}^{1/q}(1+\eps)^{k+2}\) for \(i \notin C_{k,\vu}\).
Hence for \(i \in D_{k,\vu} = C_{k,\vu} \setminus \bigcup_{k' > k+1} C_{k',\vu}\) for \(k\in[1,\ell)\), we have
\[
	W_{ii}^{1/q}(1+\eps)^{(k-2)}
	\leq
	\abs{\vu(i)}
	\leq
	W_{ii}^{1/q}(1+\eps)^{k+2}
\]
as desired. 
We next show this bound holds for all \(i\in D_{\ell,\vu}\).
The lower bound follows from the same argument as above, but the upper bound needs to be argued separately since there is no \(C_{\ell+1,\vu}\).
We do this by going through the sensitivity bounds on \(\vu\), which is in the range of \mB:
\begin{align*}
	\abs{\vu(i)}
	&\leq (\psi_{q}[\mB](i))^{1/q} \normof{\vu}_q \\
	&\leq (w_q[\mB](i))^{1/q} \tag{Lemma 3.8 from \cite{ClarksonWW19}} \\
	&= (\frac{w_q[\mB](i)}{2d_B})^{1/q} (1+\eps)^{\ell} \tag{\(\ell = \log_{1+\eps}(2d_B)^{1/q}\)} \\
	&\leq W_{ii}^{1/q} (1+\eps)^{\ell+2} 
\end{align*}
Lemma 3.8 from \cite{ClarksonWW19} simply says that for \(q\in[1,2]\), the \(\ell_q\) sensitivities lower bound the \(\ell_q\) Lewis weights.
This then completes the bulk of the proof since \(\frac1{2d_B} w_q[\mB](i) \leq \frac12(\frac{w_q[\mB](i)}{d_B} + \frac1{n_0}) = W_{ii}\).
As a last note, when \(\gamma \leq 1\), or equivalently \(k \leq \log_{1+\eps}\frac3\eps\), we take \(f_{k,\vu} = \vu\) since \(\cN_k = \cN_\eps\), which trivially gives \(\normof{\mW^{-1/q}(\vf_{k,\vu}-\vu)}_\infty \leq \tsfrac13\eps(1+\eps)^k\).
Further \(\log \abs{\cN_k} = \log\abs{\cN_\eps} \leq O(d \log \frac1\eps) \leq O(d \log(n_0)) = O(\frac{d \log(n_0)}{(\eps(1+\eps)^k)^q})\) by Lemma 2.4 of \cite{bourgain1989approximation}, and since \(\frac1\eps \leq n_0\).
\end{proof}
We similarly define index sets on the measurement vector \(\vv\), though we remark that the definition is conceptual and not algorithmic, in the sense that the entries of \(\vv\) do not need to be read.
\begin{lemma}[Index sets for \vv]
\label{lem:index:sets:b}
\cite{MuscoMWY21}
Fix some \(\vu\in\cN_\eps\).
Consider the following sets:
\begin{align*}
	B_{k,\vu}
	&\defeq \left\{
		i \in D_{k,\vu}
		\,:\,
		\abs{\vv(i)} \leq \tsfrac1\eps W_{ii}^{1/q} (1+\eps)^{k+2}
	\right\} \tag{for \(k\in\{0,\ldots,\ell\}\)} \\
	H_k
	&\defeq \left\{
		i \in [n_0]
		\,:\,
		\tsfrac1\eps W_{ii}^{1/q}(1+\eps)^{k+1}
		<
		\abs{\vv(i)}
		\leq
		\tsfrac1\eps W_{ii}^{1/q}(1+\eps)^{k+2}
	\right\} \tag{for \(k\in\{1,\ldots,\ell\}\)}\\
	G_{k,\vu} &\defeq H_k \setminus \bigcup_{k' \geq k} C_{k',\vu}
	\tag{for \(k\in\{1,\ldots,\ell\}\)}
\end{align*}
Then \(B_{0,\vu} , \ldots , B_{\ell,\vu} , G_{1,\vu} , \ldots , G_{\ell,\vu}\) form a partition of \([n_0]\).
\end{lemma}

\begin{proof}
We first prove that \(B_{0,\vu} , \ldots , B_{\ell,\vu} , G_{1,\vu} , \ldots , G_{\ell,\vu}\) are disjoint.
We then prove that their union is \([n_0]\).

First note that \(D_{0,\vu},\ldots,D_{\ell,\vu}\) are clearly disjoint by their definition.
Further, since \(D_{0,\vu}\) is defined by subtracting away all other \(D_{k,\vu}\) from \([n_0]\), we know that the union of all \(D_{0,\vu},\ldots,D_{\ell,\vu}\) is \([n_0]\).
That is, \(D_{0,\vu},\ldots,D_{\ell,\vu}\) partition \([n_0]\).

Now consider any \(k,k'\).
Then,
\begin{itemize}
	\item
	\(B_{k,\vu} \bigcap B_{k',\vu} = \emptyset\) since \(i\in B_{k,\vu}\) implies \(i\in D_{k,\vu}\) implies \(i \notin D_{k',\vu}\) implies \(i \notin B_{k',\vu}\).
	\item
	\(G_{k,\vu} \bigcap G_{k,\vu} \subseteq H_{k} \bigcap H_{k'} = \emptyset\) since \(H_k\) and \(H_{k'}\) have no intersection by definition.
	\item
	For \(k \geq k'\), \(B_{k,\vu} \bigcap G_{k',\vu} = \emptyset\) since \(i\in B_{k,\vu} \subseteq D_{k,\vu} \subseteq C_{k,\vu}\) implies \(i \in \bigcup_{k'' \geq k'} C_{k',\vu}\), so that \(i\notin G_{k',\vu}\) by definition.
	\item
	For \(k < k'\), \(B_{k,\vu} \bigcap G_{k',\vu} = \emptyset\) since \(k' \geq k+1\) and \(i \in G_{k',\vu} \subseteq H_{k'}\) means \(\abs{\vx(i)} > \frac1\eps W_{ii}^{1/q}(1+\eps)^{k'+1} \geq \frac1\eps W_{ii}^{1/q}(1+\eps)^{k+2}\), which contradicts \(i \in B_{k,\vu} \subseteq D_{k,\vu}\).
\end{itemize}
So, \(B_{0,\vu} , \ldots , B_{\ell,\vu} , G_{1,\vu} , \ldots , G_{\ell,\vu}\) are disjoint.

Now, consider any \(i\in[n_0]\).
Then there exists some \(k\) such that \(i\in D_{k,\vu}\).
If \(\abs{\vv(i)} \leq \frac1\eps W_{ii}^{1/q} (1+\eps)^{k+1}\), then we immediately get that \(i\in B_{k,\vu}\).
Otherwise, if \(\abs{\vv(i)} > \frac1\eps W_{ii}^{1/q} (1+\eps)^{k+1}\), then there exists some \(k' > k\) such that \(i \in H_{k'}\).
Notably, this uses the fact that \(H_{\ell}\) can contain the largest entries of \(\abs{\vv(i)}\), since \(\abs{\vv(i)} \leq \frac1\eps (w_q[\mB](i))^{1/q} \leq \frac1\eps W_{ii}^{1/q}(1+\eps)^{\ell+2}\).
Since \(i\in D_{k,\vu} = C_{k,\vu} \setminus \bigcup_{k''>k} C_{k'',\vu}\), we know that \(i\notin C_{k'',\vu}\) for any \(k'' > k\).
Therefore, we know that \(i\in G_{k',\vu}\), since both \(i\in H_{k'}\) and \(i\in \bigcup_{k''\geq k'} C_{k'',\vu}\).
Therefore, all \(i\in[n_0]\) belongs to exactly one of \(B_{0,\vu} , \ldots , B_{\ell,\vu} , G_{1,\vu} , \ldots , G_{\ell,\vu}\).
In other words, \(B_{0,\vu} , \ldots , B_{\ell,\vu} , G_{1,\vu} , \ldots , G_{\ell,\vu}\) partitions \([n_0]\).
\end{proof}

\begin{lemma}[\(\ell_\infty\) error bound]
\cite{MuscoMWY21}
\label{lem:ell:inf:error:bound}
Fix some \(\vu\in\cN_\eps\).
Then we let \(\vr' = \ve + \sum_{k=0}^\ell \vd_k\) with \ve and \(\vd_k\) as follows:
\begin{align*}
	\vd_0(i) &\defeq \vu(i) - \vv(i) \qquad &i\in B_{0,\vu} \\
	\vd_k(i) &\defeq \vf_{k,\vu}(i) - \vv(i)\qquad &k\in[\ell], i\in B_{k,\vu} \\
	\vd_k(i) &\defeq (1+\eps)^k\cdot W_{ii}^{1/q} - \vv(i)\qquad &k\in[\ell], i\in B_{k,\vu} \\
	\vd_k(i) &\defeq -\vv(i)\qquad &i\in G_{k,\vu} \\
	\vd_k(i) &\defeq 0\qquad &\text{otherwise}
\end{align*}
Then \(\abs{\vr(i)-\vr'(i)} \leq \eps \max\{|\vu(i)|,|\vv(i)|\}\) for all \(i\in[n_0]\).
\end{lemma}
\begin{proof}
Fix any \(i \in [n_0]\).
Since \(B_{0,\vu},\ldots,B_{\ell,\vu},G_{1,\vu},\ldots,G_{\ell,\vu}\) partition \([n_0]\), it suffices to show that if \(i \in B_{0,\vu}\) or \(i \in B_{k,\ell}\) or \(i \in G_{k\ell}\) then \(\abs{\vr(i) - \vr'(i)} \leq \eps\max\{\abs{\vu(i)}, \abs{\vv(i)}\}\).
That is, this proof proceeds by case analysis over these three possible cases.
First, if \(i\in B_{0,\vu}\) then \(\vr'(i) = \vu(i) - \vv(i)\), so that \(\abs{\vr(i) - \vr'(i)} = 0\).
Second, if \(i\in B_{k,\vu}\) for \(k \geq 1\), then \(\vr'(i) = \vf_{k,\vu}(i) - \vv(i)\), and since \(i\in D_{k,\vu}\) we get
\[
	\abs{\vr(i) - \vr'(i)} = \abs{\vf_{k,\vu}(i) - \vu(i)} \leq \eps \tsfrac{(1+\eps)^2}{3} \cdot W_{ii}^{1/q} (1+\eps)^{k-2} \leq \eps \cdot \abs{\vu(i)}
\]
Third, if \(i \in G_{k, \vu}\) then \(\vr'(i) = -\vv(i)\), so that \(\abs{\vr(i) - \vr'(i)} = \abs{\vu(i)}\).
We have \(\abs{\vv(i)} \geq \frac1\eps W_{ii}^{1/q} (1+\eps)^{k+1}\), and since \(i \in G_{k,\vu}\) implies \(i \notin C_{k',\vu}\) for all \(k' \geq k\), we also have \(\abs{\vu(i)} \leq W_{ii}^{1/q} (1+\eps)^{k+2}\)
So,
\[
	\abs{\vr(i) - \vr'(i)} = \abs{\vu(i)} \leq W_{ii}^{1/q}(1+\eps)^{k+2} \leq \eps \abs{\vv(i)}
\]
\end{proof}

\noindent
We now prove the compact rounding of \lemmaref{compact-rounding-full}:
\begin{replemma}{compact-rounding-full}
Let \(\mB\in\bbR^{n_0 \times d_B}\) and \(q\in[0,2]\).
Let \(\vv\in\bbR^{n_0}\) with \(\abs{\vv(i)} \leq \frac1\eps (w_q[\mB](i))^{1/q}\).
Let \(\cN_\eps\) be an \(\eps\)-Net over \(\normof{\mB\vy}_q=1\) with \(\log\abs{\cN_\eps} = O(d \log(\frac1\eps))\).
Let \(\ell = \log_{1+\eps}((2d_B)^{1/q})\).
Then, there exists sets of vectors \(\cD_0,\ldots,\cD_\ell \subseteq \bbR^{n_0}\), such that:
For all \(\vu\in\cN_\eps\) we can define \(\vr \defeq \vu-\vv\) and pick \(\vd_0\in\cD_0,\ldots,\vd_\ell\in\cD_\ell\) to create a ``compact rounding'' \(\vr' = \sum_{k=0}^\ell \vd_{k}\) where:
\begin{enumerate}
	\item \(\abs{\vr(i) - \vr'(i)} \leq \eps\max\{\abs{\vu(i)}, \abs{\vv(i)}\}\) for all \(i\in[n_0]\)
	\item \(\abs{\vd_{k}(i)} \leq \frac{2}{\eps} (\frac12(\frac{w_q[\mB](i)}{d_B} + \frac1{n_0}))^{1/q} (1+\eps)^{k+2}\) for all \(i\in[n_0], k\in\{0,\ldots,\ell\}\)
	\item \(\vd_0,\ldots,\vd_\ell\) all have disjoints supports
\end{enumerate}
Further, we have that the sets \(\cD_0,\ldots,\cD_\ell\) are not too large:
\[
	\log\abs{\cD_k} \leq C_r \frac{d_B \log(n_0)}{\eps^{1+q}(1+\eps)^{qk}}
\]
\end{replemma}
\begin{proof}
Fix any \(\vu\in\cN_\eps\).
Observe that the first property follows from \lemmaref{ell:inf:error:bound}. 
Moreover, note that \(\{\vd_0,\ldots,\vd_\ell\}\) are disjoint by \lemmaref{index:sets:b}, since the vectors \(\vd_k\) only have nonzero support on the indices in \(B_{k,\vu}\) and \(G_{k,\vu}\). 
Hence, the third property holds.

Furthermore, note that \(\abs{\vd_k(i)} \leq \abs{\vv(i)} + \max\left( \abs{\vu(i)}, (1+\eps)^k W_{ii}^{1/q}\right)\) for \(i\in B_{k,\vu}\cup G_{k,\vu}\).
In particular, \(\abs{\vv(i)} \leq \frac1\eps W_{ii}^{1/q} (1+\eps)^{k+2}\).
Similarly, we have from \lemmaref{index:sets} that \(\abs{\vu(i)} \leq W_{ii}^{1/q} (1+\eps)^{k+1}\). 
Therefore, \(\abs{\vd_k(i)} \leq \frac{2}{\eps} W_{ii}^{1/q} (1+\eps)^{k+2}\), which gives the second property.

To bound the number of possible vectors \(\vd_k\), note that \(\vd_k\) is a deterministic function of \(\vf_{k,\vu}\), \(B_{k,\vu}\), and \(G_{k,\vu}\).
So, let \(B_k \defeq \{B_{k,\vu} \,:\, \vu \in \cN_\eps\}\) be the set of all possible ``\(B\)'' index sets generated by the net \(\cN_\eps\) at layer \(k\), and similarly let \(G_k \defeq \{G_{k,\vu} \,:\, \vu\in\cN_\eps\}\).
Then, looking across all possible fixings of \(\vu\in\cN_\eps\), each \(\vd_k\) is deterministic in some \(\vf_{k,\vu}\in\cN_k\), in some \(\cS_1 \in B_k\), and in some \(\cS_2 \in G_k\).
So, the number of possible \(\vd_k\) is at most
\[
	\abs{\cD_k}
	= \abs{\{\vd_k \,:\, \vu\in\cN_\eps\}}
	\leq \abs{\{(\vf_{k,\vu}, \cS_1, \cS_2) \,:\, \vf_{k,\vu}\in\cN_k,\, S_1\in B_k,\, S_2\in G_k\}}
	= \abs{\cN_k} \cdot \abs{B_k} \cdot \abs{G_k}
\]
Next, recall that \(B_{k,\vu} \subseteq D_{k,\vu}\), so \(\abs{B_{k,\vu}} \leq \abs{D_{k,\vu}}\).
Further, recall that \(D_{k,\vu} = C_{k,\vu} \setminus \bigcup_{k' > k} C_{k',\vu}\) and that all \(C_{k,\vu}\) are deterministic in some vector \(\vf_{k,\vu}\) from the net \(\cN_{k}\).
So, \(E_k \defeq \{D_{k,\vu} \,:\, \vu\in\cN_\eps\}\) has\footnote{We use \(E_k\) instead of \(D_k\) to avoid confusion with \(\cD_k\).}
\[
	\abs{B_k}
	\leq \abs{E_k}
	= \abs{\{D_{k,\vu} \,:\, \vu\in\cN_\eps\}}
	\leq \abs{\{(\vf_{k,\vu},\ldots,\vf_{\ell,\vu}) \,:\, \vu\in\cN_\eps\}}
	\leq \prod_{k'=k}^\ell \abs{\cN_{k'}}
\]
By \lemmaref{entropy:est}, we have \(\log\abs{\cN_{k'}} \leq c_q \frac{d_B \log n_0}{(3\eps(1+\eps)^{k'})^q}\), so that
\[
	\log(\abs{B_k})
	\leq \sum_{k'=k}^\ell 3^q c_q \frac{d_B \log n_0}{(\eps (1+\eps)^{k'})^q}
	\leq 2 \cdot 3^qc_q \frac{d_B \log n_0}{\eps^{1+q}(1+\eps)^{qk}}
\]
Where the last inequality comes from bounding \(\sum_{k'=k}^\ell \frac{1}{(1+\eps)^{qk}} \leq \sum_{k'=k}^\infty \frac{1}{(1+\eps)^{qk}} = \frac{1}{(1+\eps)^{qk}} \cdot \frac{(1+\eps)^q}{(1+\eps)^q-1} \leq \frac{1}{(1+\eps)^{qk}} \cdot \frac{(1+\eps)^2}{1+\eps-1} \leq \frac{2}{\eps(1+\eps)^{qk}}\).
Similarly, since \(G_{k,\vu} = H_k \setminus \bigcup_{k' \geq k} C_{k',\vu}\) where \(H_k\) is a fixed set independent of \vu, we again get \(\log\abs{G} \leq \prod_{k'=k}^\ell \abs{\cN_{k'}}\), and so we conclude:
\[
	\log\abs{\cD_k}
	\leq \log\abs{\cN_k} + \log\abs{B_k} + \log\abs{G_k}
	\leq 6c_q3^q \frac{d_B \log n_0}{\eps^{1+q}(1+\eps)^{qk}}
\]
which completes the bulk of the proof.
\end{proof}

\section{Smaller Relegated Proofs}

\subsection{Bounding the Generalized Christoffel Function}
\label{app:christoffel-bound}

\begin{lemma}
Let \(f(s)\) be a differentiable concave function on interval \([a-\Delta, a+\Delta]\).
Then, \(\int_{a-\Delta}^{a+\Delta} f(s) ds \leq 2\Delta f(a)\).
\end{lemma}
\begin{proof}
First recall that concave functions have \(f(s) \leq f(a) + f'(a)(s-a)\), so we have
\[
    \int_{a-\Delta}^{a+\Delta} f(s)dx
    \leq \int_{a-\Delta}^{a+\Delta} \left(f(a) - f'(a)(s-a)\right) ds
    = 2\Delta f(a) + f'(a)\int_{a-\Delta}^{a+\Delta} \left((s-a)\right) ds
    = 2\Delta f(a) + 0
\]
\end{proof}

\begin{lemma}
The generalized Christoffel function \(\lambda_d(\alpha,2,t) \defeq \min_{q:\text{deg}(q)\leq d} \frac{\int_{-1}^1 (q(s))^2 \alpha(s) ds}{(q(t))^2}\), where \(\alpha(s)\defeq(1-s^2)^{\frac1p-\frac12}\), has \(\lambda_d(z,2,t) \leq \frac{C}{d-1}(1-t^2)^{\frac1p}\) for some universal constant \(C\), for all \(t\in\cI_{mid}\), for all \(p\in[\frac23,2]\).
\end{lemma}
\begin{proof}
By Theorem 2.1 of \cite{erdelyi1992generalized}, we know that
\[
    \lambda_d(\alpha,2,t) \leq C^{\Gamma + 3} \alpha_M(t)
\]
where \(C\) is a universal constant, \(\Gamma = \frac2p - 1 \leq 1\), and \(\alpha_M(t) \defeq \int_{\abs{s-t}\leq \Delta_M(t)} \alpha(s) ds\) where \(\Delta_M(s) \defeq \max\{\frac{\sqrt{1-s^2}}M, \frac{1}{M^2}\}\) and \(M = 1 + \frac{2(d-1)}{\Gamma+3} = 1 + \frac{p}{p+1}(d-1) \in [\frac{d-1}{3}, d]\).
To bound the integral within \(\alpha_M(t)\), we use the above lemma about concave functions.
Since \(\frac{d^2}{ds^2}\alpha(s) = -2(\frac1p-\frac12)(1-t^2)^{\frac1p - \frac52}((\frac1p - \frac32)t - 1) \leq 0\) for all \(t\in[-1,1]\) for all \(p \geq \frac23\), we know that \(\alpha\) is concave.
Then, since \(\alpha(s)\) is concave on \([-1,1]\), we find that for any \(t\) such that \(\abs{t} + \Delta_M(t) \leq 1\) (for which \(\abs{t}\leq \sqrt{1-\frac4{M^2}}\) suffices), we get
\[
    \int_{t - \Delta_M(t)}^{t+\Delta_M(t)} \alpha(s) ds \leq 2\Delta_M(t)\alpha(t) = 2 \frac{\sqrt{1-t^2}}{M}(1-t^2)^{\frac1p - \frac12} \leq \frac2M (1-t^2)^{\frac1p}
\]
Putting this together,
\[
    \lambda_d(\alpha,2,t)
    \leq C^{\Gamma + 3} \alpha_M(t)
    \leq C^{4} \frac{2}{\frac{d-1}3} (1-t^2)^{\frac1p}
    = \frac{3C^4}{d-1} (1-t^2)^{\frac1p}
\]
\end{proof}

\subsection{Smoothness of the Chebyshev Measure}

\begin{lemma}
\label{lem:cheby-smooth}
Let \(x\in(-1,1)\), and let \(y \in (-1+\Delta,1-\Delta)\) for \(\Delta = \frac{1-x}{2}\).
Then,
\[
    \frac{1}{\sqrt{1-y^2}} \leq \frac{2}{\sqrt{1-x^2}}
\]
In particular, if \(x=1-\frac1{d}\), then we get \(y\in[1-\frac1d,1-\frac1{2d}]\).
\end{lemma}
\begin{proof}
WLOG, we consider \(x > 0\).
Since \(x \mapsto \frac{1}{\sqrt{1-x^2}}\) is monotonically increasing on \(x>0\), we just need to show that \(\frac{1}{\sqrt{1-(x+\Delta)^2}} \leq \frac{2}{\sqrt{1-x^2}}\).
Rearranging that equation, we get
\[
    4\Delta^2 + 8x\Delta + (3x^2-3) \leq 0
\]
Plugging in \(\Delta = \frac{1-x}{2}\) and simplifying, we see the bound holds for all \(x<1\).
\end{proof}

\subsection{Binomial Approximation}

\begin{lemma}
\label{lem:binomial-approx}
Let \(x > 0\), \(p > 2\), and \(x \leq \frac{1}{p}\).
Then,
\[
    1-\tsfrac12 px \leq (1-x)^p
    \hspace{2cm}
    (1+x)^p \leq 1+3px
\]
In other words, \((1\pm x)^p = 1 + \Theta(px)\).
\end{lemma}
\begin{proof}
Let \(f(u) \defeq (1+u)^p\), so that the Taylor Approximation \(\tilde f(u) \defeq f(0) + f'(0)u = 1 + pu\) has the following residual for \(u\in[-x,x]\):
\[
    \abs{f(u) - \tilde f(u)}
    \leq \max_{\xi\in\cI} \frac{f''(\xi)}{2}u^2
    = \frac12 p(p-1)u^2 \max_{\xi\in\cI} (1+\xi)^{p-2}
    \leq \frac12 p^2u^2 \max_{\xi\in\cI} (1+\xi)^{p-2}
\]
Where \(\cI=[0,x]\) if \(u\geq0\) and where \(\cI=[-x,0]\) if \(u<0\).
For \(u < 0\), we have \(\max_{\xi\in[-x,0]} (1+\xi)^{p-2} = 1\), so that \(\abs{f(u) - \tilde f(u)} \leq \frac12 p^2u^2\).
So, \(\abs{(1-x)^p - (1-px)} = \abs{f(-x) + \tilde f(-x)} \leq \frac12 p^2 x^2 \leq \frac12 px \), and so \((1-x)^p \geq 1-px-\frac12px = 1-\frac32px\).

For \(u\geq0\), we need to be more detailed.
We get \(\max_{\xi\in\cI} (1+\xi)^{p-2} = (1+x)^{p-2}\).
Since \(x \leq \frac1p \leq \frac1{p-2}\), we get \(p < 2 + \frac1x\), so that \((1+x)^{p-2} \leq (1+x)^{\frac1x} \leq 3\).
So,
\[
     \abs{(1+x)^p - (1+px)} \leq \frac12 p^2 x^2 (1+x)^{p-2} \leq \frac32 p^2x^2 \leq \frac32 px
\]
And therefore
\[
    (1+x)^p \leq 1+px + \frac32 px \leq 1 + 3px
\]
\end{proof}

\begin{lemma}
\label{lem:max-unif-cheby}
Suppose we sample \(n_0\) times \(s_1,\ldots,s_{n_0}\) uniformly from \([-1,1]\).
Then, \(\max_i v(s_i) \leq \frac{d+1}{\pi} \sqrt{\frac{n_0 \ln(2)}{\ln(\frac1{1-\delta})}} = \Theta(\frac{d\sqrt{n_0}}{\sqrt{\delta}})\) with probability \(1-\delta\).
\end{lemma}
\begin{proof}
\begin{align*}
    \Pr[\max_i \abs{v(s_i)} \leq F]
    &= (\Pr[v(s_1) \leq F])^{n_0} \\
    &= \left(\Pr\left[s_i \in \pm \sqrt{1-\tsfrac{(d+1)^2}{\pi^2F^2}}\right]\right)^{n_0} \tag{Inverse Function of \(v(s)\)} \\
    &= \left(1-\left(\frac{d+1}{\pi F}\right)^2\right)^{n_0/2} \\
    &= \left(1-\frac1x\right)^{x \cdot \frac{n_0}{2x}} \tag{Let \(x \defeq (\frac{\pi F}{d+1})^2\)}\\
    &\geq 0.25^{\frac{n_0}{2x}} \tag{\(x \geq 2\)}\\
    &= 2^{\frac{-n_0}{x}} \tag{\(0.25 = 2^{-2}\)}
\end{align*}
Making this fail with probability \(\delta\), we get
\begin{align*}
    2^{\frac{-n_0}{x}} &\geq 1-\delta \\
    -\frac{n_0}{x} \ln(2) &\geq \ln(1-\delta) \\
    x &\leq - \frac{n_0 \ln(2)}{\ln(1-\delta)} \\
    (\frac{\pi F}{d+1})^2 &\leq \frac{n_0 \ln(2)}{\ln(\frac{1}{1-\delta})} \\
    F &\leq \frac{d+1}{\pi} \sqrt{\frac{n_0 \ln(2)}{\ln(\frac{1}{1-\delta})}}
\end{align*}
\end{proof}

\end{document}